\documentclass{LMCS}

\def\dOi{11(1:7)2015}
\lmcsheading%
{\dOi}
{1--50}
{}
{}
{Oct.~26, 2011}
{Mar.~13, 2015}
{}

\ACMCCS{[{\bf Theory of computation}]: Semantics and
  reasoning---Program constructs / Program semantics; Logic; [{\bf
      Software and its engineering}]: Software notations and
  tools---Formal language definitions---Semantics}

\usepackage{hyperref}
\usepackage[all]{xy}
\usepackage{amsfonts,amssymb}
\usepackage{oldlfont}
\usepackage{tikz}
\usetikzlibrary{shapes}
\theoremstyle{plain}
\newtheorem{maintheorem}[thm]{Main Theorem}
\newtheorem{convention}[thm]{Convention}
\newtheorem{theorem}[thm]{Theorem}
\newtheorem{lemma}[thm]{Lemma}
\newtheorem{proposition}[thm]{Proposition}
\newtheorem{corollary}[thm]{Corollary}
\theoremstyle{definition}
\newtheorem{definition}[thm]{Definition}
\newtheorem{example}[thm]{Example}
\newtheorem{remark}[thm]{Remark}
\newenvironment{proof*}{\proof}{}
\newcommand{\eg}{{\em e.g.}}
\newcommand{\Eg}{{\em E.g.}}

\newcommand{\ie}{{\em i.e.}}
\newcommand{\Ie}{{\em I.e.}}
\newcommand{\wrt}{{\em w.r.t.}}
\newcommand{\st}{{\em s.t.}}
\newcommand{\FIGURE}{Fig.}
\newcommand{\lam}{\lambda}
\newcommand{\lr}{\ensuremath{\mbox{\tt letrec}}\;}
\newcommand{\inn}{\ensuremath{\; \mbox{\tt in}}\;}
\newcommand{\annRed}[1]{\xrightarrow{#1}}
\newcommand{\annRedM}[1]{\xrightarrow{#1,*}}
\newcommand{\ITT}{\mathit{IT}}
\newcommand{\ANSWERS}{{\mathbb{A}}}
\newcommand{\EXPRESSIONS}{{\mathbb{E}}}
\newcommand{\CONTEXTS}{{\mathbb{C}}}
\newcommand{\LLAZYCC}{L_{\LCC}}
\newcommand{\LLCC}{L_{\LCC}}
\newcommand{\LCC}{{\mathit{lcc}}}
\newcommand{\tBot}{{\texttt{Bot}}}
\newcommand{\Ctxt}{C}
\newcommand{\italt}{\mathit{alt}}
\newcommand{\LABELCOMP}{\|}
\mathchardef\mhyphen="2D
\newcommand{\impl}{\Rightarrow}
\newcommand{\vect}[1]{{\overrightarrow{#1}}}
\newcommand{\tbot}{{\tt Bot}}
\newcommand{\leb}{\preccurlyeq}
\newcommand{\geb}{\succcurlyeq}
\newcommand{\simb}{\simeq}
\newcommand{\lec}{\le}
\newcommand{\gec}{\geq}
\newcommand{\simc}{\sim}
\newcommand{\UptoSim}{F_{\LCC,\sim}}
\newcommand{\UptoSimLR}{F_{\LR,\sim}}
\newcommand{\cBot}{\mathit{cBot}}
\newcommand{\redrule}[1]{{\ensuremath{\mathrm{{#1}}}}}
\newcommand{\rlapp}{\redrule{lapp}}
\newcommand{\rlcase}{\redrule{lcase}}
\newcommand{\rlseq}{\redrule{lseq}}
\newcommand{\rcp}{\redrule{cp}}
\newcommand{\rcpin}{\redrule{cp\mhyphen{}in}}
\newcommand{\rcpe}{\redrule{cp\mhyphen{}e}}

\newcommand{\rbeta}{\redrule{beta}}
\newcommand{\rlbeta}{\redrule{lbeta}}
\newcommand{\rnbeta}{\redrule{nbeta}}
\newcommand{\rnbet}{\redrule{nbeta}}
\newcommand{\rlletin}{\redrule{llet\mhyphen{}in}}
\newcommand{\rllet}{\redrule{llet}}
\newcommand{\rllete}{\redrule{llet\mhyphen{}e}}
\newcommand{\rlll}{\redrule{lll}}
\newcommand{\rcase}{\redrule{case}}
\newcommand{\rgcp}{\redrule{gcp}}

\newcommand{\rcasec}{\redrule{case\mhyphen{}c}}
\newcommand{\rcasein}{\redrule{case\mhyphen{}in}}
\newcommand{\rcasee}{\redrule{case\mhyphen{}e}}
\newcommand{\rncase}{\redrule{ncase}}
\newcommand{\rseq}{\redrule{seq}}
\newcommand{\rnseq}{\redrule{nseq}}
\newcommand{\rseqc}{\redrule{seq\mhyphen{}c}}
\newcommand{\rseqin}{\redrule{seq\mhyphen{}in}}

\newcommand{\rseqe}{\redrule{seq\mhyphen{}e}}
\newcommand{\rbetaTr}{\redrule{betaTr}}
\newcommand{\rseqTr}{\redrule{seqTr}}
\newcommand{\rcaseTr}{\redrule{caseTr}}
\newcommand{\DIV}{{\Uparrow}}
\newcommand{\IEXPR}{\ensuremath{\mathcal{E}_I}}
\newcommand{\IECtxts}{\ensuremath{{\mathcal{R}}_{\TREE}}}
\newcommand{\ECtxt}{R}
\newcommand{\OR}{\syxor}
\newcommand{\NL}{\mathit{NL}}
\newcommand{\cand}{\leb_{\mathit{cand}}}
\newcommand{\candc}{(\leb_{\mathit{cand}})^c}
\newcommand{\Fcand}{F_{\mathit{cand}}}
\newcommand{\tauop}{\tau} 
\newcommand{\closed}[1]{(#1)^c}
\newcommand{\open}[1]{(#1)^o}
\newcommand{\Y}{{\boldsymbol{Y}}}
\newcommand{\RRAP}[1]{\xrightarrow{#1}}
\newcommand{\tnil}{{\tt Nil}}
\newcommand{\tif}{{\tt if}}
\newcommand{\tthen}{{\tt then}}
\newcommand{\telse}{{\tt else}}

\newcommand{\ttrue}{{\tt True}}
\newcommand{\tfalse}{{\tt False}}
\newcommand{\tlet}{{\tt let}}
\newcommand{\tletrec}{{\tt letrec}}
\newcommand{\tin}{{\tt in}}
\newcommand{\tcase}{{\tt case}}
\newcommand{\tof}{{\tt of}}
\newcommand{\tletrx}[2]{(\tletrec~#1 ~{\tt in}~#2)}
\newcommand{\tletrxx}[2]{\tletrec~#1 ~{\tt in}~#2}
\newcommand{\ari}{{\mathrm{ar}}}
\newcommand{\tseq}{{\tt seq}}

\newcommand{\iEnv}{{\mathit{Env}}}
\newcommand{\ialts}{{\mathit{alts}}}
\newcommand{\FV}{{\mathit{FV}}}
\newcommand{\dotcup}{\ensuremath{\mathaccent\cdot\cup}}

\newcommand{\syxor}{\mathrel{|}}
\newcommand{\bchainGen}[6]{\{{#1_{#2}=#3_{#4}}\}_{i=#5}^{#6}}

\newcommand{\bchainXInd}[2]{\bchainGen{x}{i}{x}{i-1}{#1}{#2}}
\newcommand{\bchainN}[3]{\bchainGen{#1}{i}{#2}{i}{1}{#3}}
\newcommand{\tCons}{{\tt Cons}}
\newcommand{\maycon}{{\downarrow}}

\newcommand{\chole}{[\cdot]}
\newcommand{\LAMBDAEXPR}{\ensuremath{{\EXPRESSIONS}_{\lambda}}}
\newcommand{\LAMBDACTXT}{\ensuremath{{\CONTEXTS}_{\lambda}}}
\newcommand{\LETRECEXPR}{\ensuremath{{\EXPRESSIONS}_{\cal L}}}
\newcommand{\LETRECCTXT}{\ensuremath{{\CONTEXTS}_{\cal L}}}
\newcommand{\INFCTXT}{\ensuremath{{\CONTEXTS}_{\cal I}}}
\newcommand{\LLR}{L_{\mathit{LR}}} 
\newcommand{\LNAME}{L_{\mathit{name}}}

\newcommand{\LTREE}{\ensuremath{L_{\mathit{tree}}}}

\newcommand{\LR}{{\mathit{LR}}}
\newcommand{\NAME}{\mathit{name}}

\newcommand{\TREE}{\ensuremath{{\mathit{tree}}}}
\newcommand{\NTREE}{\ensuremath{\neg{\mathit{tree}}}}
\newcommand{\transComp}{N}
\newcommand{\transN}{N'}
\begin{document}
%%
%%
%%%%%%%%%%%%%%%%%%%%%%%%%%%%%%%%%%%%%%%%%%%%%%%%%%%%%%%%%%%%%%%%%%%%%%%%%%%%%%
%% title, authors, affiliations
\title[Simulation in the Call-by-Need Lambda Calculus]{Simulation in the Call-by-Need Lambda-Calculus with Letrec, Case, Constructors, and Seq\rsuper*}
%%%%%%%%%%%%%%%%%
\author[M.~Schmidt-Schau{\ss}]{Manfred Schmidt-Schau{\ss}\rsuper a}
\address{{\lsuper{a,b}}Dept. Informatik und Mathematik, Inst. Informatik, J.W. Goethe-University, PoBox 11 19 32, D-60054 Frankfurt, Germany}
\email{\{schauss,sabel\}@ki.informatik.uni-frankfurt.de}
\thanks{{\lsuper a}The first author is supported by the DFG under grant SCHM 986/9-1.}
%%%%%%%%%%%%%%%%%
\author[D.~Sabel]{David Sabel\rsuper b}
\address{\vspace{-18 pt}}
%\email{sabel@ki.informatik.uni-frankfurt.de}  
%%%%%%%%%%%%%%%%%
\author[E.~Machkasova]{Elena Machkasova\rsuper c} 
\address{{\lsuper c}Division of Science and Mathematics, University of Minnesota, Morris, MN 56267-2134, U.S.A} %optional
\email{elenam@morris.umn.edu}

%%%%%%%%%%%%%%%%%%%%%%%%%%%%%%%%%%%%%%%%%%%%%%%%%%%%%%%%%%%%%%%%%%%%%%%%%%%%%%
%% mandatory lists of keywords and classifications:
%%%%%%%%%%%%%%%%%
\titlecomment{{\lsuper*}This paper is an extended version of \cite{schmidt-schauss-sabel-machkasova-rta:10} for more expressive calculi, and also of \cite{schmidt-schauss-copy-rta:07} \wrt~infinite trees, with fully worked out proofs.}
%%%%%%%%%%%%%%%%%%%%%%%%%%%%%%%%%%%%%%%%%%%%%%%%%%%%%%%%%%%%%%%%%%%%%%%%%%%%%%
%% ABSTRACT
\begin{abstract}
This paper shows equivalence of several versions of applicative similarity and
contextual approximation, and hence also of applicative bisimilarity and 
contextual equivalence, in LR, the deterministic call-by-need lambda calculus
with letrec extended by data constructors, case-expressions and Haskell's 
seq-operator. LR models an untyped version of the core language of Haskell. 
The use of bisimilarities simplifies equivalence proofs in calculi and opens a
way for more convenient correctness proofs for program transformations.
 
The proof is by a fully abstract and surjective transfer into a 
call-by-name calculus, which is an extension of Abramsky's lazy lambda calculus.    
In the latter calculus equivalence of our similarities and contextual 
approximation can be shown by Howe's method. Similarity is transferred back to
LR on the basis of an inductively defined similarity.
 
The translation from the call-by-need letrec calculus into the extended 
call-by-name lambda calculus is the composition of two translations. The first
translation replaces the call-by-need strategy by a call-by-name strategy and
its correctness is shown by exploiting infinite trees which emerge by unfolding
the letrec expressions. The second translation encodes letrec-expressions by 
using multi-fixpoint combinators and its correctness is shown syntactically 
by comparing reductions of both calculi.

A further result of this paper is an isomorphism  between the mentioned calculi, 
which is also an identity on letrec-free expressions.
\end{abstract}
%% END ABSTRACT
%%%%%%%%%%%%%%%%%%%%%%%%%%%%%%%%%%%%%%%%%%%%%%%%%%%%%%%%%%%%%%%%%%%%%%%%%%%%%%
%%
\maketitle\vfill
%%
%%%%%%%%%%%%%%%%%%%%%%%%%%%%%%%%%%%%%%%%%%%%%%%%%%%%%%%%%%%%%%%%%%%%%%%%%%%%%%
%% SECTION 1: INTRODUCTION
\section{Introduction}
\subsection*{Motivation}
Non-strict functional programming languages, such as the core-language of 
Haskell \cite{peyton-jones-haskell-98:03}, can be modeled using extended 
call-by-need lambda calculi. 

The operational semantics of such a programming language defines how programs
are evaluated and how the value of a program is obtained. 
Based on the operational semantics, the notion of {\em contextual equivalence} 
(see \eg\ \cite{morris:68,plotkin:75}) is a natural notion of program 
equivalence which follows Leibniz's law to identify the indiscernibles, that is
two programs are equal iff their observable (termination) behavior is 
indistinguishable even if the programs are used as a subprogram of any other 
program (\ie\ if the programs are plugged into any arbitrary {\em context}). 
For pure functional programs it suffices to observe whether or not the 
evaluation of a program terminates with a value (\ie\ whether the program
{\em converges}). 
Contextual equivalence has several advantages: 
Any reasonable notion of program equivalence should be a congruence which
distinguishes obvious different values, \eg\ different constants are 
distinguished, and functions (abstractions) are distinguished from constants. 
Contextual equivalence satisfies these requirements and is usually the 
coarsest of such congruences. 
Another (general) advantage is that once expressions, contexts, an evaluation,
and a set of values are defined in a calculus, its definition of contextual 
equivalence can be derived, and thus this approach can be used for a broad 
class of program calculi.
 
On the other hand, due to the quantification over all program contexts, 
verifying equivalence of two programs \wrt\ contextual equivalence is often a
difficult task.
Nevertheless such proofs are required to ensure the 
{\em correctness of program transformations} where the correctness notion 
means that contextual equivalence is preserved by the transformation.
Correctness of program transformations is indispensable for the correctness of 
compilers, but program transformations also play an important role in several 
other fields, \eg\ in code refactoring to improve the design of programs, or in
software verification to simplify expressions and thus to provide proofs 
or tests.

Bisimulation is another notion of program equivalence which was first invented
in the field of process calculi 
(\eg\ \cite{Milner:80,milner-pi-calc:99,sangiorgi-walker:01}), but has also 
been applied to functional programming and several extended lambda calculi 
(\eg\ \cite{howe:89,abramsky-lazy:90,howe:96}).
Finding adequate notions of bisimilarity is still an active research topic 
(see~\eg\ \cite{koutavas-wand:2006,Sangiorgi-Kobayashi-Sumii:2011}).
Briefly explained, bisimilarity equates two programs $s_1,s_2$ if all 
experiments passed for $s_1$ are also passed by $s_2$ and vice versa. For 
applicative similarity (and also bisimilarity) the experiments are evaluation
and then recursively testing the obtained values: 
Abstractions are applied to all possible arguments, data objects are decomposed
and the components are tested recursively.
Applicative similarity is usually defined co-inductively, \ie\ as a greatest
fixpoint of an operator. 
Applicative similarity allows convenient and automatable proofs of correctness
of program transformations, \eg\ in mechanizing 
proofs~\cite{Dennis-Bundy-Green:1997}.

Abramsky and Ong showed that applicative bisimilarity is the same as contextual
equivalence in a specific simple lazy lambda calculus  
\cite{abramsky-lazy:90,abramsky-ong:93}, and Howe \cite{howe:89,howe:96} proved
that in classes of lambda-calculi applicative bisimulation is the same as 
contextual equivalence. 
This leads to the expectation that some form of applicative bisimilarity may be
used for calculi with  Haskell's cyclic letrec. 
However, Howe's proof technique appears to be not adaptable to lambda calculi
with cyclic let, since there are several deviations from the requirements for
the applicability of Howe's framework.
(i) Howe's technique is for call-by-name calculi and it is not obvious how to
adapt it to call-by-need evaluation. 
(ii) Howe's technique requires that the values (results of reduction) are 
recognizable by their top operator. 
This does not apply to calculi with $\tletrec$, since $\tletrec$-expressions
may be values as well as non-values. 
(iii) Call-by-need calculi with letrec usually require reduction rules to shift
and join \tletrec-bindings. These modifications of the syntactic structure of
expressions do not fit well into the proof structure of Howe's method.

Nevertheless, Howe's method is also applicable to calculi with non-recursive 
let even in the presence of nondeterminism \cite{mannmss:10}, where for the
nondeterministic case applicative bisimilarity is only sound (but not complete)
\wrt\ contextual equivalence.
However, in the case of (cyclic) letrec and {\em nondeterminism} applicative
bisimilarity is unsound \wrt\ contextual equivalence
\cite{schmidt-schauss-sabel-machkasova-IPL:11}.
This raises a question: 
which call-by-need calculi with letrec permit applicative bisimilarity as a 
tool for proving contextual equality?

\subsection*{Our Contribution}
In \cite{schmidt-schauss-sabel-machkasova-rta:10} we have already shown that
for the minimal extension of Abramsky's lazy lambda calculus with letrec which
implements sharing and explicit recursion, the equivalence of contextual
equivalence and applicative bisimilarity indeed holds. 
However, the full (untyped) core language of Haskell has data constructors, 
case-expressions and the seq-operator for strict evaluation. 
Moreover, in ~\cite{schmidt-schauss-machkasova-sabel:rta:2013} it is shown that
the extension of Abramsky's lazy lambda calculus with \tcase{}, constructors, 
and \tseq{} is not conservative, \ie\ it does not preserve contextual 
equivalence of expressions. 
Thus our results obtained in \cite{schmidt-schauss-sabel-machkasova-rta:10} for
the lazy lambda calculus extended by letrec only are not transferable to the
language extended by \tcase{}, constructors, and \tseq.
For this reason we provide a new proof for the untyped core language of Haskell.

As a model of  Haskell's core language we use the call-by-need lambda calculus
$\LLR$ which was introduced and motivated in 
\cite{schmidt-schauss-schuetz-sabel:08}. 
The calculus $\LLR$ extends the lazy lambda calculus with letrec-expressions, 
data constructors, \tcase-expressions for deconstructing the data, and 
Haskell's \tseq-operator for strict evaluation.

We define the operational semantics of $\LLR$ in terms of a small-step 
reduction, which we call normal order reduction.
As it is usual for lazy functional programming languages, evaluation of 
$\LLR$-expressions successfully halts if a {\em weak head normal form} is 
obtained, \ie\ normal order reduction does not reduce inside the body of
abstractions nor inside the arguments of constructor applications.
The $\LLR$ calculus has been studied in detail in 
\cite{schmidt-schauss-schuetz-sabel:08} and correctness of several important
program transformations has been established for it.

Our main result in this paper is that several variants of applicative 
bisimilarities are sound and complete for contextual equivalence in $\LLR$,
\ie\ coincide with contextual equivalence. 
Like context lemmas, an applicative bisimilarity can be used as a proof tool
for showing contextual equivalence of expressions and for proving correctness
of program transformations in the calculus $\LLR$. 
Since we have completeness of our applicative bisimilarities in addition to
soundness, our results can also be used to disprove contextual equivalence of
expressions in $\LLR$.
Additionally, our result shows that the untyped applicative bisimilarity is
sound for a polymorphic variant of $\LLR$, and hence for the typed core 
language of Haskell. 

Having the  proof tool of applicative bisimilarity in $\LLR$ is also very 
helpful for more complex calculi if their pure core can be conservatively
embedded in the full calculus.
An example is our work on Concurrent Haskell 
\cite{sabel-schmidt-schauss-PPDP:2011,sabel-schmidt-schauss:12:LICS},
where our calculus CHF that models Concurrent Haskell has top-level processes
with embedded lazy functional evaluation.
We have shown in the calculus CHF that Haskell's deterministic core language
can be conservatively embedded in the calculus CHF.

\begin{figure}
\begin{tikzpicture}
\node (LR)    at (-3,0)   [] [line width=1pt,shape=ellipse,draw,fill=blue!5!white] {\rule{0mm}{6mm}\rule{6mm}{0mm}};
\node (txt)   at (-3,0)   [] {$\LLR$};
\node (IT)    at (0,-1.5) [] [line width=1pt,dotted,shape=ellipse,draw,fill=blue!5!white] {\rule{0mm}{6mm}\rule{6mm}{0mm}};
\node (txt)   at (0,-1.5) [] {$L_{\TREE}$};
\node (LName) at (3,0)    [] [line width=1pt,shape=ellipse,draw,fill=blue!5!white] {\rule{0mm}{6mm}\rule{6mm}{0mm}};
\node (txt)   at (3,0)    [] {$\LNAME$};
\node (Lcc)   at (6,0)    [] [line width=1pt,shape=ellipse,draw,fill=blue!5!white] {\rule{0mm}{6mm}\rule{6mm}{0mm}};
\node (txt)   at (6,0)    [] {$\LLAZYCC$};
\draw[->,line width=1pt]               (LR)    to node [above] {W} (LName);
\draw[->,line width=1pt]               (LName) to node [above] {N} (Lcc);
\draw[<->,line width=1pt,dotted]       (LR)    to node [] {}       (IT);
\draw[<->,line width=1pt,dotted]       (LName) to node [] {}       (IT);
\draw[->,line width=1pt,bend left =20] (LR)    to node [above] {$N \circ W$} node [] {} (Lcc);
\end{tikzpicture}

\caption{Overall structure. Solid lines are fully-abstract translations, which
are also isomorphisms and identities on letrec-free expressions; dotted lines
are convergence preservation to/from the system $\LTREE$ of infinite trees.
\label{figure:results}}
\end{figure}

We prove the equivalence between the applicative similarities and contextual 
equivalence in $\LLR$, by lifting the equivalence from a \tletrec-free call-by-name calculus $\LLAZYCC$. 
The calculus $\LLAZYCC$ minimally extends Abramsky's lazy calculus by Haskell's
primitives. 
As shown in~\cite{schmidt-schauss-machkasova-sabel:rta:2013}, data constructors
and $\tseq$ are explicitly needed in $\LLAZYCC$.
The structure of the proof, with its intermediate steps, is shown in Figure~\ref{figure:results}.
We prove the equivalence between the applicative similarities and contextual
equivalence in $\LLAZYCC$, by extending Howe's method. 
We bridge $\LLR$ and $\LLAZYCC$ in two steps, using intermediate calculi 
$\LNAME$ and $\LTREE$. $\LNAME$ is the call-by-name variant of $\LLR$, and
$\LLAZYCC$ is obtained from $\LNAME$ by encoding letrec using multi-fixpoint
combinators. 
The calculi $\LLR$ and $\LNAME$ are related to each other via their infinite 
unfoldings, thus we introduce a calculus $\LTREE$ of infinite trees 
(similar infinitary rewriting, see \cite{kennaway-klop:97,schmidt-schauss-copy-rta:07}).
Convergence of expressions in $\LLR$ and $\LNAME$ is shown to be equivalent to
their translation as an infinite tree in the calculus $\LTREE$ (dotted lines in
the picture).
We establish full abstractness of translation $N$ and $W$ between calculi 
$\LLR$, $\LNAME$, and $\LLAZYCC$ with respect to contextual equivalence. 
Correctness of similarity is transferred back from $\LLAZYCC$ to $\LLR$ on the
basis of an inductively defined similarity (for more details see 
\FIGURE~\ref{subsec:sim-lr-def}).
 
A consequence of our result is that the three calculi $\LLR$, $\LNAME$, and
$\LLAZYCC$ are isomorphic, modulo the equivalence 
(see Corollaries~\ref{cor:N-iso} and~\ref{cor:W-iso}), and also that the 
embedding of the calculus $\LLAZYCC$ into the call-by-need calculus $\LLR$ is
an isomorphism of the respective term models. 

\subsection*{Related Work}
In \cite {Gordon:99} Gordon shows that bisimilarity and contextual equivalence
coincide in an extended call-by-name PCF language.
Gordon provides a bisimilarity in terms of a labeled transition system.  
A similar result is obtained in \cite{pitts:97} for PCF extended by product 
types and lazy lists where the proof uses Howe's method 
(\cite{howe:89,howe:96}; see also \cite{mannmss:10,Pitts:2011}), and where the
operational semantics is a big-step one for an extended PCF-language.
The observation of convergence in the definition of contextual equivalence is 
restricted to programs (and contexts) of ground type (\ie\ of type integer or
{\tt Bool}). 
Therefore $\Omega$ and $\lambda x.\Omega$ are equal in the calculi considered
by Gordon and Pitts. 
This does not hold in our setting for two reasons: 
first, we observe termination for functions and thus the empty context already
distinguishes $\Omega$ and $\lambda x.\Omega$, and second, our languages employ 
Haskell's seq-operator which permits to test convergence of any expression and
thus the context $\tseq~[\cdot]~\ttrue$ distinguishes $\Omega$ and 
$\lambda x.\Omega$.

\cite{jeffrey-short:94} presents an investigation into the semantics of a 
lambda-calculus that permits cyclic graphs, where a fully abstract denotational
semantics is described.
However, the calculus is different from our calculi in its expressiveness since
it permits a parallel convergence test, which is required for the full 
abstraction property of the denotational model. 
Expressiveness of programming languages was investigated 
\eg\ in \cite{Felleisen:91} and the usage of syntactic methods was formulated as
a research program there, with non-recursive \tlet{} as the paradigmatic example. 
Our isomorphism-theorem \ref{thm:isomorphism} shows that this approach is 
extensible to a cyclic \tlet{}.

Related work on calculi with recursive bindings includes the following 
foundational papers.
An early paper that proposes cyclic let-bindings (as graphs) 
is \cite{ariola-klop-short:94}, where reduction and confluence properties are
discussed.
\cite{ariola:95,ariola:97} study equational theory for call-by-need lambda
calculus extended with non-recursive \tlet{}, 
which is finer than contextual equivalence, and in \cite{maraistoderskywadler:98}
it is shown that call-by-name and call-by-need
evaluation induce the same observational equivalences for a call-by-need lambda
calculus with non-recursive \tlet{}.
Additionally, the extension of the corresponding calculi by recursive \tlet{} is
discussed in \cite{ariola:95,ariola:97}, and further call-by-need lambda calculi
with a recursive \tlet{} are presented 
in \cite{ariola-blom:97,ariola-blom:02,Nakata-hasegawa:2009:jfp} where 
\cite{Nakata-hasegawa:2009:jfp} study the equivalence between a natural 
semantics and a reductions semantics.
In \cite{ariola-blom:02} it is shown that there exist infinite normal forms and
that the calculus satisfies a form of confluence.
All these calculi correspond to our calculus $\LLR$.
A difference is that the let-shifting in the standard reduction in the mentioned works is different from $\LLR$.
However, this difference is not substantial, since it does not influence the contextual semantics.
A more substantial difference is that $\LLR$ combines 
recursive $\tlet{}$ with data constructors, case-expressions and \tseq, which none of the related works do.

In \cite{moran-sands:99} a call-by-need calculus is analyzed which is closer to
our calculus $\LLR$, since \tletrec, \tcase, and constructors are present (but not
\tseq).
Another difference is that \cite{moran-sands:99} uses an abstract machine semantics
as operational semantics, while their approach to program equivalence is based on
contextual equivalence, as is ours.

The operational semantics of call-by-need lambda calculi with \tletrec{} are 
investigated in \cite{launch:93} and \cite{sestoft:97}, where the former proposed
a natural semantics, and proved it correct and adequate with respect to a 
denotational semantics, and the latter derived an efficient abstract machine 
from the natural semantics.

Investigations of the semantics of lazy functional programming languages including
the $\tseq$-operator can be found in 
\cite{johann-voigtlaender:06,voigtlaender-johann:07}.

\subsection*{Outline}
In Sect.~\ref{sec:common} we introduce some common notions of program calculi,
contextual equivalence, similarity and also of translations between those 
calculi.
In Sect.~\ref{sec:calculi} we introduce the extension $\LLAZYCC$ of Abramsky's
lazy lambda calculus with \tcase, constructors, and \tseq, and two 
letrec-calculi $\LLR$, $\LNAME$ as further syntactic extensions.
In Sect.~\ref{sec:simulation-lazy} we show that for so-called 
``convergence admissible'' calculi an alternative inductive characterization of
similarity is possible. 
We then use Howe's method in $\LLCC$ to show that contextual approximation and
a standard version of applicative similarity coincide. 
Proving that $\LLCC$ is convergence admissible then implies that the 
alternative inductive characterization of similarity can be used for $\LLCC$.
In Sect.~\ref{sec-translation-NEED-NAME} and~\ref{sec:NAME-to-LAZY} the 
translations $W$ and $N$ are introduced and the full-abstraction results are
obtained.
In Sect.~\ref{sec-simulations} we show soundness and completeness of our 
variants of applicative similarity w.r.t. contextual equivalence in $\LLR$.
We conclude in Sect.~\ref{sec:conclusion}.

%%
%%%%%%%%%%%%%%%%%%%%%%%%%%%%%%%%%%%%%%%%%%%%%%%%%%%%%%%%%%%%%%%%%%%%%%%%%%%%%%
%% SECTION 2: COMMON NOTIONS AND NOTATIONS FOR CALCULI
% DS, leb-lec Notation adjusted
\section{Common Notions and Notations for Calculi}\label{sec:common}
Before we explain the specific calculi, some common notions are introduced.
A calculus definition consists of its syntax  together with its operational 
semantics which defines the evaluation of programs and the implied equivalence
of expressions: 

\begin{definition}\label{def:calculus}
An untyped deterministic {\em calculus} $D$ is a four-tuple 
\mbox{$({\EXPRESSIONS},{\CONTEXTS},\to,\ANSWERS)$}, where ${\EXPRESSIONS}$ 
are expressions (up to $\alpha$-equivalence), 
${\CONTEXTS}: {\EXPRESSIONS} \to {\EXPRESSIONS}$ is  a set of functions 
(which usually represents contexts), $\to$ is a small-step reduction relation
(usually the normal-order reduction), which is a partial function on expressions
(\ie, deterministic), and $\ANSWERS \subset {\EXPRESSIONS}$ is a set of 
{\em answers} of the calculus.

For $C \in {\CONTEXTS}$ and an expression $s$, the functional application is
denoted as  $C[s]$. 
For contexts, this is the replacement of the hole of $C$ by $s$. 
We also assume that the identity function $\mathit{Id}$ is contained in 
${\CONTEXTS}$ with $\mathit{Id}[s] = s$ for all expressions $s$, and that 
$\CONTEXTS$ is closed under composition, 
\ie\ $C_1,C_2 \in {\CONTEXTS}\implies C_1 \circ C_2 \in {\CONTEXTS}$. 

The {\em transitive  closure} of $\to$ is denoted as $\xrightarrow{+}$ and the
{\em transitive and reflexive closure} of $\to$ is denoted as $\xrightarrow{*}$. 
The notation $\xrightarrow{0 \vee 1}$ means equality or one reduction, 
and $\xrightarrow{k}$ means $k$ reductions.
Given an expression $s$, a sequence $s \to s_1 \to \ldots \to s_n$ is called a
{\em reduction sequence}; it is called an {\em evaluation} if $s_n$ is an answer,
\ie\ $s_n \in \ANSWERS$; in this case we say  $s$ {\em converges} and denote 
this as $s \maycon_D s_n$ or as $s \maycon_D$ if $s_n$ is not important.  
If there is no $s_n$ \st\ $s \maycon_D s_n$ then $s$ {\em diverges},
denoted as $s \DIV_D$. 
When dealing with multiple calculi, we often use the calculus name to mark its
expressions and relations, \eg\ $\xrightarrow{D}$ denotes a reduction 
relation in ${D}$. 
\end{definition}
We will have to deal with several calculi and preorders. Throughout this paper
we will use the symbol $\leb$ for co-inductively defined preorders 
(\ie\ similarities), and $\lec$ for (inductively defined or otherwise defined)
contextual preorders. 
For the corresponding symmetrizations we use $\simb$ for $\leb \cap \geb$ and
$\simc$ for $\lec \cap \gec$. 
All the symbols are always indexed by the corresponding calculus and sometimes
more restrictions like specific sets of contexts are attached to the indices
of the symbols. 

Contextual approximation and equivalence can be defined in a general way:
\begin{definition}[Contextual Approximation and Equivalence, $\lec_{D}$ and $\simc_D$]
Let $D = ({\EXPRESSIONS},{\CONTEXTS},\to,\ANSWERS)$ be a calculus and 
$s_1,s_2$ be $D$-expressions.
{\em Contextual approximation} (or {\em contextual preorder}) $\lec_{D}$  and 
{\em contextual equivalence} $\simc_{D}$ are defined as: 
\[\begin{array}{lcll}
 s_1 \lec_{D} s_2  & \text{iff}  & ~\forall C \in{\CONTEXTS}: ~ ~   C[s_1]\maycon_D \impl  C[s_2]\maycon_D\\
 s_1 \simc_{D} s_2 &  \text{iff} & ~s_1 \lec_{D} s_2 \wedge s_2 \lec_{D} s_1
 \end{array}
\]

A {\em program transformation} is a binary relation 
$\eta\subseteq ({\EXPRESSIONS}\times {\EXPRESSIONS}$). 
A program transformation $\eta$ is called {\em correct} iff 
$\eta\ \subseteq\ \sim_D$.
\end{definition}
Note that $\lec_{D}$ is a precongruence, \ie, $\lec_{D}$ is reflexive, 
transitive, and $s  \lec_{D} t$ implies $C[s] \lec_{D} C[t]$ for all 
$C \in {\CONTEXTS}$, and that $\simc_{D}$ is a congruence, \ie\ a 
precongruence and an equivalence relation.

We also define a general notion of similarity coinductively for untyped
deterministic calculi.
We first define the operator $F_{D,{\cal Q}}$ on binary relations of 
expressions:
\begin{definition}
Let $D = ({\EXPRESSIONS},{\CONTEXTS},\to,\ANSWERS)$ be an untyped deterministic
calculus and let  ${\cal Q} \subseteq {\CONTEXTS}$ be a set of functions on 
expressions (\ie\ $\forall Q \in {\cal Q}: Q : {\EXPRESSIONS} \to {\EXPRESSIONS}$).
Then the {\em ${\cal Q}$-experiment operator} 
$F_{D,\cal Q} : 2^{({\EXPRESSIONS} \times {\EXPRESSIONS})} \to 2^{({\EXPRESSIONS} \times {\EXPRESSIONS})}$  
is defined as follows for $\eta \subseteq {\EXPRESSIONS} \times {\EXPRESSIONS}$:
$${s_1}\,{F_{D,\cal Q}(\eta)}\,{s_2} 
\text{ iff } 
{s_1 \maycon_D v_1} \implies \exists v_2.\left(s_2 \maycon v_2 \wedge \forall Q\in{\cal Q}: {Q(v_1)}\,\eta\,{Q(v_2)}\right)$$
\end{definition}

\begin{lemma}
The operator $F_{D,{\cal Q}}$ is monotonous \wrt\ set inclusion, \ie\ for all
binary relations $\eta_1, \eta_2$ on expressions 
${\eta_1 \subseteq \eta_2} \implies {F_{D,{\cal Q}}(\eta_1) \subseteq F_{D,{\cal Q}}(\eta_2)}$.
\end{lemma}
\begin{proof}
Let $\eta_1 \subseteq \eta_2$ and $s_1\,{F_{D,{\cal Q}}(\eta_1)}\,s_2$. 
From the assumption $s_1\,{F_{D,{\cal Q}}(\eta_1)}\,s_2$ the implication 
$s_1 \maycon_D v_1$ $\implies \left(s_2 \maycon_D v_2 \wedge \forall Q\in{\cal Q}: Q(v_1)\,\eta_1\,Q(v_2)\right)$
follows.
From $\eta_1 \subseteq \eta_2$ the implication
$s_1 \maycon v_1 \implies \left(s_2 \maycon_D v_2 \wedge \forall Q\in{\cal Q}: Q(v_1)\,\eta_2\,Q(v_2)\right)$
follows.
Thus, $s_1\,F_{\cal Q}(\eta_2)\,s_2$.
\end{proof}
\noindent Since $F_{D,{\cal Q}}$ is monotonous, its greatest fixpoint exists:
\begin{definition}[${\cal Q}$-Similarity, $\leb_{D,{\cal Q}}$]\label{def:Q-gfp-preorder}
The behavioral preorder $\leb_{D,{\cal Q}}$, called {\em ${\cal Q}$-similarity},
is defined as the greatest fixed point of $F_{D,\cal Q}$.
\end{definition}

We also provide an inductive definition of behavioral equivalence, which is
defined as a contextual preorder where the contexts are restricted to the set
${\cal Q}$ (and the empty context).

\begin{definition}\label{def:Q-simpl-preorder}
Let $D = ({\EXPRESSIONS},{\CONTEXTS}, \to, \ANSWERS)$ be an untyped 
deterministic calculus, and ${\cal Q} \subseteq {\CONTEXTS}$.
Then the relation $\lec_{D,{\cal Q}}$ is defined as follows:
$$s_1 \lec_{D,{\cal Q}} s_2 \text{~iff~} \forall n \geq 0: \forall Q_i \in {\cal Q}: Q_1(Q_2(\ldots(Q_n(s_1))))\maycon_D
\implies Q_1(Q_2(\ldots(Q_n(s_2))))\maycon_D$$
\end{definition}
Note that contextual approximation is a special case of this definition,
\ie\ ${\lec_{D}} = {\lec_{D,\CONTEXTS}}$.

Later in Section~\ref{sec:conv-admissible} we will provide a sufficient
criterion on untyped deterministic calculi that ensures that 
$\leb_{D,{\cal Q}}$ and $\lec_{D,\cal Q}$ coincide.

We are interested in translations between calculi that are faithful 
\wrt\ the corresponding contextual preorders. 
\begin{definition}[\cite{schmidt-schauss-niehren-schwinghammer-sabel-ifip-tcs:08,schmidt-schauss-niehren-schwinghammer-sabel-frank-33:09}]\label{def:translation-compo-etal}
For $i=1,2$ let $({\EXPRESSIONS}_i,{\CONTEXTS}_i,\to_i,{\ANSWERS}_i)$ be
untyped deterministic calculi.
A {\em translation} 
$\tau:  ({\EXPRESSIONS}_1,{\CONTEXTS}_1,\to_1,{\ANSWERS}_1)  \to ({\EXPRESSIONS}_2,{\CONTEXTS}_2,\to_2,{\ANSWERS}_2)$
is a mapping 
$\tau_E: {\EXPRESSIONS}_1 \to {\EXPRESSIONS}_2$ 
and a mapping $\tau_C:{\CONTEXTS}_1 \to {\CONTEXTS}_2$ such that
$\tau_C(\mathit{Id}_1) = \mathit{Id}_2$. 
The following properties of translations are defined:
\begin{itemize}
\item $\tau$ is {\em compositional}\ iff $\tau(C[s]) = \tau(C)[\tau(s)]$ for all $C,s$. 
\item $\tau$ is {\em convergence equivalent}\ iff $s\maycon_1 \iff  \tau(s)\maycon_2$ for all $s$. 
\item $\tau$ is {\em adequate}\ iff for all $s,t \in {\EXPRESSIONS}_1$: $\tau(s) \lec_2 \tau(t) \implies s \lec_1 t$.
\item $\tau$ is {\em fully abstract}\ iff for all $s,t \in {\EXPRESSIONS}_1$: $s \lec_1 t \iff \tau(s) \lec_2 \tau(t)$. 
\item $\tau$ is an {\em isomorphism}\ iff it is fully abstract and a bijection on the quotients \\
\mbox{\hspace*{5mm} $\tau/{\sim}: {\EXPRESSIONS}_1/{\sim}  ~\to {\EXPRESSIONS}_2/{\sim}$}. 
\end{itemize}
\end{definition}
Note that isomorphism means an order-isomorphism between the term-models, where
the orders are $\lec_1 / {\sim}$ and $\lec_2 / {\sim}$ (which are the relations
in the quotient).  

\begin{proposition}[\cite{schmidt-schauss-niehren-schwinghammer-sabel-ifip-tcs:08,schmidt-schauss-niehren-schwinghammer-sabel-frank-33:09}]\label{prop:adequate}
Let $({\EXPRESSIONS}_i,{\CONTEXTS}_i,\to_i,{\ANSWERS}_1)$ for $i=1,2$ be untyped deterministic calculi.
If a translation 
$\tau:  ({\EXPRESSIONS}_1,{\CONTEXTS}_1,\to_1,{\ANSWERS}_1)  \to ({\EXPRESSIONS}_2,{\CONTEXTS}_2,\to_2,{\ANSWERS}_2)$
is compositional and convergence equivalent, then it is also adequate.
\end{proposition}
\begin{proof}
Let $s,t \in {\EXPRESSIONS}_1$ with $\tau(s) \lec_2 \tau(t)$ and let 
$C[s]\maycon_1$ for some $C\in {\CONTEXTS}$. 
It is sufficient to show that this implies $C[t]\maycon_1$: 
Convergence equivalence shows that $\tau(C[s])\maycon_2$. 
Compositionality implies $\tau(C)[\tau(s)]\maycon_2$, and then
$\tau(s) \lec_2 \tau(t)$ implies $\tau(C)[\tau(t)]\maycon_2$. 
Compositionality applied once more implies $\tau(C[t])\maycon_2$, and then
convergence equivalence finally implies $C[t]\maycon_1$. 
\end{proof}

%%%%%%%%%%%%%%%%%%%%%%%%%%%%%%%%%%%%%%%%%%%%%%%%%%%%%%%%%%%%%%%%%%%%%%%%%%%%%%
%% SECTION 3: THREE CALCULI
\section{Three Calculi}\label{sec:calculi}
In this section we introduce the calculi $\LLR$, $\LNAME$, and $\LLAZYCC$.
$\LLR$ is a call-by-need calculus with recursive $\tlet{}$, data constructors,
\tcase-expressions, and the $\tseq$-operator. 
The calculus $\LNAME$ has the same syntactic constructs as $\LLR$, but uses a
call-by-name, rather than a call-by-need, evaluation. The calculus $\LLAZYCC$
does not have $\tletrec$, and also uses a call-by-name evaluation.
 
For all three calculi we assume that there is a (common) set of 
{\em data constructors} $c$ which is partitioned into {\em types}, such that 
every constructor $c$ belongs to exactly one type. 
We assume that for every type $T$ the set of its corresponding data constructors
can be enumerated as $c_{T,1}, \ldots, c_{T,|T|}$ where $|T|$ is the number of
data constructors of type $T$.
We also assume that every constructor has a fixed arity denoted as $\ari(c)$ 
which is a non-negative integer.
We assume that there is a type $\mathit{Bool}$ among the types, with the data
constructors $\tfalse$ and $\ttrue$ both of arity 0.
We require that data constructors occur only fully saturated, \ie\ a constructor
$c$ is only allowed to occur together with $\ari(c)$ arguments, written as 
$(c~s_1~\ldots~s_{\ari(c)})$ where $s_i$ are expressions of the corresponding 
calculus\footnote{%
Partial applications of constructors  of the form $c~s_1~\ldots~s_n$ 
(as \eg\ available in Haskell) thus have to be represented by 
$\lambda x_{n+1}\ldots\lambda x_{\ari(c)}.c~s_1~\ldots~s_n~x_{n+1}\ldots x_{\ari(c)}$.
}. 
We also write $(c~\vect{s})$ as an abbreviation for the constructor application
$(c~s_1~\ldots~s_{\ari(c)})$. 
All three calculi allow deconstruction via \tcase-expressions:
\[
 \tcase_T~s~\tof~(c_{T,1}~x_{1,1}~\ldots~x_{1,\ari(c_{T,1})} \to s_1)\ldots(c_{T,|T|}~x_{|T|,1}~\ldots~x_{|T|,\ari(c_{T,|T|})} \to s_{|T|})
\]
where $s,s_i$ are expressions and $x_{i,j}$ are variables of the corresponding calculus.
Thus there is a $\tcase_T$-construct for every type $T$ and we require that 
there is exactly one case-alternative 
$(c_{T,i}~x_{i,1}~\ldots~x_{i,\ari(c_{T,i})} \to s_i)$ for every constructor
$c_{T,i}$ of type $T$.
In a case-alternative $(c_{T,i}~x_{i,1}~\ldots~x_{i,\ari(c_{T,i})} \to s_i)$ we 
call $c_{T,i}~x_{i,1}~\ldots~x_{i,\ari(c_{T,i})}$ a {\em pattern} and $s_i$ the
right hand side of the alternative. 
All variables in a $\tcase$-pattern must be pairwise distinct.
We will sometimes abbreviate the case-alternatives by $\ialts$ if the exact 
terms of the alternatives are not of interest.
As a further abbreviation we sometimes write $\tif~s_1~\tthen~s_2~\telse~s_3$
for the case-expression 
$(\tcase_{\mathit{Bool}}~s_1~\tof~(\ttrue \to s_2)~(\tfalse \to s_3))$.

We now define the syntax of expressions with $\tletrec$, \ie\ the set
$\LETRECEXPR$ of expressions which are used in both of the calculi
$\LLR$ and $\LNAME$.

\begin{definition}[Expressions $\LETRECEXPR$]
The set $\LETRECEXPR$ of expressions is defined by the following grammar,
where $x, x_i$ are variables:
\begin{eqnarray*}
r,s,t,r_i,s_i,t_i\in \LETRECEXPR &::=& 
           x 
    \syxor (s~t) 
    \syxor (\lambda x. s)
    \syxor \tletrx{x_1 = s_1, \ldots , x_n = s_n}{t} 
\\
    &&\syxor (c~s_1 \ldots s_{\ari(c)})
    \syxor (\tseq~s~t)
    \syxor (\tcase_T~ s ~\tof~alts)
\end{eqnarray*}
We assign the names {\em application}, {\em abstraction}, $\tseq$-expression,
or {\em \tletrec-expression} to the expressions $(s~t)$, $(\lambda x. s)$,
$(\tseq~s~t)$, or $\tletrx{x_1 = s_1, \ldots, x_n = s_n}{t}$, respectively. 
A {\em value} $v$ is defined as an abstraction or a constructor application.
A group of $\tletrec$ bindings is sometimes abbreviated as $\iEnv$. 
We use the notation $\bchainGen{x}{g(i)}{s}{h(i)}{m}{n}$ for the chain 
$x_{g(m)} = s_{h(m)}, x_{g(m+1)} = s_{h(m+1)}, \ldots, x_{g(n)} = s_{h(n)}$
of bindings where $g,h: \mathbb{N} \to \mathbb{N}$ are injective, 
\eg, $\bchainGen{x}{i}{s}{i-1}{m}{n}$ means the bindings 
$x_{m} = s_{m-1}, x_{m+1} = s_{m}, \ldots x_{n} = s_{n-1}$. 
We assume that variables $x_i$ in \tletrec-bindings are all distinct, that 
\tletrec-expressions are identified up to reordering of binding-components,
and that, for convenience, there is at least one binding. 
$\tletrec$-bindings are recursive, \ie, the scope of $x_j$ in 
$\tletrx{x_1 = s_1,\ldots, x_{n-1} = s_{n-1}}{s_{n}}$ are all expressions 
$s_i$ with $1 \leq i \leq n$. 

$\LETRECCTXT$ denotes the set of all contexts for the expressions $\LETRECEXPR$.
\end{definition}

Free and bound variables in expressions and $\alpha$-renamings are defined as usual.
The set of free variables in $s$ is denoted as $\FV(s)$.

\begin{convention}[Distinct Variable Convention]
We use the {\em distinct variable convention}, \ie, all bound variables in expressions
are assumed to be distinct, and free variables are distinct from bound variables.
All reduction rules are assumed to implicitly $\alpha$-rename bound variables 
in the result if necessary. 
\end{convention}

In all three calculi we will use the symbol $\Omega$ for the specific 
(\tletrec-free) expression $(\lambda z. (z~z))~(\lambda x. (x~x))$. 
In all of our calculi $\Omega$ is divergent and the least element of the 
corresponding contextual preorder.
This is proven in \cite{schmidt-schauss-schuetz-sabel:08} for $\LLR$ and
can easily be proven for the other two calculi using standard methods, 
such as context lemmas. 
Note that this property also follows from the Main Theorem~\ref{thm:maintheorem}
for all three calculi.

\subsection{\texorpdfstring{The Call-by-Need Calculus $\LLR$}{The Call-by-Need-Calculus LR}}\label{sec:LR-calc}
We begin with the call-by-need lambda calculus $\LLR$ which is exactly the 
call-by-need calculus of \cite{schmidt-schauss-schuetz-sabel:08}.
It has a rather complex form of reduction rules using variable chains. 
The justification is that this formulation permits direct syntactic proofs
of correctness \wrt\ contextual equivalence for a large class of transformations. 
Several modifications of the reduction strategy, \ removing indirections, do 
not change the semantics of the calculus, however, they appear to be not 
treatable by syntactic proof methods using diagrams 
(see \cite{schmidt-schauss-schuetz-sabel:08}). 
$\LLR$-expressions are exactly the expressions $\LETRECEXPR$.

\begin{definition}\label{def-red-rules} 
The {\em reduction rules} for the calculus and language $\LLR$ are defined in 
\FIGURE~\ref{figure-reductions-LLR}, where the labels $S,V$ are used for the 
exact definition of the normal-order reduction below. 
Several reduction rules are denoted by their name prefix: 
the union of (\rlletin) and (\rllete) is called (\rllet).  
The union of (\rllet), (\rlapp), (\rlcase), and (\rlseq) is called (\rlll).
\end{definition}

\begin{figure*}[htpb] 
\[\begin{array}{|l@{\,}l|}
\hline
&\\[-1.8ex]
(\rlbeta)&C[((\lambda x. s)^S~t)]   \to    C[\tletrxx{x = t}{s}]
\\[1.1ex]
(\rcpin) & \tletrxx{x_1 = (\lambda x.s)^S, \bchainXInd{2}{m}, \iEnv}{C[x_m^V]}\\
& \quad \to \tletrxx{x_1 = (\lambda x.s), \bchainXInd{2}{m},  \iEnv}{C[(\lambda x.s)]}  
\\[1.1ex]
(\rcpe) & \tletrxx{x_1 = (\lambda x.s)^S, \bchainXInd{2}{m},  \iEnv, y = C[x_m^V]}{t}  \\
& \quad \to \tletrxx{x_1 = (\lambda x.s), \bchainXInd{2}{m},  \iEnv, y = C[(\lambda x.s)]}{t}   
\\[1.1ex]
(\rlapp) & C[(\tletrx{\iEnv}{s}^S~t)]  \to   C[\tletrx{\iEnv}{(s~t)}]
\\[1.1ex]
(\rlcase)& C[(\tcase_T~\tletrx{\iEnv}{s}^S~\tof~alts)]  \\
         &\quad\to 
         C[\tletrx{\iEnv}{(\tcase_T~s~\tof~alts)}]
\\[1.1ex]
(\rlseq) & C[(\tseq~\tletrx{\iEnv}{s}^S~t)]   \to   C[\tletrx{\iEnv}{(\tseq~s~t)}]
\\[1.1ex]
(\rlletin)& \tletrxx{\iEnv_1}{\tletrx{\iEnv_2}{s}^S} \to  \tletrxx{\iEnv_1,\iEnv_2}{s}
\\[1.1ex]
(\rllete)& \tletrxx{\iEnv_1, x =  \tletrx{\iEnv_2}{s}^S}{t} \to  \tletrxx{\iEnv_1, \iEnv_2, x = s}{t}
\\[1.1ex]
(\rseqc)& 
C[(\tseq~v^S~s)]   \to  C[s]  \hspace*{1cm}  \mbox{if } v  \mbox{ is a  value}
\\[1.1ex]
(\rseqin)& 
\tletrx{x_1 = v^S, \bchainXInd{2}{m}, \iEnv}{C[(\tseq~x_m^V~s)]} \\
&\quad\to 
 \tletrx{x_1 = v, \bchainXInd{2}{m}, \iEnv}{C[s]} \\
 &  \mbox{if } v  \mbox{ is a constructor application}
\\[1.1ex]
(\rseqe)& 
\tletrx{x_1 = v^S, \bchainXInd{2}{m}, \iEnv, y = C[(\tseq~x_m^V~s)]}{t} \\
& \quad \to  \tletrx{x_1 = v, \bchainXInd{2}{m}, \iEnv, y =C[s]}{t} \\
&  \mbox{if } v  \mbox{ is a {constructor application}}

\\[1.1ex]
(\rcasec)& C[(\tcase_T~(c_{i}~\vect{s})^S \tof\ldots ((c_{i}~\vect{y}) \to t_i) \ldots)] \to C[\tletrx{\bchainN{y}{s}{\ari(c_{i})}}{t_i}]\\ 
   & \mbox{if } \ari(c_{i}) \ge 1\\
(\rcasec)&  C[(\tcase_T~ c_{i}^S \tof\ldots~  (c_{i}  \to t_i) \ldots)]   \to C[t_i] \quad \mbox{if~}  \ari(c_{i}) = 0
\\[1.1ex]
(\rcasein)&\tletrec~x_1 = (c_{i}~\vect{s})^S, \bchainXInd{2}{m}, \iEnv~\\
          &\tin~C[\tcase_T~x_{m}^V~\tof\ldots((c_{i}~\vect{z}) \to t)\ldots]\\
                   &\quad\to\tletrec~  x_1 = (c_{i}~\vect{y}), \bchainN{y}{s}{n}, \bchainXInd{2}{m}, \iEnv \\ 
                   &\quad\phantom{\to\,\,}\tin~     C[\tletrx{\bchainN{z}{y}{\ari(c_{i})}}{t}]~\mbox{if } \ari(c_{i}) \ge  1 \mbox{ and where } y_i  \mbox{ are fresh}
\\[1.1ex]
(\rcasein)  & \tletrec~ x_1 =  c_{i}^S,  \bchainXInd{2}{m}, \iEnv~\tin~C[\tcase_T~ x_{m}^V ~ \ldots~ (c_{i} \to t) \ldots] \\
                   & \quad\to \tletrec~ x_1 =  c_{i}, \bchainXInd{2}{m}, \iEnv~\tin~~C[t]~\mbox{if  }  \ari(c_{i})  = 0
\\[1.1ex]
(\rcasee)&\begin{array}[t]{@{}l@{}l}\tletrec~&x_1 = (c_{i}~\vect{s})^S, \bchainXInd{2}{m},\\
                                                   &u = C[\tcase_T~ x_{m}^V ~\tof\ldots ((c_{i}~\vect{z}) \to t)\ldots], \iEnv\\
                                          \multicolumn{2}{@{}l}{\tin~r}
                    \end{array}\\
     &\quad\to\begin{array}[t]{@{}l@{}l}
                 \tletrec~&x_1 = (c_{i}~\vect{y}), \bchainN{y}{s}{\ari(c_i)}, \bchainXInd{2}{m},\\
                          &u = C[\tletrx{\bchainN{z}{y}{\ari(c_i)}}{t}], \iEnv\\
                 \multicolumn{2}{@{}l}{\tin~r}
               \end{array}
\\
     & \mbox{if } \ari(c_{i}) \ge 1 \mbox{ and where } y_i  \mbox{ are fresh }
\\
(\rcasee)& \tletrec~ x_1 = c_{i}^S,  \bchainXInd{2}{m}, %\\
                u = C[\tcase_T~ x_{m}^V ~ \ldots~ (c_{i} \to t) \ldots],  \iEnv 
                ~\tin~ r\\
                &\quad\to \tletrec~x_1 = c_{i}, \bchainXInd{2}{m}\ldots, u = C[t], \iEnv
                ~\tin~ r~\mbox{ if }  \ari(c_{i}) = 0 \\     
 \hline
\end{array}\]
\caption{Reduction rules of $\LLR$} \label{figure-reductions-LLR}
\end{figure*}

For the definition of the normal order reduction strategy of the calculus 
$\LLR$ we use the labeling algorithm in \FIGURE~\ref{figure:label-LLR} which
detects the position where a reduction rule is applied according to the
normal order.  
It uses the following labels: 
$S$ (subterm), $T$ (top term), $V$ (visited), and $W$ (visited, but not target).
We use $\vee$ when a rule allows two options for a label, \eg\ $s^{S \vee T}$
stands for $s$ labeled with $S$ or $T$.

A labeling rule $l \leadsto r$ is applicable to a (labeled) expression $s$ if
$s$ matches $l$ with the labels given by $l$, where $s$ may have more labels
than $l$ if not otherwise stated. 
The labeling algorithm takes an expression $s$ as its input and exhaustively
applies the rules in \FIGURE~\ref{figure:label-LLR} to $s^T$, where no other
subexpression in $s$ is labeled. 
The label $T$ is used to prevent the labeling algorithm from descending into 
\tletrec-environments that are not at the top of the expression.
The labels $V$ and $W$ mark the visited bindings of a chain of bindings, where
$W$ is used for variable-to-variable bindings.
The labeling algorithm either terminates  with \textit{fail} or with success,  
where in general the direct super\-term of the $S$-marked subexpression 
indicates a potential normal-order  redex.
The use of such a labeling algorithm corresponds to the search of a redex in
term graphs where it is usually called unwinding. 

\begin{figure*}[t]
\[\begin{array}{|lll|}
\hline
&&\\[-1.8ex]
\tletrx{\iEnv}{s}^{T} & \leadsto & \tletrx{\iEnv}{s^{S}}^V\\
(s~t)^{S \vee T} & \leadsto & (s^{S}~t)^V\\
(\tseq~s~t)^{S \vee T} & \leadsto & (\tseq~s^{S}~t)^V\\
(\tcase_T~s~\tof~alts)^{S \vee T} & \leadsto & (\tcase_T~s^{S}~\tof~alts)^V\\
 \tletrx{x=s, \iEnv}{C[x^{S}]}  & \leadsto & \tletrx{x=s^{S}, \iEnv}{C[x^V]}\\
 \tletrx{x=s^{V \vee W}, y = C[x^{S}], \iEnv}{t}  & \leadsto & \textit{fail}\\
 \tletrx{x = C[x^{S}], \iEnv}{s}  & \leadsto & \textit{fail}\\
 \tletrx{x=s, y = C[x^{S}], \iEnv}{t}  & \leadsto & \tletrx{x=s^{S}, y = C[x^V], \iEnv}{t}\\
       & & \mbox{if } C[x] \not= x \\
\tletrx{x=s, y = x^{S}, \iEnv}{t} & \leadsto & \tletrx{x=s^{S}, y = x^W, \iEnv}{t}\\
\hline
\end{array}\]
\caption{Labeling algorithm for $\LLR$\label{figure:label-LLR}}
\end{figure*}

\begin{definition}[Normal Order Reduction of $\LLR$] \label{def-no-reduction}
Let $s$ be an expression. 
Then a single normal order reduction step $\RRAP{\LR}$ is defined as follows:
first the labeling algorithm in \FIGURE~\ref{figure:label-LLR} is applied to $s$.  
If the labeling algorithm terminates successfully, then one of the rules 
in \FIGURE~\ref{figure-reductions-LLR} is applied, if possible, where the
labels $S,V$  must match the labels in the expression $s$ 
(again $s$ may have more labels).   
The {\em normal order redex} is defined as the left-hand side of the applied
reduction rule.
The notation for a normal-order reduction that applies the rule $a$ is 
$\xrightarrow{\LR,a}$, \eg\ $\xrightarrow{\LR,\rlapp}$ applies the rule $(\rlapp)$. 
\end{definition}

The normal order reduction of $\LLR$ implements a call-by-need reduction with sharing
which avoids substitution of arbitrary expressions.
We describe the rules: 
The rule (\rlbeta) is a sharing variant of classical $\beta$-reduction,
where the argument of an abstraction is shared by a new \tletrec-binding 
instead of substituting the argument in the body of an abstraction. 
The rules (\rcpin) and (\rcpe) allow to copy abstractions into needed positions. 
The rules (\rlapp), (\rlcase), and (\rlseq) allow moving \tletrec-expressions 
to the top of the term if they are inside a reduction position of an application,
a \tcase-expression, or a \tseq-expression. 
To flatten nested \tletrec-expressions, the rules (\rlletin) and (\rllete) are
added to the reduction.
Evaluation of $\tseq$-expressions is performed by the rules (\rseqc), (\rseqin),
and (\rseqe), where the first argument of $\tseq$ must be a value (rule \rseqc)
or it must be a variable which is bound in the outer $\tletrec$-environment to
a constructor application.
Since normal order reduction avoids copying constructor applications, the rules
(\rseqin) and (\rseqe) are required.
Correspondingly, the evaluation of $\tcase$-expressions requires several variants: 
there are again three rules for the cases where the argument of $\tcase$ is 
already a constructor application (rule (\rcasec)) or where the argument is a
variable which is bound to a constructor application (perhaps by several 
indirections in the \tletrec-environment) which are covered by the rule 
(\rcasein) and (\rcasee). 
All three rules have two variants: one variant for the case when a constant
is scrutinized (and thus no arguments need to be shared by new $\tletrec$-bindings)
and another variant for the case when arguments are present (and thus the arity of
the scrutinized constructor is strictly greater than 0). 
For the latter case the arguments of the constructor application are shared by
new \tletrec-bindings, such that the newly created variables can be used as 
references in the right hand side of the matching alternative.

\begin{definition}
A {\em reduction context} $R_{\LR}$ is any context, such that its hole is 
labeled with $S$ or $T$ by the $\LLR$-labeling algorithm. 
\end{definition}
Of course, reduction contexts could also be defined recursively, as 
in \cite[Definition 1.5]{schmidt-schauss-schuetz-sabel:08}, but such a 
definition is very cumbersome due to a large number of special cases. 
The labeling algorithm provides a definition that, in our experience,
is easier to work with.
 
By induction on the term structure one can easily verify that the normal
order redex, as well as the normal order reduction, is unique. 
A {\em weak head normal form in $\LLR$ ($\LLR$-WHNF)} is either an
abstraction $\lambda x.s$, a constructor application $(c~s_1~\ldots~s_{\ari(c_i)})$,
or an expression $\tletrx{\iEnv}{v}$ where $v$ is a constructor application
or an abstraction, or an expression of the form 
$\tletrx{x_1 = v, \bchainXInd{2}{m},\iEnv}{x_{m}}$, where $v =   (c~\vect{s})$. 
We distinguish abstraction-WHNF (AWHNF) and  constructor WHNF (CWHNF) based
on whether the value $v$ is an abstraction or a constructor application, 
respectively.
The notions of convergence, divergence and contextual approximation are as
defined in Sect.~\ref{sec:common}.
If there is no normal order reduction originating at an expression $s$ then
$s \DIV_\LR$. 
This, in particular, means that expressions for which the labeling algorithm
fails to find a redex, or for which there is no matching constructor for a
subexpression (that is a WHNF) in a \tcase{} redex position, or expressions
with cyclic dependencies like $\tletrec~{x=x}~\tin~x$, are diverging. 

\begin{example}
We consider the expression
$
 s_1 := \tletrec~x = (y~\lambda u.u), y = \lambda z.z~\tin~x
$.
The labeling algorithm applied to $s_1$ yields
$
 (\tletrec~x = (y^V~\lambda u.u)^V, y = (\lambda z.z)^S~\tin~x^V)^V
$.
The reduction rule that matches this labeling is the reduction rule \mbox{(\rcpe)}, \ie\
$
 s_1\xrightarrow{\LR}(\tletrec~x = ((\lambda z'.z')~\lambda u.u), y = (\lambda z.z)~\tin~x) = s_2
$.
The labeling of $s_2$ is
$
(\tletrec~x = ((\lambda z'.z')^S~\lambda u.u)^V, y = (\lambda z.z)~\tin~x^V)^V
$, which makes the rule  \mbox{(\rlbeta)} applicable, \ie\
$s_2\xrightarrow{\LR} (\tletrec~x = (\tletrec~z'=\lambda u.u~\tin~z'), y = (\lambda z.z)~\tin~x) = s_3$.
The labeling of  $s_3$ is
$
  (\tletrec~x = (\tletrec~z'=\lambda u.u~\tin~z')^S, y = (\lambda z.z)~\tin~x^V)^V
$.
Thus an \mbox{(\rllete)}-reduction is applicable to $s_3$, \ie\
$s_3 \xrightarrow{\LR} (\tletrec~x = z', z'= \lambda u.u, y = (\lambda z.z)~\tin~x) = s_4$.
Now $s_4$ gets labeled as 
$
(\tletrec~x = z'^W, z'= (\lambda u.u)^S, y = (\lambda z.z)~\tin~x^V)^V
$,
and a \mbox{(\rcpin)}-reduction is applicable, \ie\
$
s_4 \xrightarrow{\LR} 
(\tletrec~x = z', z'= (\lambda u.u), y = (\lambda z.z)~\tin~(\lambda u.u)) = s_5$.
The labeling algorithm applied to $s_5$ yields
$(\tletrec~x = z', z'= (\lambda u.u), y = (\lambda z.z)~\tin~(\lambda u.u)^S)^V$,
but no reduction is applicable to $s_5$, since $s_5$ is a WHNF.
\end{example}

Concluding, the calculus $\LLR$ is defined by the tuple 
$(\LETRECEXPR,\LETRECCTXT,\xrightarrow{\LR},{\ANSWERS}_{\LR})$ where
${\ANSWERS}_{\LR}$ are the $\LLR$-WHNFs, where we equate alpha-equivalent
expressions, contexts and answers. 

\vspace{1mm} 

In \cite{schmidt-schauss-schuetz-sabel:08} correctness of several program
transformations was shown: 
\begin{figure*}[t] 
\[\begin{array}{|l@{\,}l|}
\hline
&\\[-1.8ex]
\mbox{(gc)} & C[\tletrec~\{x_i = s_i\}_{i=1}^n~\tin~t] \to C[t],~~~\text{if $\FV(t) \cap \{x_1,\ldots,x_n\} = \emptyset$}
\\[1.1ex]
\mbox{(gc)} & C[\tletrec~\{x_i = s_i\}_{i=1}^n,\{y_i=t_i\}_{i=1}^m~\tin~t] \to C[\tletrec~\{y_i=t_i\}_{i=1}^m~\tin~t],
\\&\multicolumn{1}{r|}{\text{if $(\FV(t) \cup\bigcup_{i=1}^m{\FV(t_i)}) \cap \{x_1,\ldots,x_n\} = \emptyset$}}
\\[1.1ex]
\mbox{(lwas)} &C[(s~(\tletrec~\iEnv~\tin~t))] \to C[\tletrec~\iEnv~\tin~(s~t)]
\\[1.1ex]
\mbox{(lwas)} &C[(c~s_1~\ldots(\tletrec~\iEnv~\tin~s_i)\ldots s_n)] \to C[\tletrec~\iEnv~\tin~(c~s_1~\ldots~s_i~\ldots s_n)]
\\[1.1ex]
\mbox{(lwas)} &C[(\tseq~s~(\tletrec~\iEnv~\tin~t))] \to C[\tletrec~\iEnv~\tin~\tseq~s~t]
\\\hline
\end{array}\]
\caption{Transformations for garbage collection and \tletrec-shifting\label{figure-gc-lwas}}
\end{figure*}
\begin{theorem}[{\cite[Theorems 2.4 and 2.9]{schmidt-schauss-schuetz-sabel:08}}]\label{theo:lr-trans-corr}
All reduction rules shown in \FIGURE~\ref{figure-reductions-LLR} are correct
program transformations, even if they are used with an arbitrary context $C$
in the rules without requiring the labels.
The transformations for garbage collection (gc) and for shifting of \tletrec-expressions
(lwas) shown in \FIGURE~\ref{figure-gc-lwas} are also correct program transformations.\qed
\end{theorem}
 
\subsection{\texorpdfstring{The Call-by-Name Calculus $\LNAME$}{The Call-by-Name Calculus LNAME}}\label{subsec:name-calc}
Now we define a call-by-name calculus on $\LETRECEXPR$-expressions.
The calculus $\LNAME$ has $\LETRECEXPR$ as expressions, but the reduction rules
are different from $\LLR$.
The calculus $\LNAME$ does not implement a sharing strategy but instead performs
the usual call-by-name beta-reduction and copies arbitrary expressions directly
into needed positions.

\begin{figure*}[t]
\[\begin{array} {|lll|}
\hline
&&\\[-1.8ex]
(\tletrec~\iEnv~\tin~s)^{X} &\leadsto& (\tletrec~\iEnv~\tin~s^{X})\text{ if $X$ is $S$ or $T$}\\
(s~t)^{S \vee T} &\leadsto& (s^S~t)\\
(\tseq~s~t)^{S \vee T} &\leadsto& (\tseq~s^S~t)\\
(\tcase_T~s~\tof~\ialts)^{S \vee T} &\leadsto& (\tcase_T~s^S~\tof~\ialts)\\
\hline
\end{array}
\]
\caption{Labeling algorithm for $\LNAME$\label{fig:labeling-LNAME}}
\end{figure*}  

In \FIGURE~\ref{fig:labeling-LNAME} the rules of the labeling algorithm for $\LNAME$ are given.
The algorithm uses the labels $S$ and $T$. For an expression $s$ the labeling starts with $s^T$.

An $\LNAME$ reduction context $R_{\NAME}$ is any context where the hole is labeled $T$ or $S$
by the labeling algorithm, or more formally they can be defined as follows:
\begin{definition}
{\em Reduction contexts} $R_{\NAME}$ are contexts of the form $L[A]$ where the context classes
${\cal A}$ and ${\cal L}$ are defined by the following grammar, where $s$ is any expression:

$$
\begin{array}{rcl}
 L\in {\cal L}  &::=& [\cdot] \syxor \tletrec~\iEnv~\tin~L\\ 
 A \in {\cal A} &::=&  [\cdot] \syxor (A~s) \syxor (\tcase_T~A~\tof~alts) \syxor (\tseq~A~s)
\end{array}
$$
\end{definition}

Normal order reduction $\xrightarrow{\NAME}$ of $\LNAME$ is defined by the rules
shown in \FIGURE~\ref{figure-reductions-LNAME} where the labeling algorithm 
according to \FIGURE~\ref{fig:labeling-LNAME} must be applied first.
Note that the rules (\rseqc), (\rlapp), (\rlcase), and (\rlseq) are identical
to the rules for $\LLR$ (in \FIGURE~\ref{figure-reductions-LLR}), but the 
labeling algorithm is different.

\begin{figure}[t]
\[
\begin{array}{|ll|}
\hline
&\\[-1.8ex]
(\rbeta) & C[(\lambda x.s)^S~t] \to C[s[t/x]]
\\[1.1ex]
(\rgcp)   &
C_1[\tletrec~\iEnv,~x=s~\tin~C_2[x^{S\vee T}]] \to  C_1[\tletrec~\iEnv,~x=s~\tin~C_2[s]] 
\\[1.1ex]
(\rlapp) & C[(\tletrx{\iEnv}{s}^S~t)]  \to   C[\tletrx{\iEnv}{(s~t)}]
\\[1.1ex]
(\rlcase)& C[(\tcase_T~\tletrx{\iEnv}{s}^S~\tof~alts)]  \\
         &\quad\to 
         C[\tletrx{\iEnv}{(\tcase_T~s~\tof~alts)}]
\\[1.1ex]
(\rlseq) & C[(\tseq~\tletrx{\iEnv}{s}^S~t)]   \to   C[\tletrx{\iEnv}{(\tseq~s~t)}]
\\[1.1ex]
(\rseqc)& 
C[(\tseq~v^S~s)]   \to  C[s]  \hspace*{1cm}  \mbox{if } v  \mbox{ is a  value}
\\[1.1ex]
(\rcase) &C[(\tcase_T~(c~s_1 \dots s_{ar(c)})^S~\tof \ldots ((c~x_1 \dots x_{ar(c)}) \rightarrow t)\ldots)]\\
&\quad \to C[t[{s_1}/{x_1}, \dots, {s_{ar(c)}}/{x_{ar(c)}}]] 
\\[1.1ex]
\hline
\end{array}
\]
\caption{Normal order reduction rules $\xrightarrow{\NAME}$ of $\LNAME$ \label{figure-reductions-LNAME}}
\end{figure}

Unlike $\LLR$, the normal order reduction of $\LNAME$ allows substitution of arbitrary
expressions in (\rbeta), (\rcase), and (\rgcp) rules.
An additional simplification (compared to $\LLR$) is that nested $\tletrec$-expressions
are not flattened by reduction (\ie\ there is no (\rllet)-reduction in $\LNAME$).
As in $\LLR$ the normal order reduction of $\LNAME$ has reduction rules (\rlapp),
(\rlcase), and (\rlseq) to move \tletrec-expressions out of an application, a
\tseq-expression, or a \tcase-expression.

\vspace*{1mm} 
Note that $\xrightarrow{\NAME}$ is unique. 
An $\LNAME$-WHNF is defined as an expression either of the form $L[\lambda x.s]$
or of the form $L[(c~s_1~\ldots~s_{ar(c)})]$ where $L$ is an ${\cal L}$ context. 
Let ${\ANSWERS}_{\NAME}$ be the set of $\LNAME$-WHNFs, then the calculus $\LNAME$
is defined by the tuple 
$(\LETRECEXPR,\LETRECCTXT,\xrightarrow{\NAME},{\ANSWERS}_{\NAME})$ 
(modulo $\alpha$-equivalence). 

\subsection{\texorpdfstring{The Extended Lazy Lambda Calculus $\LLAZYCC$}{The Extended Lazy Lambda Calculus Llcc}}
In this subsection we give a short description of the  lazy lambda 
calculus \cite{abramsky-lazy:90} extended by data constructors, 
\tcase-expressions and \tseq-expressions, denoted with $\LLAZYCC$. 
Unlike the calculi $\LNAME$ and $\LLR$, this calculus has no 
$\tletrec$-expressions.
The set $\LAMBDAEXPR$ of $\LLAZYCC$-expressions is that of the usual (untyped)
lambda calculus extended by data constructors, $\tcase$, and $\tseq$: 
\[
r,s,t,r_i,s_i,t_i \in \LAMBDAEXPR ::=  x \syxor(s ~t) \syxor  (\lambda x. s) \syxor (c~s_1\ldots s_{\ari(c)}) \syxor (\tcase_T~s~\tof~\ialts) \syxor (\tseq~s~t)
\]

Contexts $\LAMBDACTXT$ are $\LAMBDAEXPR$-expressions where a subexpression is
replaced by the hole $[\cdot]$.
The set ${\ANSWERS}_{\LCC}$ of {\em answers} (or also {\em values}) are the 
$\LLAZYCC$-abstractions and constructor applications. 
Reduction contexts ${\cal R}_{\LCC}$ are defined by the following grammar, where
$s$ is any $\LAMBDAEXPR$-expression:
\[
R_{\LCC} \in {\cal R}_{\LCC}  :=   [\cdot] \syxor (R_\LCC~s) \syxor \tcase_T~R_\LCC~\tof~\ialts \syxor \tseq~R_\LCC~s
\]

An $\xrightarrow{\LCC}$-reduction is defined by the three rules shown 
in \FIGURE~\ref{figure-reductions-LLAZY}, and thus the calculus $\LLAZYCC$ is defined 
by the tuple $(\LAMBDAEXPR,\LAMBDACTXT,\xrightarrow{\LCC},{\ANSWERS}_{\LCC})$
(modulo $\alpha$-equivalence).

\begin{figure}[htpb]
\centering
\begin{tabular}{|ll|}
\hline
&\\[-1.8ex]
$(\rnbeta)$&  $R_\LCC[((\lambda x.s)~t)] \xrightarrow{\LCC} R_\LCC[s[t/x]]$
\\
$(\rncase)$ & $R_{\LCC}[(\tcase_T~(c~s_1 \ldots s_{\ari(c)})~\tof~\ldots ((c~x_1 \ldots x_{\ari(c)}) \rightarrow t) \ldots)]$\\
&$\quad \xrightarrow{\LCC} t[s_1/x_1,\ldots,s_{\ari(c)}/x_{\ari(c)}]$
\\
$(\rnseq)$ & $R_{\LCC}[\tseq~v~s] \xrightarrow{\LCC} R_{\LCC}[s]$, if $v$ is an abstraction or a constructor application
\\
\hline
\end{tabular}
\caption{Normal order reduction $\xrightarrow{\LCC}$ of $\LLAZYCC$\label{figure-reductions-LLAZY}}
\end{figure}

%%%%%%%%%%%%%%%%%%%%%%%%%%%%%%%%%%%%%%%%%%%%%%%%%%%%%%%%%%%%%%%%%%%%%%%%%%%%%%
%% SECTION 4: PROPERTIES OF SIMILARITY AND EQUIVALENCES IN LCC
\section{\texorpdfstring{Properties of Similarity and Equivalences in $\LLAZYCC$}{Properties of Similarity and Equivalences in LCC}\label{sec:simulation-lazy}}
An applicative bisimilarity for $\LLAZYCC$ and other alternative definitions
are presented in subsection \ref{subsec:bisimulation-llc}. 
As a preparation, we first analyze similarity for deterministic calculi in general. 

\subsection{Characterizations of Similarity in Deterministic Calculi}\label{sec:conv-admissible}
In this section we prove that for deterministic calculi (see Def.~\ref{def:calculus}),
the applicative similarity and its generalization to extended calculi, defined as the
greatest fixpoint of an operator on relations, is equivalent to the inductive definition
using Kleene's fixpoint theorem.

This implies that for deterministic calculi employing only beta-reduction,
applicative similarity can be equivalently defined as $s \leb t$, iff for all
$n \geq 0$ and closed expressions $r_i, i=1,\ldots,n$, the implication 
$(s~r_1 \ldots r_n)\maycon_D \implies (t~r_1 \ldots r_n)\maycon_D$
holds, provided the calculus is convergence-admissible, which means that for all 
$r$: $(s~r) \maycon_D  v \iff \exists v': s \maycon_D v' \wedge  (v'~r)\maycon_D v$ 
(see Def. \ref{def:convergence-admissible}).

This approach has a straightforward extension to calculi with other types of reductions, 
such as case- and seq-reductions.
The calculi may also consist of a set of open expressions, contexts, and answers, 
as well as a subcalculus consisting of closed expressions,  closed contexts 
and closed answers. 
We will use convergence-admissibility only for closed variants of the calculi.

In the following we assume $D = ({\EXPRESSIONS}, {\CONTEXTS}, \to, {\ANSWERS})$ 
to be an untyped deterministic calculus
and ${\cal Q} \subseteq {\CONTEXTS}$ be a set of functions on expressions.
Note that the relations $\leb_{D,{\cal Q}}$ and $\lec_{D,{\cal Q}}$ are defined in  Definitions \ref{def:Q-gfp-preorder}
and \ref{def:Q-simpl-preorder}, respectively.

\begin{lemma}\label{lemma-leqb-fixpointeq}
For all expressions $s_1,s_2 \in{\EXPRESSIONS}$ the following holds: 
$s_1 \leb_{D,{\cal Q}} s_2$ if, and only if, $s_1 \maycon_D v_1 \implies \left(s_2 \maycon_D v_2 \wedge \forall Q\in{\cal Q}: Q(v_1)\leb_{D,{\cal Q}} Q(v_2)\right)$.
\end{lemma}
\begin{proof}
Since $\leb_{D,{\cal Q}}$ is a fixpoint of $F_{D,{\cal Q}}$, we have $\leb_{D,{\cal Q}}\ = F_{D,{\cal Q}}(\leb_{D,{\cal Q}})$.
This equation is equivalent to the claim of the lemma.
\end{proof}

Now we show that the operator $F_{D,{\cal Q}}$ is 
lower-continuous, and thus we can apply Kleene's fixpoint theorem to 
derive an alternative characterization of $\leb_{D,{\cal Q}}$.

For infinite chains of sets $S_1, S_2\ldots,$ we 
define the greatest lower bound \wrt\ set-inclusion ordering as ${{\mathrm{glb}}}(S_1,S_2,\ldots) = \bigcap\limits_{i=1}^{\infty}{S_i}$.

\begin{proposition}
$F_{\cal Q}$ is lower-continuous \wrt\ countably infinite descending chains $C = \eta_1 \supseteq \eta_2 \supseteq \ldots$,
\ie\ ${\mathrm{glb}}(F_{\cal Q}(C)) = F_{\cal Q}({\mathrm{glb}}(C))$ where $F_{\cal Q}(C)$ is the infinite descending chain 
$F_{\cal Q}(\eta_1) \supseteq F_{\cal Q}(\eta_2) \supseteq \ldots$.
\end{proposition}
\begin{proof}
``$\supseteq$'':
Since ${\mathrm{glb}}(C) = \bigcap\limits_{i=1}^{\infty}{\eta_i}$, we have for all $i$: ${\mathrm{glb}}(C)\subseteq \eta_i$. Applying monotonicity of $F_{\cal Q}$ yields
$F_{\cal Q}({\mathrm{glb}}(C)) \subseteq F_{\cal Q}(\eta_i)$ for all $i$. This implies $F_{\cal Q}({\mathrm{glb}}(C))\subseteq \bigcap\limits_{i=1}^{\infty}{F_{\cal Q}(\eta_i)}$, \ie\ 
$F_{\cal Q}({\mathrm{glb}}(C))\subseteq {\mathrm{glb}}(F_{\cal Q}(C))$.

``$\subseteq$'':
Let $(s_1,s_2) \in {\mathrm{glb}}(F_{\cal Q}(C))$, \ie\ for all $i$: $(s_1, s_2) \in F_{\cal Q}(\eta_i)$. Unfolding the definition of $F_{\cal Q}$ gives:
$\forall i: s_1 \maycon_D v_1 \implies \left(s_2 \maycon_D v_2 \wedge \forall Q \in {\cal Q}: Q(v_1)\,\eta_i\,Q(v_2)\right)$. 
Now we can move the universal quantifier for $i$ inside the formula:
 $s_1 \maycon_D v_1 \implies (s_2 \maycon_D v_2 \wedge \forall Q \in {\cal Q}: \forall $i$: Q(v_1)\,\eta_i\,Q(v_2))$.
This is equivalent to 
$s_1 \maycon_D v_1 \implies (s_2 \maycon_D v_2 \wedge \forall Q \in {\cal Q}: {Q(v_1)\,\big(\bigcap\limits_{i=1}^{\infty}\eta_i\big)\,Q(v_2)})$
or 
$s_1 \maycon_D v_1 \implies (s_2 \maycon_D v_2 \wedge \forall Q \in {\cal Q}: (Q(v_1),Q(v_2))\in {\mathrm{glb}}(C))$
and thus $(s_1,s_2) \in F_{\cal Q}({\mathrm{glb}}(C))$.
\end{proof}

\begin{definition}
Let $\leb_{D,{\cal Q},i}$ for $i \in {\mathbb{N}}_0$ be defined as follows:  
\[
 \begin{array}{rcl@{~~~\mbox{and}~~~~}rcl}
   \leb_{D,{\cal Q},0} &= &{\EXPRESSIONS}\times{\EXPRESSIONS} &\leb_{D,{\cal Q},i} &=&F_{D,{\cal Q}}(\leb_{D,{\cal Q},{i-1}}) \text{if $i > 0$}
  \end{array}
\]
\end{definition}

\begin{theorem}\label{theo:inductive-character-gfp}
 $\leb_{D,{\cal Q}}\ = \bigcap\limits_{i=1}^{\infty}{\leb_{D,{\cal Q},i}}$
\end{theorem}
\begin{proof}
The claim follows from Kleene's fixpoint theorem, since $F_{\cal Q}$ is monotonous and lower-continuous, and since
   $\leb_{D,{\cal Q},i+1}\ \subseteq\  \leb_{D,{\cal Q},i}$ for all $i \ge 0$.  
\end{proof}
%% Manfred: korrigiert, ein ,b,  weg:
This representation of $\leb_{D,\cal Q}$ allows {\em inductive} proofs to show similarity. 
Now we show that ${\cal Q}$-similarity is identical to $\lec_{D,{\cal Q}}$ under moderate conditions,
\ie\ our characterization result will only apply if the underlying calculus is convergence-admissible \wrt\ ${\cal Q}$:
\begin{definition}\label{def:convergence-admissible}
An untyped deterministic calculus $({\EXPRESSIONS},{\CONTEXTS}, \to, {\ANSWERS})$ is {\em convergence-admissible \wrt\ ${\cal Q}$} if, and only if
$\forall Q\in{\cal Q}, s \in {\EXPRESSIONS}, v \in \ANSWERS: 
Q(s)\maycon_D v \iff \exists v': s \maycon_D v' \wedge Q(v')\maycon_D v$
\end{definition}
Convergence-admissibility can be seen as a restriction on choosing the set ${\cal Q}$: In
most calculi (subsets of) reduction contexts satisfy the property for convergence-admissibility, 
while including non-reduction contexts into ${\cal Q}$ usually breaks convergence-admissibility. 

\begin{lemma}\label{lemma:ca-implies-leqQ-over-O}
 Let  $({\EXPRESSIONS},{\CONTEXTS}, \to, {\ANSWERS})$ be convergence-admissible \wrt\ ${\cal Q}$. Then the following holds:
\begin{itemize}
 \item $s_1 \lec_{D,{\cal Q}} s_2 \implies {Q(s_1) \lec_{D,{\cal Q}} Q(s_2)}$ for all $Q \in {\cal Q}$
 \item $s_1 \lec_{D,{\cal Q}} s_2, s_1 \maycon_D v_1,\text{ and } s_2\maycon_D v_2 \implies v_1 \lec_{D,\cal Q} v_2$
 \end{itemize}
\end{lemma}
\begin{proof}
The first part is easy to verify.
For the second part it is important that $D$ is deterministic. Let $s_1 \lec_{D,{\cal Q}} s_2$, and  $s_1 \maycon_D v_1$, $s_2 \maycon_D v_2$ hold.
Assume that $Q_1(\ldots(Q_n(v_1))) \maycon_D v_1'$ for some $n \geq 0$ where all $Q_i \in {\cal Q}$.
Convergence-admissibility implies $Q_1(\ldots((Q_n(s_1)))) \maycon_D v_1'$. 
Now $s_1 \lec_{D,{\cal Q}} s_2$ implies $Q_1(\ldots(Q_n(s_2)))\maycon_D v_2'$. Finally, 
convergence-admissibility (applied multiple times) shows that $s_2 \maycon_D v_2$ and 
$Q_1(\ldots(Q_n(v_2)))\maycon_D v_2'$ holds.
\end{proof}

We prove that $\leb_{D,{\cal Q}}$ respects functions $Q \in {\cal Q}$ provided the underlying
deterministic calculus is convergence-admissible \wrt\ ${\cal Q}$:

\begin{lemma}\label{lemma:simb-implies-Qsimb}
Let $({\EXPRESSIONS},{\CONTEXTS}, \to, {\ANSWERS})$ be convergence-admissible \wrt\ ${\cal Q}$.
Then for all $s_1,s_2 \in E:$ $s_1 \leb_{D,{\cal Q}} s_2 \implies Q(s_1) \leb_{D,{\cal Q}} Q(s_2)$ for all $Q\in{\cal Q}$
\end{lemma}
\begin{proof}
Let $s_1 \leb_{D,{\cal Q}} s_2$, $Q_0 \in{\cal Q}$,  and $Q_0(s_1)\maycon_D v_1$.
By convergence admissibility $s_1\maycon_D v_1'$ holds and $Q_0(v_1')\maycon_D v_1$.
Since $s_1 \leb_{D,{\cal Q}} s_2$ this implies $s_2\maycon_D v_2'$ and
for all $Q\in{\cal Q}: Q(v_1') \leb_{D,\cal Q} Q(v_2')$.  Hence, from $Q_0(v_1') \maycon_D v_1$
we derive $Q_0(v_2')\maycon_D v_2$. Convergence admissibility now implies $Q_0(s_2)\maycon_D v_2$.

It remains to show  for all $Q\in{\cal Q}$: $Q(v_1) \leb_{D,\cal Q} Q(v_2)$: 
Since $Q_0(v_1')\maycon_D v_1$ and $Q_0(v_2')\maycon_D v_2$, applying Lemma~\ref{lemma-leqb-fixpointeq} to $Q_0(v_1') \leb_{D,\cal Q} Q_0(v_2')$ implies $Q(v_1) \leb_{D,\cal Q} Q(v_2)$ for all $Q \in {\cal Q}$.
\end{proof}

We now prove that $\lec_{D,{\cal Q}}$ and ${\cal Q}$-similarity coincide for convergence-admissible deterministic calculi:

\begin{theorem}\label{thm:conv-admissibile-implies}
Let $({\EXPRESSIONS},{\CONTEXTS}, \to, {\ANSWERS})$ be convergence-admissible \wrt\ ${\cal Q}$.
Then \mbox{${\lec_{D,\cal Q}}\,=\,{\leb_{D,{\cal Q}}}$}.
\end{theorem}
\begin{proof}
``$\subseteq$'':
Let $s_1 \lec_{D,{\cal Q}} s_2$. We use Theorem~\ref{theo:inductive-character-gfp} and show
$s_1 \leb_{D,{\cal Q},i} s_2$ for all $i$. We use induction on $i$.
The base case ($i=0$) obviously holds.
Let $i > 0$ and let $s_1 \maycon_D v_1$. Then $s_1 \lec_{D,{\cal Q}} s_2$ implies
$s_2\maycon_D v_2$. Thus, it is sufficient to show that $Q(v_1) \leb_{D,{\cal Q},i-1} Q(v_2)$
for all $Q \in {\cal Q}$: As induction hypothesis we use that $s_1 \lec_{D,{\cal Q}}~s_2 \implies s_1 \leb_{D,{\cal Q},i-1} s_2$ holds.
Using Lemma~\ref{lemma:ca-implies-leqQ-over-O} twice and $s_1 \lec_{D,{\cal Q}} s_2$, we have
$Q(v_1) \lec_{D,{\cal Q}} Q(v_2)$. The induction hypothesis shows that $Q(v_1) \leb_{D,{\cal Q},i-1} Q(v_2)$.
Now the definition of $\leb_{D,{\cal Q},i}$ is satisfied, which shows $s_1 \leb_{D,{\cal Q},i} s_2$.

``$\supseteq$'':
 Let $s_1 \leb_{D,{\cal Q}} s_2$. By induction on the number $n$ of $Q$-contexts we show  
$\forall n, Q_i \in{\cal Q}: Q_1(\ldots(Q_n(s_1)))\maycon_D \implies Q_1(\ldots(Q_n(s_2)))\maycon_D$.
The base case follows from $s_1 \leb_{D,{\cal Q}} s_2$. For the induction step we use the following induction hypothesis:
$t_1 \leb_{D,{\cal Q}} t_2 \implies \forall j < n, Q_i \in{\cal Q}: Q_1(\ldots(Q_j(t_1)))\maycon_D \implies Q_1(\ldots(Q_j(t_2)))\maycon_D$
for all $t_1, t_2$.
Let $Q_1(\ldots(Q_n(s_1)))\maycon_D$.
From Lemma~\ref{lemma:simb-implies-Qsimb} we have $r_1 \leb_{D,{\cal Q}} r_2$, where $r_i = Q_n(s_i)$.
Now the induction hypothesis shows that $Q_1(\ldots(Q_{n-1}(r_1)))\maycon_D \implies Q_1(\ldots(Q_{n-1}(r_2)))\maycon_D$
and thus $Q_1(\ldots(Q_n(s_2)))\maycon_D$.
\end{proof}

\subsection{\texorpdfstring{Applicative Simulation in $\LLAZYCC$}{Applicative simulation in LCC}}\label{subsec:bisimulation-llc}
In this section we will show that applicative similarity and contextual preorder coincide in $\LLAZYCC$. 
\begin{figure*}[t]
\begin{minipage}{.53\textwidth}
\begin{tikzpicture}
\node (lec) at (0,0) {$\lec_\LCC$};
\node (leb) at (2,-1) {$\leb_\LCC^o$};
\node (lebq) at (4,0) {$\leb_{\LCC,Q_\LCC}^o$};
\node (lecq) at (6.5,0) {$\lec_{\LCC,Q_\LCC}^o$};
\node (cand) at (2,-2) {$\cand$};
\draw[double,double distance=1.2pt,dotted,line width=.8pt] (lec) to node [above] {!} (lebq);
\draw[double,double distance=1.2pt,line width=.5pt] (lec) to node [left] {\rule{0px}{15px}\footnotesize Thm~\ref{thm:sim_c-equiv-b}} (leb);
\draw[double,double distance=1.2pt,line width=.5pt] (leb) to node [left] {\footnotesize Thm~\ref{theorem:cand-eq-le-b}} (cand);
\draw[double,double distance=1.2pt,line width=.5pt] (leb) to node [right] {\rule{0px}{15px}\footnotesize Thm~\ref{thm:llc-Q-equivalence}} (lebq);
\draw[double,double distance=1.2pt,line width=.5pt] (lebq) to node [above] {\footnotesize Thm~\ref{thm:llc-Q-equivalence}} (lecq);
\end{tikzpicture}
\end{minipage}~%
\begin{minipage}{.46\textwidth}
\begin{tabular}{@{}l@{~}p{.8\textwidth}@{}} 
$\lec_\LCC$             & \footnotesize contextual preorder in $\LCC$
\\
$\leb_\LCC^o$           & \footnotesize open extension of similarity in $\LCC$
\\
$\cand$                 & \footnotesize co-inductively defined candidate relation for Howe's technique
\\
$\leb_{\LCC,Q_\LCC}^o$  & \footnotesize open extension of ${\cal Q}$-similarity in $\LLCC$ with ${\cal Q} = Q_\LCC$
\\
$\lec_{\LCC,Q_\LCC}^o$  & \footnotesize open extension of contextual preorder in $\LLCC$ restricted to contexts $Q_\LCC$
\end{tabular}
\end{minipage}
\caption{Structure of soundness and completeness proofs for similarities in $\LLCC$. The ${=}!{=}$ indicates a required equality which can only be proved via Howe's technique.\label{fig-bc-equiv-lcc}
}
\end{figure*}

\vspace{1mm}

{\bf Notation.}~In abuse of notation we use higher order abstract syntax as 
\eg\ in \cite{howe:89} for the proof and write $\tauop(..)$ for an expression
with top operator $\tauop$,  which may be all possible term constructors, like
\tcase, application, a constructor, \tseq, or $\lambda$, and  $\theta$ for an
operator that may be the head of a value, \ie\ a constructor or $\lambda$. 
\begin{definition} 
For a relation $\eta$ on closed $\LAMBDAEXPR$-expressions $\eta^o$ is the open
extension on $\LLAZYCC$:
For (open) $\LAMBDAEXPR$-expressions $s_1,s_2$, the relation $s_1~\eta^o ~s_2$ 
holds, if for all substitutions $\sigma$ such that $\sigma(s_1), \sigma(s_2)$
are closed, the relation  $\sigma(s_1)~ \eta ~\sigma(s_2)$ holds.
Conversely, for binary relations $\mu$ on open expressions,
$\closed{\mu}$ is the restriction to closed expressions.

We say a binary relation $\mu$ is {\em operator-respecting}, iff
$s_i\ \mu\ t_i$ for $i = 1,\ldots,n$ implies 
$\tauop(s_1,\ldots,s_n)\ \mu\ \tauop(t_1,\ldots,t_n)$.
\end{definition}
Note that $\tauop$ and $\theta$ may also represent the binding $\lambda$ using $\lambda(x.s)$ as  representing $\lambda x.s$. 
For consistency of terminology and treatment with that in other papers such as \cite{howe:89}, we assume
that removing the top constructor $\lambda x$ in relations is done after a renaming. 
For example, $\lambda x.s\ \mu\ \lambda y.t$ is renamed before further treatment
to $\lambda z.s[z/x]\ \mu\ \lambda z.t[z/y]$ for a fresh variable $z$. 

\vspace{1mm}

{\bf Plan of Subsection \ref{subsec:bisimulation-llc}.}~We start by explaining the subgoals
 of the soundness and completeness proofs for similarities in $\LLCC$ and its structure, 
illustrated in  \FIGURE~\ref{fig-bc-equiv-lcc}.
The main result we want to show is that contextual preorder $\lec_\LCC$ and $\leb_{\LCC,Q_\LCC}^o$ coincide, where 
$\leb_{\LCC,Q_\LCC}^o$ is the open extension of $\leb_{\LCC,Q_\LCC}$,
and $\leb_{\LCC,Q_\LCC}$ is ${\cal Q}$-similarity introduced in 
 Definition~\ref{def:Q-gfp-preorder} instantiated with the subcalculus of $\LLCC$ which consists of
{\em closed} expressions, {\em closed} contexts, and {\em closed} answers, and $Q_\LCC$ is a specific set of small closed $\LLCC$-contexts.
${\cal Q}$-similarity does not allow a direct proof of soundness and completeness for contextual equivalence
using Howe's method \cite{howe:89,howe:96}, since it is not stated in terms of the syntactic form of values derived by
evaluation. We overcome this obstacle by defining another similarity $\leb_\LCC$ in $\LLAZYCC$ for which we will 
perform the proof of soundness and completeness \wrt\ contextual preorder. 
Since the definition of $\leb_\LCC$ does not obviously imply that $\leb_\LCC$ is a precongruence, a candidate relation $\cand$ is defined,
which is trivially compatible with contexts, but needs to be shown to be transitive.  
After proving $\cand = {\leb_\LCC^o}$, \ie\ that ${\leb_\LCC^o}$ is a precongruence, soundness of $\leb_\LCC^o$ 
\wrt\ contextual preorder $\lec_\LCC$ follows.
Completeness can then also be proven. 
In a second step we prove that $\leb_{\LCC,Q_\LCC}^o$ is sound  and complete  for contextual equivalence,  \ie ${\lec_\LCC} = {\leb_{\LCC,Q_\LCC}^o}$.
After showing that $\LLCC$ is convergence-admissible we are also able to show that the inductive description $\lec_{\LCC,Q_\LCC}$ of 
${\cal Q}$-similarity coincides with $\leb_{\LCC,Q_\LCC}$.

Another obstacle is that the contextual preorder contains the irregularity $\lambda x.\Omega \lec_\LCC c~s_1 \ldots s_n$ for any constructor $c$.
This requires an adapted definition of the similarity relation, and a slightly modified proof route.

In the following let $\cBot$ be the set of $\LAMBDAEXPR$-expressions $s$ with the property that for all 
$\LAMBDAEXPR$-substitutions $\sigma$:
if $\sigma(s)$ is closed, then $\sigma(s) \DIV_\LCC$. 
That $\lambda x.s \lec_\LCC (c~s_1 \ldots s_n)$  indeed holds is shown in Proposition \ref{prop-classification-lcc}. 
Now we define an applicative similarity $\leb_{\LCC}$ in $\LLAZYCC$ analogous to \cite{howe:89,howe:96},  
where this irregularity is taken into account.

\begin{definition}[Similarity  in $\LLAZYCC$]\label{def:lazycc-le-b} 
Let $\eta$  be a binary relation on closed $\LAMBDAEXPR$-expressions. 
Let $F_{\LCC}$ be the following operator on relations on closed $\LAMBDAEXPR$-expressions:
\\
$s~F_{\LCC}(\eta)~t$ holds iff 
\begin{itemize}
  \item $s \maycon_{\LCC} \lambda x.s' \implies$  $\big($~$t \maycon_{\LCC} \lambda x.t'$ and $s'\ \eta^o~t'$, or \\
      \hspace*{2.8cm} $t \maycon_{\LCC} (c~t_1' \ldots t_n')$  and  $s' \in \cBot \big)$  
  \item  $s \maycon_{\LCC} (c~s_1' \ldots s_n')  \implies$    $\big($~$t \maycon_{\LCC} (c~t_1' \ldots t_n')$  
     and the relation  $s_i'~\eta~t_i'$  holds for all $i\big)$ 
\end{itemize}

{\em Similarity} $\leb_{\LCC}$ is defined as the greatest fixpoint of the operator $F_{\LCC}$.
Bisimilarity $\simb_{\LCC}$ is defined as $s \simb_{\LCC} t$ iff $s \leb_{\LCC} t \wedge t \leb_{\LCC} s$. 
\end{definition} 

Note that the operator $F_{\LCC}$ is monotone, hence the greatest fixpoint  $\leb_{\LCC}$ exists. 

\subsubsection{\texorpdfstring{Similarity and Contextual Preorder Coincide in $\LLCC$}{Similarity and Contextual Preorder Coincide in LCC}}
Although applying Howe's proof technique is standard,  for the sake of completeness, and to demonstrate that the irregularity 
$\lambda x.\Omega \le_\LCC (c~s_1 \ldots s_n)$
can  also be treated, we will explicitly show in this section that $\leb_{\LCC}^o ~=~\lec_{\LCC}$
using Howe's method \cite{howe:89,howe:96}. 
 
\begin{lemma}\label{lem:open-closed-rel-general}
For a relation $\eta$ on closed expressions it holds
$\closed{\open{\eta}} = \eta$, and also $s\  \eta^o\  t$ implies $\sigma(s)~ \eta^o\  \sigma(t)$ for any substitution $\sigma$. 
For a relation $\mu$ on open expressions, $\mu\ \subseteq  \open{\closed{\mu}}$ is equivalent to
$s~ \mu ~t \implies \sigma(s)~\closed{\mu}~\sigma(t)$ for all closing   substitutions $\sigma$.\qed
\end{lemma}

\begin{proposition}[Co-Induction]\label{prop-coinduction}
The principle of co-induction for the greatest fixpoint of $F_{\LCC}$  shows that for every relation $\eta$ on 
closed expressions with $\eta \subseteq F_{\LCC} (\eta)$, 
we derive $\eta\ \subseteq\ \leb_{\LCC}$.  This obviously also implies $\open{\eta}\ \subseteq\ \open{\leb_{\LCC}}$.\qed
\end{proposition}

The fixpoint property of $\leb_{\LCC}$  implies:
\begin{lemma}\label{lem:fixpoint-constructors} 
For a closed value  $\theta_1(s_1, \ldots, s_n)$, and a closed term $t$ with
    $\theta_1(s_1, \ldots, s_n) \leb_{\LCC}  t$, we have  $t \maycon_\LCC~\theta_2(t_1, \ldots, t_n)$,
and there are two cases: 
\begin{enumerate}
 \item $\theta_1 = \theta_2$ are constructors or $\lambda$  and  $s_i \leb_{\LCC}^o t_i$ for all $i$.  
 \item $\theta_1(s_1, \ldots, s_n) = \lambda(x.s)$ with $s \in \cBot$ and $\theta_2$ is a constructor.  
 \end{enumerate}
\end{lemma}

\begin{lemma}\label{lem:fixpoint-cBot}
For  two expressions $s,t$: $s \in \cBot$ implies $s  \leb_{\LCC}^o  t$. 
Thus any two expressions $s,t \in \cBot$ are bisimilar: %equivalent:  
$s   \simb_{\LCC}^o t$. \qed
\end{lemma}

Particular expressions in $\cBot$ are \mbox{$(\tcase~(\lambda x.s)~alts)$}  and \mbox{$(c(s_1,\ldots,s_n)~a_1 \ldots a_m)$} for \mbox{$m \ge 1$};
also $s \in \cBot$ implies that $(s~t)$, $(\tseq~s~t)$, $(\tcase~s~alts)$ and $\sigma(s)$  are also in $\cBot$. 

\begin{lemma}\label{lem:leq-b-refl-trans} The relations $\leb_{\LCC}$  and $\leb_{\LCC}^o$  are reflexive and transitive.
\end{lemma}
\begin{proof} 
Reflexivity follows by showing that $\eta := ~\leb_{\LCC} \cup\ \{(s,s)~|~s \in \LAMBDAEXPR, s~\text{closed}\}$  
satisfies $\eta \subseteq~F_{\LCC}(\eta)$.
Transitivity follows by showing that $\eta := ~\leb_{\LCC} \cup\ (\leb_{\LCC} \circ \leb_{\LCC})$  
satisfies $\eta \subseteq~F_{\LCC}(\eta)$ and then using the co-induction principle.
\end{proof}

The goal in the following is to show that $\leb_{\LCC}$ is a precongruence. 
This proof proceeds by defining a precongruence candidate $\cand$  as a closure of $\leb_{\LCC}$ 
within contexts, which obviously is operator-respecting, 
but transitivity needs to be shown.
By proving that $\leb_{\LCC}^o$ and $\cand$ coincide, on the one hand transitivity of $\cand$ follows (since $\leb_{\LCC}$ is transitive)
and on the other hand (and more importantly) it follows that $\leb_{\LCC}^o$ is operator-respecting (since $\cand$ is operator-respecting) 
and thus a precongruence.

\begin{definition}\label{def:cand} 
The precongruence candidate $\cand$ is a binary relation on 
open expressions and is defined as the greatest fixpoint of the monotone operator $\Fcand$ on relations on all expressions:
\begin{enumerate}
\item $x~\Fcand(\eta)~s$ iff  $x \leb_{\LCC}^o s$.
\item $\tauop(s_1, \ldots, s_n)~  \Fcand(\eta)~s$ iff there is some expression 
 $\tauop(s_1', \ldots, s_n') \leb_{\LCC}^o s$ with $s_i~ \eta~ s_i'$ for $i = 1,\ldots,n$.
\end{enumerate}
\end{definition}

\begin{lemma} If some relation $\eta$ satisfies ${\eta} \subseteq {\Fcand(\eta)}$, 
then ${\eta} \subseteq {\cand}$.\qed
\end{lemma}
\noindent Since $\cand$ is a fixpoint of $\Fcand$, we have:
\begin{lemma}\label{lemma:cand-fixpoint}\quad
\begin{enumerate}
\item $x \cand s$ iff  $x \leb_{\LCC}^o s$.
\item $\tauop(s_1, \ldots, s_n)  \cand s$ iff there is some expression $\tauop(s_1', \ldots, s_n') \leb_{\LCC}^o s$ with $s_i \cand s_i'$ 
  for $i = 1,\ldots,n$.\qed
\end{enumerate}
\end{lemma}
\noindent Some technical facts  about the precongruence candidate are now proved:
\begin{lemma}
  \label{lem:howe-standard} The following properties hold:
  \begin{enumerate}
       \setlength{\itemsep}{1mm}%
  \item \label{lem:howe-standard-reflexive}
    $\cand$ is reflexive.   
  \item \label{lem:howe-standard-operator_respecting}
    $\cand$ and $\candc$ are operator-respecting.  
  \item \label{lem:howe-standard-subsumes_preorder}
    ${\leb_{\LCC}^o} \subseteq {\cand}$ and ${\leb_{\LCC}} \subseteq {\candc}$.
  \item \label{lem:howe-standard-precon_cand:composition}
    ${{\cand} \circ {\leb_{\LCC}^o}} \subseteq {\cand}$. 
  \item\label{lem:howe-standard-precon_cand:subst}
    $({s \cand s'} \wedge {t \cand t'}) \implies   t[s/x] \cand t'[s'/x]$.  
  \item\label{lem:howe-standard-cand-subs} $s \cand t$ implies that $\sigma(s) \cand \sigma(t)$ for every substitution $\sigma$.
  \item\label{lem:howe-standard-subset-open-closed} ${\cand} \subseteq {\open{\candc}}$
  \end{enumerate}
\end{lemma}
\begin{proof} 
Parts \eqref{lem:howe-standard-reflexive} -- \eqref{lem:howe-standard-subsumes_preorder} can be shown by
structural induction and using reflexivity of  $\leb_{\LCC}^o$.
Part~\eqref{lem:howe-standard-precon_cand:composition} follows from the definition, Lemma \ref{lemma:cand-fixpoint},  and transitivity of $\leb_{\LCC}^o$.

For part~\eqref{lem:howe-standard-precon_cand:subst} let $\eta := {{\cand} \cup {\{(r[s/x],r'[s'/x])~|~r \cand r'\}}}$. 
Using co-induction it suffices to show that ${\eta} \subseteq {\Fcand(\eta)}$: In the case $x \cand r'$, we obtain $x \leb_{\LCC}^o r'$ from the definition, and $s' \leb_{\LCC}^o r'[s'/x]$
and thus $x[s/x] \cand r'[s'/x]$. In the case $y \cand r$, we obtain $y \leb_{\LCC}^o r'$ from the definition, and $y[s/x] = y \leb_{\LCC}^o r'[s'/x]$
and thus $y = y[s/x] \cand r'[s'/x]$. 
If $r = \tauop(r_1,\ldots,r_n)$, $r \cand r'$ 
and $r[s/x]~\eta ~r'[s'/x]$, then there is some $\tauop(r_1',\ldots,r_n') \leb_{\LCC}^o r'$ with  $r_i \cand r_i'$.
W.l.o.g. bound variables have fresh names.  We have  $r_i[s/x] ~\eta ~r_i'[s'/x]$ and $\tauop(r_1',\ldots,r_n')[s'/x] \leb_{\LCC}^o r'[s'/x]$.
Thus $r[s/x]~\Fcand(\eta) ~r'[s'/x]$.

Part~\eqref{lem:howe-standard-cand-subs} follows from item \eqref{lem:howe-standard-precon_cand:subst}.
Part~\eqref{lem:howe-standard-subset-open-closed} follows from item \eqref{lem:howe-standard-cand-subs} and Lemma \ref{lem:open-closed-rel-general}.\qedhere
\end{proof}

\begin{lemma} \label{lemma:lcs:middle-terms-closed} The middle expression in the definition of $\cand$ can  be chosen to be closed if $s,t$ are closed:
  Let  $s = \tauop(s_1,\ldots,s_{\ari(\tauop)})$, such that  $s \cand t$ holds. 
  Then there are operands  $s_i'$, such that $\tauop(s'_1,\ldots,s'_{\ari(\tauop)})$ is closed,  $\forall i: s_i ~\cand ~ s_i'$
  and $\tauop(s'_1,\ldots,s'_{\ari(\tauop)}) \leb_{\LCC}^o\ t$.
\end{lemma}
 \begin{proof} The definition of $\cand$ implies that there is an expression $\tauop(s_1'',\ldots,s''_{\ari(\tauop)})$ such that
  $s_i  \cand\ s_i''$ for all $i$ and   $\tauop(s_1'',\ldots,s''_{\ari(\tauop)})\ \leb_{\LCC}^o\ t$. 
   Let $\sigma$ be the substitution with $\sigma(x) := r_x$ for all $x \in \FV(\tauop(s_1'',\ldots,s''_{\ari(\tauop)}))$, 
where $r_x$ is any closed expression.  Lemma \ref{lem:howe-standard} now shows that $s_i = \sigma(s_i)  \cand  \sigma(s_i'')$ holds for all $i$.  
  The relation $\sigma(\tauop(s_1'',\ldots,s''_{\ari(\tauop)}))\ \leb_{\LCC}^o t$ holds, since $t$ is closed and due to the definition of an open extension. 
The requested expression is  $\tauop(\sigma(s_1''),\ldots,\sigma(s''_{\ari(\tauop)}))$. 
\end{proof}

\noindent Since reduction $\xrightarrow{\LCC}$ is deterministic:
\begin{lemma}\label{lemma:reduction-stable} If $s   \xrightarrow{\LCC} s'$, then $s' \leb_{\LCC}^o s$ and  $s  \leb_{\LCC}^o s'$.\qed
\end{lemma}

Lemmas \ref{lemma:reduction-stable} and \ref{lem:howe-standard} imply that $\cand$  is right-stable \wrt\ reduction:
\begin{lemma}\label{lemma:rightstable} If $s  \cand t$ and $t \xrightarrow{\LCC} t'$, then $s  \cand t'$.\qed
\end{lemma}

We show that $\cand$ is left-stable \wrt\ reduction:
\begin{lemma}\label{lem:left-value} Let $s,t$ be closed expressions such that $s = \theta(s_1,\ldots,s_n)$ is a value and $s \cand t$. 
Then there are two possibilities:
\begin{enumerate}
  \item\label{lem:left-value-1}  $s = \lambda x.s_1$ and $t \maycon_{\LCC} c(t_1,\ldots,t_n)$,  where $c$ is a constructor; 
  \item \label{lem:left-value-2} there is some closed value $t'= \theta(t_1,\ldots,t_n)$ with $t \xrightarrow{\LCC,*} t'$ and for all $i: s_i \cand t_i$.   
\end{enumerate}
\end{lemma}
\begin{proof}  
The definition of $\cand$ implies that there is a closed expression $\theta(t_1',\ldots,t_n')$ with $s_i \cand t_i'$ for all $i$ and 
$\theta(t_1',\ldots,t_n') \leb_{\LCC} t$. Lemma \ref{lem:fixpoint-constructors}  implies  that  $t\maycon_\LCC $,
hence  either  $t \xrightarrow{\LCC,*}  c(t_1'',\ldots,t_n'')$ or 
$t \xrightarrow{\LCC,*}   \lambda x.t_1''$. 
\begin{itemize}
\item First let  $\theta = \lambda$.  
The case  $t \xrightarrow{\LCC,*}  c(t_1'',\ldots,t_n'')$ is possibility \eqref{lem:left-value-1} of the lemma. \\
In the second case, $t \xrightarrow{\LCC,*}   \lambda x.t_1''$, 
Lemma~\ref{lemma:rightstable}
 implies $\lambda x.s_1 \cand \lambda x.t_1''$. Definition of $\cand$ and Lemma \ref{lemma:lcs:middle-terms-closed} 
now show that there is some closed  $\lambda x.t_1'''$ 
with $s_1 \cand t_1'''$ and $\lambda x.t_1''' \leb_{\LCC}   \lambda x.t_1''$. 
The latter relation implies \mbox{$t_1''' \leb_{\LCC}^o t_1''$}, which shows
$s_1' \cand t_1''$ by \mbox{Lemma \ref{lem:howe-standard} \eqref{lem:howe-standard-precon_cand:composition}.} 
\item ~If $\theta$ is a constructor, then there is a closed expression $\theta(t_1',\ldots,t_n')$ with $s_i \cand t_i'$ for all $i$ and 
$\theta(t_1',\ldots,t_n') \leb_{\LCC} t$. 
The properties of $\leb_{\LCC}$ imply that $t \xrightarrow{\LCC,*}   \theta(t_1'',\ldots,t_n'')$ with $t_i' \leb_{\LCC} t_i''$ for all $i$. 
By definition of $\cand$ and \mbox{Lemma \ref{lem:howe-standard} \eqref{lem:howe-standard-precon_cand:composition}}, we obtain $s_i \cand t_i''$ for all $i$.
\qedhere
\end{itemize}
\end{proof}

\begin{proposition}\label{prop:stable-base-case} Let $s,t$ be closed expressions, $s \cand t$  and $s \xrightarrow{\LCC} s'$ where $s$ is the redex. 
Then $s' \cand t$.
\end{proposition}
 \begin{proof}
 The relation $s \cand t$  implies that 
$s = \tauop(s_1,\ldots,s_n)$ and by Lemma~\ref{lemma:lcs:middle-terms-closed} there is some closed $t' = \tauop(t_1',\ldots,t_n')$ 
with $s_i \cand t_i'$ for all $i$ and $t'~ \leb_{\LCC}^o~t$.  
\begin{itemize}
\item For the (\rnbeta)-reduction, $s = (s_1~s_2)$, where $s_1 = (\lambda x.s_1')$,  $s_2$ is a closed term, and $t' = (t_1'~t_2')$. 
 The relation $(\lambda x.s_1') = s_1 \cand t_1'$ implies that there exists a closed expression $\lambda x.t_1'' \leb_{\LCC}^o t_1'$ with $s_1' \cand t_1''$.\\
 $\circ$~The first case is   $t_1' \xrightarrow{\LCC,*} c(\ldots)$ and $t_1'' \in \cBot$.  Lemma \ref{lem:howe-standard} implies $\lambda x.s_1' \cand \lambda x.t_1''$,
 and again by Lemma \ref{lem:howe-standard}, we derive $s_1'[s_2/x] \cand t_1''[s_2/x]$, where $t_1''[s_2/x] \in \cBot$. 
  Then  $t_1''[s_2/x]   \leb_{\LCC}^o t$ by Lemma \ref{lem:fixpoint-cBot}, which implies $s_1'[s_2/x] \cand t$. 
  Since $s \xrightarrow{\LCC} s_1'[s_2/x]$, the lemma is proven for this case. \\
 $\circ$~The second case is  $t_1' \xrightarrow{\LCC,*} \lambda x.t_1'''$   with $t_1'' \leb_{\LCC}^o  t_1'''$. 
 We also obtain $\lambda x.t_1'' \leb_{\LCC}^o  \lambda x. t_1'''$, and by the properties of  $\leb_{\LCC}^o$ \wrt\ reduction, 
   also $t_1''[t_2'/x] \leb_{\LCC}^o t_1'''[t_2'/x]$.
 From $t' \xrightarrow{\LCC,*} t_1'''[t_2'/x]$ we obtain \mbox{$t_1'''[t_2'/x] \leb_{\LCC} t$}. 
Lemma \ref{lem:howe-standard} and transitivity of $\leb_{\LCC}$ now show $s_1'[s_2/x] \cand t_1''[t_2'/x]$.  Hence   $s_1'[s_2/x] \cand t$, again using Lemma \ref{lem:howe-standard}.

\item Similar arguments as for the second case apply to the case-reduction.
\item Suppose, the reduction is a $(\rnseq)$-reduction. Then $s \cand t$ and $s = (\tseq~ s_1~s_2)$. Lemma \ref{lemma:lcs:middle-terms-closed}  implies
 that there is some closed $(\tseq~ t_1'~t_2') \leb_{\LCC}^o t$ with $s_i \cand t_i'$. Since $s_1$ is a value, Lemma \ref{lem:left-value} shows that
 there is a reduction 
 $t_1' \xrightarrow{\LCC,*} t_1''$, where $t_1''$ is a value.  There are the  reductions $s \xrightarrow{\LCC} s_2$ and 
 $(\tseq~ t_1'~t_2') \xrightarrow{\LCC,*} (\tseq~ t_1''~t_2') \xrightarrow{\LCC} t_2'$.
 Since $t_2' \leb_{\LCC}^o (\tseq~ t_1'~t_2') \leb_{\LCC}^o t$, and  $s_2 \cand t_2'$, we obtain $s_2 \cand t$.
\qedhere
\end{itemize}
\end{proof}

\begin{proposition}\label{prop:stable-in-surface}  Let $s,t$ be closed expressions, $s \cand t$  and $s \xrightarrow{\LCC} s'$. 
Then $s' \cand t$.
\end{proposition}
 \begin{proof} We use induction on the length of the path to the redex. 
  The base case is proven in Proposition \ref{prop:stable-base-case}.
  Let $R[s],t$ be closed, $R[s] \cand t$  and $R[s] \xrightarrow{\LCC}~R[s']$, 
  where we assume that the redex $s$ is not at the top level
  and that $R$ is an $\LLCC$-reduction context.
  The relation $R[s] \cand t$  implies that 
 $R[s] = \tauop(s_1,\ldots,s_n)$ and that there is some closed  expression $t'$ with $t' = \tauop(t_1',\ldots,t_n') \leb_{\LCC}^o~t$ with $s_i \cand t_i'$ for all $i$.
If $s_j \xrightarrow{\LCC} s_j'$, then by induction hypothesis $s_j' \cand t_j'$. 
Since $\cand$ is operator-respecting, we also obtain  $R[s'] = \tauop(s_1,\ldots,s_{j-1},s_j',s_{j+1},\ldots,s_n)$ $\cand$ 
$\tauop(t_1',\ldots,t_{j-1}',t_j',t_{j+1}',\ldots,t_n')$, and from $\tauop(t_1',\ldots,t_n') \leb_{\LCC}^o~t$ we have 
$R[s'] = \tauop(s_1,\ldots,s_{j-1},s_j',s_{j+1},\ldots,s_n) \cand t$.
\end{proof}

\begin{lemma}\label{lemma:s-cand-cbot}
If $\lambda x. s,\lambda x. t$ are closed, $\lambda x. s \cand \lambda x. t$, and $t \in \cBot$, then also $s \in \cBot$.
\end{lemma}
\begin{proof}
For any closed $r$, we also have $(\lambda x. s)~r \cand (\lambda x. t)~r$, since $\cand$ is operator-respecting. From $t \in \cBot$, 
we obtain that $((\lambda x. t)~r)\DIV_\LCC$. Now suppose that  $(\lambda x. s)~r \xrightarrow{\LCC,*} s'$, where $s'$ is a value. 
Lemma  \ref{prop:stable-in-surface} implies that $s' \cand (\lambda x. t)~r$. Now Lemma \ref{lem:left-value} shows that this is impossible.
Hence $s \in \cBot$.
\end{proof}
\noindent Now we can prove an improvement of Lemma \ref{lem:left-value}:

\begin{lemma}\label{lem:left-value-full} Let $s,t$ be closed expressions such that $s = \theta(s_1,\ldots,s_n)$ is a value and $s \cand t$. 
Then there are two possibilities:
\begin{enumerate}
  \item\label{lem:left-value-full-1}  $s = \lambda x.s_1$, $t \maycon_{\LCC} c(t_1,\ldots,t_n)$  where $c$ is a constructor,
    and $s_1 \in \cBot$. 
  \item \label{lem:left-value-full2} there is some closed value $t'= \theta(t_1,\ldots,t_n)$ with $t \xrightarrow{\LCC,*} t'$ and for all $i: s_i \cand t_i$.   
\end{enumerate}
\end{lemma}
\begin{proof} This follows from Lemma \ref{lem:left-value} and Lemma \ref{lemma:s-cand-cbot}.
\end{proof}
Now we are ready to prove that the precongruence candidate and similarity coincide. 
 
\begin{theorem}\label{theorem:cand-eq-le-b} 
$\candc~=~\leb_{\LCC}$ and $\cand\ =\ \leb_{\LCC}^o$.
\end{theorem}
\begin{proof}
Since ${\leb_{\LCC}}\subseteq{\candc}$ by Lemma \ref{lem:howe-standard}, we  have to show that ${\candc} \subseteq {\leb_{\LCC}}$.
Therefore it is sufficient to show that $\candc$ satisfies the fixpoint equation for $\leb_{\LCC}$.
We show that ${\candc} \subseteq {F_{\LCC}(\candc)}$.
Let $s\ \candc\ t$ for  closed terms $s,t$. We show that $s\ F_{\LCC}(\candc)\ t$: 
If $s\Uparrow_{\LCC}$, 
then $s~ F_{\LCC}(\candc)~t$ holds by Definition  \ref{def:lazycc-le-b}.    
If $s \maycon_{\LCC} \theta(s_1,\ldots,s_n)$, then $\theta(s_1,\ldots,s_n)~\candc~ t$ by Proposition \ref{prop:stable-in-surface}.

Lemmas  \ref{prop:stable-in-surface}  and \ref{lem:left-value-full} show that there are two possibilities:
\begin{itemize}
  \item $t \xrightarrow{\LCC,*}  c(t_1,\ldots,t_n)$ for a constructor $c$,   $s \maycon_{\LCC} \lambda x.s_1$,
      and  $s_1 \in \cBot$.
  \item $t \xrightarrow{\LCC,*}  \theta(t_1,\ldots,t_n)$ and for all $i: s_i \cand t_i$. 
\end{itemize} 
This implies $s~ F_{\LCC}(\candc)~t$.  
Thus the fixpoint property of $\candc$ \wrt\ $F_{\LCC}$ holds, and hence $\candc~ =~ \leb_{\LCC}$.

\vspace{1mm}
Now we prove the second part. The  first part, ${\candc} \subseteq {\leb_{\LCC}}$, 
implies  ${\open{\candc}} \subseteq {\leb_{\LCC}^o}$ by monotonicity.   
     Lemma \ref{lem:howe-standard}  \eqref{lem:howe-standard-subset-open-closed} implies 
    ${\cand} \subseteq {\open{\candc}} \subseteq {\leb_{\LCC}^o}$.  
The other direction is proven in Lemma \ref{lem:howe-standard}  \eqref{lem:howe-standard-subsumes_preorder}.
\end{proof}

Since $\leb_{\LCC}^o$ is reflexive and transitive (Lemma~\ref{lem:leq-b-refl-trans}) and $\candc$ is
operator-respecting (Lemma~\ref{lem:howe-standard}~\eqref{lem:howe-standard-operator_respecting}), 
this immediately implies: 
\begin{corollary}  
$\leb_{\LCC}^o$ is a precongruence on expressions $\LAMBDAEXPR$.
If $\sigma$ is a substitution, then 
$s \leb_{\LCC}^o t$ implies $\sigma(s) \leb_{\LCC}^o \sigma(t)$. \qed
\end{corollary}

\begin{lemma}\label{lem:PFI-leb-impl-lec}  $\leb_{\LCC}^o~\subseteq~\lec_{\LCC}$.
\end{lemma}
\begin{proof}
Let $s,t$ be expressions with $s \leb_{\LCC}^o t$ such that $\Ctxt[s]\maycon_{\LCC}$. 
Let $\sigma$ be a substitution that replaces all 
free variables of $C[s], C[t]$ by $\Omega$.
The properties of the call-by-name reduction show that also $\sigma(C[s])\maycon_{\LCC}$. 
Since $\sigma(C[s]) = \sigma(C)[\sigma(s)]$,
$\sigma(C[t]) = \sigma(C)[\sigma(t)]$ and since
$\sigma(s)  \leb_{\LCC}^o \sigma(t)$, we obtain from the precongruence property of 
$\leb_{\LCC}^o$ that  also $\sigma(C[s]) \leb_{\LCC} \sigma(C[t])$.
Hence $\sigma(C[t])\maycon_{\LCC}$. This is equivalent to  $C[t]\maycon_{\LCC}$, 
since free variables are replaced by $\Omega$, and thus they cannot overlap with redexes.
Hence $\leb_{\LCC}^o ~\subseteq~ \lec_{\LCC}$. 
\end{proof}

\begin{corollary}\label{cor:reductioncorrect}\label{cor:lcc-red-rules-correct}
  $s \xrightarrow{\LCC} s'$ implies $s \simc_{\LCC} s'$.  
Thus the reduction rules of the calculus $\LLAZYCC$ are correct \wrt\ $\simc_{\LCC}$ in any context.
\end{corollary}
\begin{proof}
This follows from Lemmas \ref{lemma:reduction-stable}  and \ref{lem:PFI-leb-impl-lec}.
\end{proof}
Now we show a characterization for $\LAMBDAEXPR$-expressions, which includes the previously mentioned irregularity of $\lec_\LCC$: 
\begin{proposition}\label{prop-classification-lcc}
Let $s$ be a closed $\LLAZYCC$-expression. 
Then there are three cases: 
 $s \simc_{\LCC} \Omega$, 
 $s \simc_{\LCC} \lambda x.s'$ for some $s'$,  
 $s \simc_{\LCC} (c~s_1 \ldots s_n)$
for some terms $s_1,\ldots,s_n$ and constructor $c$. 

For two closed $\LLAZYCC$-expressions $s,t$ with $s \lec_{\LCC} t$: Either $s \simc_{\LCC} \Omega$, or $s \simc_{\LCC} (c~s_1 \ldots s_n)$, 
$t \simc_\LCC (c~t_1 \ldots t_n)$ and $s_i \lec_{\LCC} t_i$ for all $i$ for some terms 
  $s_1,\ldots,s_n,t_1,\ldots,t_n$ and constructor $c$,
or  $s \simc_\LCC \lambda x.s'$ and $t  \simc_\LCC \lambda x.t'$ for some expressions $s',t'$ with $s' \lec_{\LCC}^o t'$, or  
$s \simc_\LCC \lambda x.s'$ and $t  \simc_\LCC  (c~t_1 \ldots t_n)$ 
for some term $s' \in \cBot$, expressions $t_1,\ldots,t_n$ and constructor $c$. 
\end{proposition} 
\begin{proof} We apply Lemma \ref{lem:PFI-leb-impl-lec}. Corollary \ref{cor:reductioncorrect} then shows that using reduction the classification 
of closed expressions into the 
classes \wrt\ $\simc_{\LCC}$ holds.

For two closed $\LLAZYCC$-expressions $s,t$ with $s \leb_{\LCC} t$: we obtain the classification in the lemma but with $\leb_{\LCC}$ instead of 
$\lec_{\LCC}$. For the three cases $s \simb_\LCC \Omega$,   both $s,t$ are equivalent to constructor expressions, and both
$s,t$ are equivalent to abstractions, we obtain also that  $s \lec_{\LCC} t$. 
In the last case $\lambda x.s' \leb_\LCC (c~s_1 \ldots s_n)$, we also obtain from the $\leb_\LCC$-definition, that it is valid and from 
Lemma \ref{lem:PFI-leb-impl-lec}, that it implies $\lambda x.s' \lec_\LCC (c~s_1 \ldots s_n)$.
Other combinations of constructor applications, abstractions and $\Omega$ cannot be in $\lec_\LCC$-relation:
\begin{itemize}
 \item $(c~t_1 \ldots t_n) \not\lec_\LCC \Omega$ and $\lambda x.s \not\lec_\LCC \Omega$ since the empty context distinguishes them.
 \item \mbox{$(c_1~s_1\ldots s_n) \not \lec_\LCC (c_2~t_1\ldots t_m)$:} Let
$C:=\tcase_T~[\cdot]~(c_1~x_1\ldots x_n \to \lambda y.y)~alts$ where all al\-ter\-natives in $alts$ have right hand side $\Omega$.
Then $C[(c_1~s_1\ldots s_n)]\maycon_\LCC$ but $C[(c_2~t_1\ldots t_m)]\DIV_\LCC$.
\item $(c~s_1\ldots s_n) \not\lec_\LCC (c~t_1\ldots t_n)$ if $s_i \not\lec_\LCC t_i$: Let context $D$ be the witness for $s_i \not\lec_\LCC t_i$.
Then $C = \tcase_T~[\cdot]~(c~x_1\ldots x_n \to D[x_i])$ distinguishes $(c~s_1\ldots s_n)$ and $(c~t_1\ldots t_n)$ 
\item $(c~s_1 \ldots s_n) \not\lec_\LCC (\lambda x.t)$: The context $\tcase_T~[\cdot]~(c~x_1\ldots~x_n \to \lambda y.y)~alts$ is a witness.
\item $\lambda x.s \not\lec_\LCC \lambda x.t$ if $s \not\lec_\LCC t$: Let $D$ be the witness for $s \not\lec_\LCC t$.
Then $C=D[([\cdot]~x)]$ distinguishes $\lambda x.s$ and $\lambda x.t$. 
\item $\lambda x.s \not\lec_\LCC (c~t_1\ldots t_n)$ if $s \not\in\cBot$: Since $s \not\in \cBot$ and $FV(s) \subseteq \{x\}$,
there exists a closing substitution $\sigma = \{x \mapsto r\}$ such that $\sigma(s)\maycon_\LCC$.
For the context $C=([\cdot]~r)$ the expression $C[\lambda x.s]$ converges while $C[(c~t_1\ldots t_n)]$ diverges.\qedhere
\end{itemize}
\end{proof}

\begin{lemma}\label{lem:PFI-lec-impl-leb}
 $\lec_{\LCC}~\subseteq~\leb_{\LCC}^o$.
\end{lemma}
\begin{proof}
The relation  $\lec_{\LCC}^c$ satisfies the fixpoint condition, \ie\ $\lec_{\LCC}^c\ \subseteq\ F_{\LCC}(\lec_{\LCC}^c)$,
which follows from Corollary \ref{cor:reductioncorrect} and Proposition \ref{prop-classification-lcc}.
\end{proof}
\noindent Lemmas~\ref{lem:PFI-leb-impl-lec} and \ref{lem:PFI-lec-impl-leb} immediately imply:

\begin{theorem}\label{thm:sim_c-equiv-b}
   $\leb_{\LCC}^o~~= ~~\lec_{\LCC}$.\qed
\end{theorem}

\subsubsection{\texorpdfstring{Alternative Definitions of Bisimilarity in $\LLAZYCC$}{Alternative Definitions of Bisimilarity in LCC}}\label{subsec:bisim-alternative}
We want to analyze the translations  between our calculi, and the inherent contextual equivalence.
This will require to show that several differently defined relations are all equal to contextual equivalence.

Using Theorem \ref{thm:conv-admissibile-implies} we show that in  $\LLAZYCC$, 
behavioral equivalence can also be proved inductively:
\begin{definition}\label{def:sim-Q-contexts}
The set $Q_\LCC$ of contexts $Q$ is assumed to consist of the following contexts: 
\begin{enumerate}
\item[(i)]  $([\cdot]~r)$ for  all closed $r$, 
\item[(ii)] for all types $T$, 
constructors $c$ of $T$, and indices $i$:\\
 $(\tcase_T~[\cdot]~\tof\ldots (c~x_1 \ldots x_{\ari(c)} \to x_i)\ldots)$ where all right hand sides of other \tcase-alternatives are $\Omega$, 
\item[(iii)] for all types $T$ and  
constructors $c$ of $T$: $(\tcase_T~[\cdot]~\tof \ldots (c~x_1 \ldots x_{\ari(c)} \to \ttrue) \ldots)$ where all right hand sides of other \tcase-alternatives are $\Omega$. 
\end{enumerate}
The relations  $\leb_{\LCC,Q_\LCC}$, $\lec_{\LCC,Q_\LCC}$ are instantiations of Definitions \ref{def:Q-gfp-preorder} 
  and Definition \ref{def:Q-simpl-preorder}, respectively, with the set $Q_\LCC$ and the {\em closed part of} $\LLCC$
consisting of the subsets of all {\em closed} $\LAMBDAEXPR$-expressions, {\em closed} contexts $\LAMBDACTXT$, and {\em closed} answers ${\ANSWERS}_{\LCC}$.
\end{definition}

\begin{lemma}\label{lem:convergence-admissible}
The calculus $\LLAZYCC$ is convergence-admissible in the sense of Definition \ref{def:convergence-admissible}, 
where the $Q$-contexts are defined as above.
\end{lemma}
\begin{proof}
Values in $\LLAZYCC$ are $\LLAZYCC$-WHNFs. The contexts $Q$ are reduction contexts in $\LLAZYCC$.
Hence every reduction of $Q[s]$ will first evaluate $s$ to $v$ and then evaluate $Q[v]$.
\end{proof}

\begin{theorem}\label{thm:llc-Q-equivalence}
$\leb_{\LCC}~ = ~\leb_{\LCC,Q_\LCC}~ =~ \lec_{\LCC,Q_\LCC}$ and  $\leb_{\LCC}^o = ~\leb_{\LCC,Q_\LCC}^o~ =~ \lec_{\LCC,Q_\LCC}^o$
\end{theorem}
\begin{proof}
Theorem \ref{thm:conv-admissibile-implies}. 
shows $\leb_{\LCC,Q_\LCC} ~= ~\lec_{\LCC,Q_\LCC}$ since $\LLAZYCC$ is  convergence-admissible.

\vspace{1mm}

The first equation is proved by showing  that the  relations satisfy the fixpoint equations of the other one in 
Definition \ref{def:lazycc-le-b} and \ref{def:Q-gfp-preorder}, respectively.  
\begin{itemize}
  \item  ${\leb_{\LCC}} \subseteq {F_Q(\leb_{\LCC})}$:  Assume $s \leb_{\LCC} t$ for two closed $s,t$. 
    If $s\maycon_\LCC v$, then $t\maycon_\LCC w$ for values $v,w$. Since reduction preserves $\simb_{\LCC}$, the fixpoint operator
     conditions are satisfied if $v,w$ are both abstractions or both constructor applications. 
    If $v = \lambda x.s'$ with $s' \in \cBot$ and $w = c(t_1,\ldots,t_n)$, $Q(v)$ diverges for all $Q$, hence $s~F_Q(\leb_{\LCC})~t$.
  \item ${\leb_{\LCC,Q_\LCC}} \subseteq {F_\LCC(\leb_{\LCC,Q_\LCC})}$: Assume $s \leb_{\LCC,Q_\LCC} t$. 
    Let $s\maycon_\LCC v$. Then also $t\maycon_\LCC w$ for some value $w$. In the cases that $v,w$ are both abstractions or both constructor applications,
      when using appropriate $Q$ of kind (ii) or (iii), the $F_\LCC$-conditions are satisfied. If  $v = \lambda x.s'$ and $w = c(t_1,\ldots,t_n)$, we have to show
      that $s' \in \cBot$: this can be done using all $Q$-contexts of the form $([]~r)$, since $(w~r)\DIV_\LCC$ in any case. \qedhere
\end{itemize}
\end{proof}

\begin{definition}\label{def:Q-Omega-expr-contexts}
Let $\mathit{CE}_{\LCC}$ be the following set of closed $\LAMBDAEXPR$-expressions built from constructors, $\Omega$, and 
closed abstractions. 
These can be constructed according to the grammar: 
   $$r \in \mathit{CE}_{\LCC} ::= \Omega~|~\lambda x. s~|~(c~r_1 \ldots r_{\ari(c)})$$
   where $\lambda x.s$ is any closed $\LAMBDAEXPR$-expression.

The set ${Q}_{\mathit{CE}}$ is defined like the set ${Q_\LCC}$ in Definition \ref{def:sim-Q-contexts}, 
but only expressions $r$ from $\mathit{CE}_{\LCC}$ are taken into account
in the contexts $([\cdot]~r)$ in (i).
\end{definition}

We summarize several alternative definitions for contextual preorder and applicative simulation for $\LLCC$, where we also include  further alternatives.
\eqref{theo-main-0} is contextual preorder, \eqref{theo-main-1} the applicative simulation,
\eqref{theo-main-2}, \eqref{theo-main-3} and \eqref{theo-main-4} are similar to the usual call-by-value variants, where
\eqref{theo-main-3} and \eqref{theo-main-4} separate the closing part of contexts,
where \eqref{theo-main-4} can be seen as bridging the gap between call-by-need and  call-by-name. \eqref{theo-main-6}  is the ${\cal Q}$-similarity,
\eqref{theo-main-7}  is a further improved inductive ${\cal Q}$-simulation by restricting the set of test arguments for abstractions, 
and \eqref{theo-main-8} is the co-inductive version of \eqref{theo-main-7}.

\begin{theorem}\label{theorem:bisim-alternatives}
 In $\LLAZYCC$, all the following relations on open $\LAMBDAEXPR$-expressions are identical:
\begin{enumerate}
\item\label{theo-main-0} $\lec_{\LCC}$.
\item\label{theo-main-1} $\leb_{\LCC}^o$. 
\item\label{theo-main-2} The relation $\lec_{\LCC,1}$ defined by $s_1 \lec_{\LCC,1} s_2$ iff for all closing contexts $C$: $C[s_1]\maycon_{\LCC} \implies C[s_2]\maycon_{\LCC}$.
\item\label{theo-main-3} The relation $\lec_{\LCC,2}$, defined as: $s_1 \lec_{\LCC,2} s_2$ iff for all closed contexts $C$ and all closing substitutions $\sigma$: 
$C[\sigma(s_1)]\maycon_{\LCC} \implies C[\sigma(s_2)]\maycon_{\LCC}$.
\item\label{theo-main-4} The relation $\lec_{\LCC,3}$, defined as: $s_1 \lec_{\LCC,3} s_2$ iff for all multi-contexts $M[\cdot,\ldots,\cdot]$ and all 
substitutions $\sigma$: $M[\sigma(s_1),\ldots,\sigma(s_1)]\maycon_{\LCC} \implies M[\sigma(s_2),\ldots,\sigma(s_2)]\maycon_{\LCC}$.
\item\label{theo-main-4a} The relation $\lec_{\LCC,4}$, defined as: $s_1 \lec_{\LCC,4} s_2$ iff for all contexts $C[\cdot]$ and all 
substitutions $\sigma$: $C[\sigma(s_1)]\maycon_{\LCC} \implies C[\sigma(s_2)]\maycon_{\LCC}$.
\item\label{theo-main-6} 
  $\lec_{\LCC,Q_\LCC}^o$ 
\item\label{theo-main-7} 
The relation   $\lec_{\LCC,Q_{\mathit{CE}}}^o$  where $\lec_{\LCC,Q_{\mathit{CE}}}$ is 
defined as in Definition \ref{def:Q-simpl-preorder}  instantiated by the closed part of $\LLCC$ 
and by the set $Q_{\mathit{CE}}$ in  Definition \ref{def:Q-Omega-expr-contexts}.
\item\label{theo-main-8} The relation $\leb_{\LCC,Q_{\mathit{CE}}}^o$   
is defined  as in Definition \ref{def:Q-gfp-preorder} 
instantiated by the closed part of $\LLCC$ 
and by set $Q_{\mathit{CE}}$ in  Definition \ref{def:Q-Omega-expr-contexts}.
\end{enumerate}
\end{theorem}
\begin{proof*} 
\begin{itemize}
\item \eqref{theo-main-0} $\iff$ \eqref{theo-main-1} $\iff$ \eqref{theo-main-6}: This is Theorem \ref{thm:sim_c-equiv-b} and Theorem \ref{thm:llc-Q-equivalence}        
\item \eqref{theo-main-0} $\iff$ \eqref{theo-main-2}: The ``$\Rightarrow$''-direction is obvious. 
For the other direction let $s_1 \lec_{\LCC,1} s_2$ and let $C$ be a context such that 
$\emptyset \not= \FV(C[s_1]) \cup \FV(C[s_2]) = \{x_1,\ldots,x_n\}$ and let
$C[s_1]\maycon_{\LCC}$, \ie\ $C[s_1]\xrightarrow{\LCC,*}v$ where $v$ is an abstraction or a constructor application.
Let $C' = (\lambda x_1,\ldots,x_n.C)~\underbrace{\Omega\ldots\Omega}_{\text{$n$-times}}$.
Then $C'[s_i] \xrightarrow{\LCC,\rnbeta,*} s_i' = C[s_i][\Omega/x_1,\ldots,\Omega/x_n]$ for $i=1,2$.
It is easy to verify that the reduction for $C[s_1]$ can also be performed for $s_i'$, since no reduction in the sequence
$C[s_1]\xrightarrow{\LCC,*}v$ can be of the form $R[x_i]$ with $R$ being a reduction context. 
Thus $C'[s_i] \maycon_{\LCC}$. Since $C'[s_i]$ must be closed for $i=1,2$,  the precondition implies $C'[s_2]\maycon_{\LCC}$ and also
$s_2'\maycon_{\LCC}$. W.l.o.g. let $s_2' \xrightarrow{\LCC,*}  v'$ where $v'$ is an $\LLCC$-WHNF. 
It is easy to verify that no term in this sequence
can be of the form $R[\Omega]$,  where $R$ is a reduction context, since otherwise the reduction would not terminate 
(since $R[\Omega]\xrightarrow{\LCC,+}R[\Omega]$). This implies that we can replace the $\Omega$-expression by the free variables, \ie\ 
that $C[s_2]\maycon_{\LCC}$. 
Note that this also shows by the previous items (and Corollary~\ref{cor:lcc-red-rules-correct}) that $(\rnbeta)$ is correct for $\simc_{\LCC}$.
\item \eqref{theo-main-0} $\iff$ \eqref{theo-main-3}:  This follows from Corollary \ref{cor:lcc-red-rules-correct}
since closing substitutions can be simulated by a context with subsequent (\rnbeta)-reduction.
This also implies that $(\rnbeta)$ is correct for  $\simc_{\LCC,2}$ and by the previous item it also correct for $\simc_{\LCC,1}$
(where $\simc_{\LCC_i}~ = ~\lec_{\LCC,i} \cap \gec_{\LCC,i}$).
\item  \eqref{theo-main-4a} $\iff$ \eqref{theo-main-0}
The  direction ``$\implies$" is trivial. For the other direction let $s_1 \lec_{\LCC} s_2$
and let $C$ be a context, $\sigma$ be a substitution, such that $C[\sigma(s_1)]\maycon_{\LCC}$.
Let $\sigma=\{x_1 \to t_1,\ldots x_n \to t_n\}$ and let 
$C' = C[(\lambda x_1,\ldots,x_n.[\cdot])~t_1~\ldots~t_n]$. Then $C'[s_1]\xrightarrow{\rnbeta,n} C[\sigma(s_1)]$.
Since $(\rnbeta)$-reduction is correct for $\simc_{\LCC}$, 
we have $C'[s_1]\maycon_{\LCC}$. Applying $s_1 \lec_{\LCC} s_2$ 
yields $C'[s_2]\maycon_{\LCC}$. Since $C'[s_2]\xrightarrow{\rnbeta,n} C[\sigma(s_2)]$ and $(\rnbeta)$ is correct for $\simc_{\LCC}$, we have $C[\sigma(s_2)]\maycon_{\LCC}$.
\item \eqref{theo-main-4} $\iff$ \eqref{theo-main-4a}:
Obviously, $s_1 \lec_{\LCC,3} s_2 \implies  s_1 \lec_{\LCC,4} s_2$. We show the other direction by induction on $n$ -- the number of holes in $M$ -- 
that for all $\LAMBDAEXPR$-expressions $s_1,s_2$:
$s_1 \lec_{\LCC,4} s_2$ implies $M[\sigma(s_1),\ldots \sigma(s_1)]\maycon_{\LCC} \implies M[\sigma(s_2),\ldots \sigma(s_2)]\maycon_{\LCC}$.

The base cases for $n=0,1$ are obvious. For the induction step assume that $M$ has $n > 1$ holes.
Let $M' = M[\sigma(s_1),\cdot_2,\ldots,\cdot_n]$
and $M'' = M[\sigma(s_2),\cdot_2,\ldots,\cdot_n]$.
Then obviously $M'[\sigma(s_1),\ldots,\sigma(s_1)] = M[\sigma(s_1),\ldots,\sigma(s_1)]$ and thus 
$M'[\sigma(s_1),\ldots,\sigma(s_1)]\maycon_{\LCC}$. For $C = M[\cdot_1,\sigma(s_1),\ldots,\sigma(s_1)]$ we have
$C[\sigma(s_1)] = M'[\sigma(s_1),\ldots,\sigma(s_1)]$
and also $C[\sigma(s_2)]=M''[\sigma(s_1),\ldots,\sigma(s_1)]$.
Since $C[\sigma(s_1)]\maycon_{\LCC}$, the relation $s_1 \lec_{\LCC,4} s_2$ implies that
$C[\sigma(s_2)]\maycon_{\LCC}$ and hence $M''[\sigma(s_1),\ldots,\sigma(s_1)]\maycon_{\LCC}$. 
Now, since the number of holes of $M''$ is strictly smaller than $n$, the induction hypothesis show that $M''[\sigma(s_2),\ldots,\sigma(s_2)]\maycon_{\LCC}$.
Because of $M''[\sigma(s_2),\ldots,\sigma(s_2)] = M[\sigma(s_2),\ldots,\sigma(s_2)]$ we have $M[\sigma(s_2),\ldots,\sigma(s_2)]\maycon_{\LCC}$.
\item \eqref{theo-main-6} $\iff$ \eqref{theo-main-7}: The direction \eqref{theo-main-6} $\implies$ \eqref{theo-main-7} is trivial. 

For the other direction we show that $\lec_{\LCC,Q_{\mathit{CE}}}^o~ \subseteq~ \lec_{\LCC,Q_\LCC}^o$ by showing  that the
inclusion $\lec_{\LCC,Q_{\mathit{CE}}} ~ \subseteq~ \lec_{\LCC,Q_\LCC}$ holds.
 Let $s_1, s_2$ be closed expressions with $s_1 \lec_{\LCC,Q_{\mathit{CE}}} s_2$
and let $Q_1[\ldots Q_n[s_1]\ldots]\maycon_{\LCC}$ for  $Q_i \in Q_\LCC$.
Let $m$ be the number of normal-order-reductions of  $Q_1[\ldots Q_n[s_1]\ldots]$ to an $\LLCC$-WHNF.
Since the reduction rules are correct \wrt\ $\simc_{\LCC}$, for every subexpression $r$ of the contexts $Q_i$,
there is some $r'$ with $r' \lec_{\LCC} r$,  where $r' \in Q_{\mathit{CE}}$, which is derived from $r$ by (top-down)-reduction, 
which may also be non-normal order, \ie\ $r \xrightarrow{\LCC,*} r_{m+1}$  
where $r_{m+1}$ has reducible subexpressions (not in an abstraction) only at depth at least $m+1$. 
All those deep subexpressions are then replaced by $\Omega$,  and this construction results in $r'$. By construction, $r' \lec_{\LCC} r$. 
We do this for all the contexts $Q_i$, and obtain thus contexts $Q_i'$. The construction using the depth $m$ shows that 
 $(Q_1'[\ldots[Q_n'[s_1]]])\maycon_{\LCC}$, since the normal-order reduction does not use subexpressions at depth greater than $m$ in those $r'$. 
By assumption, this implies  $(Q_1'[\ldots[Q_n'[s_2]]])\maycon_{\LCC}$, 
and since $(Q_1'[\ldots[Q_n'[s_2]]])  \lec_{\LCC}  (Q_1'[\ldots[Q_n'[s_2]]])$, this also implies   $(Q_1[\ldots[Q_n[s_2]]])$.

\item  \eqref{theo-main-7} $\iff$ \eqref{theo-main-8}: This follows for the relations on closed expressions 
by Theorem~\ref{thm:conv-admissibile-implies}, since the deterministic calculus (see Def.~\ref{def:calculus}) 
  for $\LLAZYCC$ with $Q_{\mathit{CE}}$ 
as defined above 
is convergence-admissible. It also holds for the extensions to open expressions, since the construction for the open extension is identical for both relations.
\qedhere
\end{itemize}
\end{proof*}
\noindent Also the following can easily be derived from Theorem \ref{theorem:bisim-alternatives} and 
Corollary  \ref{cor:lcc-red-rules-correct}.
\begin{proposition}\label{prop:open-closed-case}
  For open $\LAMBDAEXPR$-expressions $s_1,s_2$, where all free variables of $s_1,s_2$ are in $\{x_1,\ldots,x_n\}$: 
  $s_1 \lec_{\LCC}  s_2 \iff \lambda x_1, \ldots x_n.s_1  \lec_{\LCC}  \lambda x_1, \ldots x_n.s_2$ \qed
\end{proposition}

\begin{proposition}\label{prop:finite-simulation-LLC}
Given any two closed $\LAMBDAEXPR$-expressions $s_1,s_2$: \mbox{$s_1 \lec_{\LCC}  s_2$} iff the following conditions hold:
\begin{itemize}
 \item If $s_1\maycon_\LCC \lambda x.s_1'$, then either (i) $s_2\maycon_\LCC \lambda x.s_2'$, and 
      for all (closed) $r \in \mathit{CE}_{\LCC}$:
       $s_1~r  \lec_{\LCC}   s_2~r$, or (ii) $s_2\maycon_\LCC (c~s_1'' \ldots s_n'')$ and $s_1' \in \cBot$. 
  \item if $s_1 \maycon_\LCC (c~s_1' \ldots s_n')$, then  $s_2 \maycon_\LCC (c~s_1'' \ldots s_n'')$, and  for all $i: 
  s_i' \lec_{\LCC}  s_i''$
\end{itemize}  
\end{proposition}
\begin{proof}
The if-direction follows from the congruence property of  $\lec_{\LCC} $ and the correctness of reductions. 
The only-if direction follows from Theorem \ref{theorem:bisim-alternatives}.
\end{proof}
\noindent This immediately implies
\begin{proposition}\label{prop:finite-simulation-explicit-LLC} Given any two closed $\LAMBDAEXPR$-expressions $s_1,s_2$. 
 \begin{itemize}
    \item If $s_1,s_2$ are abstractions, then 
  $s_1 \lec_{\LCC}  s_2$ iff for all closed $r \in \mathit{CE}_{\LCC}$: 
$s_1~r  \lec_{\LCC}   s_2~r$
 \item If $s_1 = (c~t_1 \ldots t_n)$ and $s_2 = (c'~t_1' \ldots t_{m}')$  are constructor expressions, then $s_1 \lec_{\LCC}  s_2$ 
   iff $c = c'$, $n=m$ and  for all $i: t_i \lec_{\LCC}  t_i'$ \qed
\end{itemize}  
\end{proposition}
\noindent We finally consider a more relaxed notion of similarity which allows to use known contextual equivalences as intermediate
steps when proving similarity of expressions:

\begin{definition}[Similarity up to $\simc_\LCC$]\label{def:lcc-sim-upto}
Let $\leb_{\LCC,\simc}$ be the greatest fixpoint of the following operator $\UptoSim$ on closed $\LLAZYCC$-expressions:

We define an operator $\UptoSim$ on binary relations $\eta$ on closed $\LLAZYCC$-expressions:\\
 $s~\UptoSim(\eta)~t$  iff the following holds:
\begin{enumerate}
 \item If $s \simc_{\LCC}\lambda x.s'$ then there are two possibilities: 
     (i) if  $t \simc_{\LCC}(c~t_1 \ldots t_n)$ then $s' \in \cBot$, or  (ii)  if  $t \simc_{\LCC}\lambda x.t'$ then  for all closed $r:   ((\lambda x.s')~r)~\eta~((\lambda x.t')~r)$;
 \item If $s \simc_{\LCC}(c~s_1\ldots s_n)$ then $t \simc_{\LCC}(c~t_1 \ldots t_n)$ and $s_i~\eta~t_i$ for all $i$.
\end{enumerate}
\end{definition}

Obviously, we have $s  \leb_{\LCC,\simc} t$ iff one of the three cases holds: (i) $s \simc_{\LCC}\lambda x.s'$, $t \simc_{\LCC}\lambda x.t'$, 
  and  $(\lambda x.s')~r   \leb_{\LCC,\simc} (\lambda x.t')~r$
for all closed $r$; (ii) $s \simc_{\LCC}\lambda x.s'$, $t \simc_{\LCC}(c~t_1 \ldots t_n)$, and $s' \in \cBot$, or (iii)  $s \simc_{\LCC}(c~s_1 \ldots s_n)$,  $t \simc_{\LCC}(c~t_1 \ldots t_n)$, 
and $s_i \leb_{\LCC,\simc} t_i$ for all $i$.

\begin{proposition}
 $\leb_{\LCC,\simc}~=~\leb_{\LCC} ~= ~\lec_{\LCC}^c$, and $\leb_{\LCC,\simc}^o~=~\leb_{\LCC}^o ~= ~\lec_{\LCC}$.
\end{proposition}
\begin{proof} We show the first equation via the fixpoint equations.
(i) We prove that the relation $\leb_{\LCC,\simc}$  satisfies the fixpoint equation for  $\leb_{\LCC}$:
Let $s  \leb_{\LCC,\simc} t$, where $s,t$ are closed.  
If $s \maycon_\LCC  (c~s_1 \ldots s_n)$, then also $s \simc_{\LCC}(c~s_1 \ldots s_n)$ which clearly implies $t \maycon_\LCC (c~t_1 \ldots t_n)$, and also 
$t \simc_{\LCC}(c~t_1 \ldots t_n)$. The relation  $\leb_{\LCC,\simc}$ is a fixpoint of $\UptoSim(\eta)$, hence  $s_i \leb_{\LCC,\simc} t_i$  for all $i$.  

 If  $s \maycon_\LCC  \lambda x.s'$ and $t \maycon_\LCC \lambda x.t'$  then similar arguments show
  $((\lambda x.s')~r)~\leb_{\LCC,\simc}~((\lambda x.t')~r)$ for all $r$. 
  If  $s \maycon_\LCC   \lambda x.s'$ and $t \maycon_\LCC  (c~t_1 \ldots t_n)$,  then $s \simc_{\LCC}\lambda x.s'$ and $t \simc_{\LCC}(c~t_1 \ldots t_n)$. Again the fixpoint property
 of  $\leb_{\LCC,\simc}$ shows $s' \in \cBot$.
  
  \vspace*{1mm}
  
(ii)  We prove that the relation $\leb_{\LCC}$  satisfies the fixpoint equation for  $\UptoSim$:
  Let $s  \leb_{\LCC} t$ for closed $s,t$. We know that this is the same as $s  \lec_{\LCC} t$. 
If $s \simc_{\LCC}(c~s_1 \ldots s_n)$, then clearly  $s \maycon_\LCC  (c~s_1' \ldots s_n')$ where $(c~s_1 \ldots s_n) \simc_{\LCC}(c~s_1' \ldots s_n')$.
Since in this case $t \simc_{\LCC}(c~t_1 \ldots t_n)$ and thus 
  $t \maycon_\LCC  (c~t_1' \ldots t_n')$  where $t \simc_{\LCC}(c~t_1 \ldots t_n) \simc_\LCC (c~t_1' \ldots t_n')$, and 
also $s_i \leb_{\LCC,\simc} t_i$ for all $i$ holds, since reduction is correct. 
If  $s \simc_{\LCC}\lambda x.s'$ and $s \simc_{\LCC}\lambda x.t'$  then $s \maycon_\LCC  \lambda x.s''$ and $t \maycon_\LCC  \lambda x.t''$ and
  $((\lambda x.s')~r)~\leb_{\LCC,\simc}~((\lambda x.t')~r)$. 
  
  If  $s \simc_{\LCC}\lambda x.s'$ and $s \simc_{\LCC}(c~t_1 \ldots t_n)$, then for $s \maycon_\LCC  \lambda x. s''$,  we have $\lambda x.s' \simc_{\LCC}\lambda x.s''$, and
  since $s \lec_{\LCC} t$,  the characterization of expressions in Proposition 
  \ref{prop-classification-lcc} shows $s',s'' \in \cBot$. 
\end{proof}

%%%%%%%%%%%%%%%%%%%%%%%%%%%%%%%%%%%%%%%%%%%%%%%%%%%%%%%%%%%%%%%%%%%%%%%%%%%%%%
%% SECTION 5: THE TRANSLATION W
\section{\texorpdfstring{The Translation $W:\LLR \to \LNAME$}{The Translation W}} \label{sec-translation-NEED-NAME}
The translation $W: \LLR \to \LNAME$ is defined as the identity on expressions and contexts, 
but the definitions of convergence predicates are different. 
In this section we prove that contextual equivalence based on $\LLR$-evaluation 
and contextual equivalence based on $\LNAME$-evaluation are equivalent. 
We use infinite trees to connect both evaluation strategies.
In \cite{schmidt-schauss-copy-rta:07} a similar result was shown for
a lambda calculus without seq, case, and constructors.

\subsection{\texorpdfstring{Calculus for Infinite Trees $\LTREE$}{Calculus for Infinite Trees Ltree}}
We define infinite expressions which are intended to be the letrec-unfolding of the $\LETRECEXPR$-expressions 
with the extra condition that cyclic variable chains lead to local nontermination represented by \tbot. 
We then define the calculus $\LTREE$ which has infinite expressions and performs reduction on infinite expressions.

\begin{definition}\label{def-IEXPR}
{\em Infinite expressions} $\IEXPR$ are defined like expressions $\LETRECEXPR$ without letrec-expressions,  adding a constant $\tBot$,
 and interpreting the grammar co-inductively, \ie\ the grammar is as follows 
\[
 \begin{array}{@{}rcl}
S,T,S_i,T_i \in\IEXPR &::=& x \syxor (S_1~S_2) \syxor (\lambda x. S) \syxor \tbot
\\
    &&\syxor (c~S_1 \ldots S_{\ari(c)})
    \syxor (\tseq~S_1~S_2)
    \syxor (\tcase_T~S~\tof~alts)
\end{array}
\]
\end{definition}
In order to distinguish in the following the usual expressions from the infinite ones, we say {\em tree} or infinite expressions.
As meta-symbols we use $s,s_i,t,t_i$ for finite expressions and $S,T,S_i,T_i$ for infinite expressions.
The constant $\tBot$ is without any reduction rule. 
  
In the following definition of a mapping from finite expressions to their infinite images,   
we sometimes use the explicit binary application operator $@$ for applications inside the trees (\ie\ an application in the tree is sometimes
written as $(@~S_1~S_2)$ instead of $(S_1~S_2)$), since it is easier to explain, but use the common notation in other places. 
A {\em position} is a finite sequence of positive integers, where the empty position is denoted as $\varepsilon$. 
We use Dewey notation for positions, \ie\ the position $i.p$ is the sequence starting with $i$ followed by position $p$. 
For an infinite tree  $S$ and position $p$, the notation  $S|_p$ means the subtree at position $p$ and
$p(S)$ denotes the head symbol of $S|_p$.

This induces the representation of an infinite expression $S$ as a (partial) function $S$ from positions to labels
where application of the function $S$ to a position $p$ is written as $p(S)$ and
where the labels are $@$, $\tcase_T$, $(c~x_1~\ldots~x_n)$ (for a \tcase-alternative), $\tseq$, $c$, $\lambda x$, and $x$.
The domain of such a function must be a prefix-closed set of positions,
and the continuations of a position $p$ depend on the label at $p$ and must 
coincide with the syntax according to the grammar in Definition~\ref{def-IEXPR}.
\begin{definition} 
The translation $\ITT:\LETRECEXPR \to {\IEXPR}$ translates an expression $s\in\LETRECEXPR$ 
into its infinite tree $\ITT(s)\in\IEXPR$. 
We define the mapping $\ITT$ by providing an algorithm that,
computes the partial function $\ITT(s)$ from positions to labels.
Given a position $p$, computing $p(\ITT(s))$ starts with $s\LABELCOMP_{p}$ and then proceeds with the rules given in \FIGURE~\ref{figure-inftree}.
The first group of rules defines the computed label for the position $\varepsilon$,
the second part of the rules describes the general case for positions.
If the computation fails (or is undefined), then the position is not valid in the tree $\ITT(s)$.
The equivalence of infinite expressions is extensional equality of the corresponding functions, 
where we additionally do not distinguish $\alpha$-equal trees.
\end{definition}

\begin{figure*}[t]
\fbox{\parbox{.98\textwidth}{
\[\begin{array}{@{\,}l@{~}c@{~}l@{~}l@{\,}}
\Ctxt[(s~t)\LABELCOMP_{\varepsilon}] & \mapsto &  @ \\
\Ctxt[({\tt case}_T~\ldots)\LABELCOMP_{\varepsilon}] & \mapsto & {\tt case}_T \\
\Ctxt[(c~x_1\ldots~x_n\to s)\LABELCOMP_{\varepsilon}] & \mapsto & (c~x_1~\ldots~x_n) \qquad \mbox{for a case-alternative}\\
\Ctxt[({\tt seq}~s~t)\LABELCOMP_{\varepsilon}] & \mapsto & {\tt seq} \\
\Ctxt[(c~s_1 \ldots s_n)\LABELCOMP_{\varepsilon}] & \mapsto & c \\
\Ctxt[(\lambda x.s)\LABELCOMP_{\varepsilon}] & \mapsto &  \lambda x \\
\Ctxt[x\LABELCOMP_{\varepsilon}] & \mapsto &  x \hspace*{.3cm} \mbox{if $x$ is a free variable or a lambda-bound variable in $\Ctxt[x]$} \\
\end{array}\]
\begin{flushleft}
The cases for general positions $p$: 
\end{flushleft} 
 \[\begin{array}{ll@{\quad}c@{\quad}ll}
1. &\Ctxt[(\lambda x.s)\LABELCOMP_{1.p}] & \mapsto &  \Ctxt[\lambda x.(s\LABELCOMP_p)] \\
2. &\Ctxt[(s~t)\LABELCOMP_{1.p} ] & \mapsto &   \Ctxt[(s\LABELCOMP_p ~t)]\\
3. &\Ctxt[(s~t)\LABELCOMP_{2.p} ] & \mapsto &   \Ctxt[(s~t\LABELCOMP_p)]\\
4. &\Ctxt[(\tseq~s~t)\LABELCOMP_{1.p} ] & \mapsto &   \Ctxt[(\tseq~s\LABELCOMP_p~t)]\\
5. &\Ctxt[(\tseq~s~t)\LABELCOMP_{2.p} ] & \mapsto &   \Ctxt[(\tseq~s~t\LABELCOMP_p)]\\
6. &\Ctxt[(\tcase_T~s~{\tt of}~ \mathit{alt}_1 \ldots \mathit{alt}_n)\LABELCOMP_{1.p} ] & \mapsto &   \Ctxt[(\tcase_T ~ s\LABELCOMP_{p}~{\tt of} ~\mathit{alt}_1 \ldots \mathit{alt}_n)]\\
7. &\Ctxt[(\tcase_T ~s~{\tt of} ~\italt_1 \ldots \italt_n)\LABELCOMP_{(i + 1).p} ] & \mapsto &   \Ctxt[(\tcase_T ~s~{\tt of} \italt_1 \ldots {\italt_i}\LABELCOMP_{p} \ldots \italt_n)]\\
8. &\Ctxt[\ldots (c~x_1~\ldots~x_n \to s)\LABELCOMP_{1.p} \ldots ] & \mapsto &  \Ctxt[\ldots  (c~x_1~\ldots~x_n \to s\LABELCOMP_{p}) \ldots ] \\
9. &\Ctxt[(c~s_1 \ldots s_n)\LABELCOMP_{i.p} ] & \mapsto &   \Ctxt[(c ~s_1 \ldots s_i\LABELCOMP_{p} \ldots s_n)]\\
10.&\Ctxt[\tletrx{\iEnv}{s}\LABELCOMP_{p}]   & \mapsto &   \Ctxt[\tletrx{ \iEnv}{s\LABELCOMP_{p}}]  \\
11. &\Ctxt_1[\tletrx{x = s, \iEnv}{\Ctxt_2[x\LABELCOMP_p]}]  &\mapsto &  \Ctxt_1[\tletrx{x = s\LABELCOMP_p, \iEnv}{\Ctxt_2[x]}] \\
12. &\begin{array}{@{}l}\Ctxt_1[\tletrec~x = s, y = \Ctxt_2[x\LABELCOMP_p],\\ 
          \hspace*{1.5cm} \iEnv~\tin~t]  \end{array}
     & \mapsto   
  &  \begin{array}{@{}l} \Ctxt_1[\tletrec~x = s\LABELCOMP_p, y = \Ctxt_2[x], \\ \hspace*{1.5cm} \iEnv~\tin~t]  \end{array}\\ 
13. &\Ctxt_1[\tletrx{x = \Ctxt_2[x\LABELCOMP_p], \iEnv}{s}]  &\mapsto &  \Ctxt_1[\tletrx{x = \Ctxt_2[x]\LABELCOMP_p, \iEnv}{s}] \\[1mm]

  \multicolumn{4}{l}{\mbox{If the  position $\varepsilon$ hits the same  (let-bound) variable twice, then the result is $\tBot$.}}\\
    \multicolumn{4}{l}{\mbox{(This can only happen by a sequence of rules 11,12,13.)}} \\
\end{array}\]
}}
\caption{Infinite tree construction from positions for fixed $s$}\label{figure-inftree}
\end{figure*}

\begin{example}
The expression 
 $\tletrec~ x = x, y = (\lambda z.z)~x~y~\tin~y$  has the corresponding  tree  
 $((\lambda z_1.z_1)~\tBot~((\lambda z_2.z_2)~\tBot~((\lambda z_3.z_3) ~\tBot~\ldots)))$.
\end{example}

The set $\INFCTXT$ of infinite tree contexts includes any infinite tree where a subtree is replaced by a hole $[\cdot]$.  
Reduction contexts on trees are defined as follows:

\begin{definition}
Call-by-name reduction contexts $\IECtxts$ of $\LTREE$ are defined as follows, where the grammar is interpreted inductively
and  $S\in\IEXPR$:
\[
\begin{array}{rcl}
 \ECtxt, \ECtxt_i \in \IECtxts &::= &[\cdot] \OR (\ECtxt~S) \OR (\tcase~\ECtxt~\tof~alts) \OR (\tseq~\ECtxt~S)\\
\end{array}
\]
For an infinite tree, a {\em reduction position} $p$ is any position such that $p(S)$ is defined and
there exists some $\ECtxt\in\IECtxts$ with $\ECtxt[S'] = S$ and $\ECtxt|_p = [\cdot]$
\end{definition}

\begin{definition} An $\LTREE$-answer (or an $\LTREE$-WHNF) is any infinite $\IEXPR$-expression $S$ which is an abstraction or constructor application,
\ie\ $\varepsilon(S) = \lambda x$ or $\varepsilon(S) = c$ for some constructor $c$.
The reduction rules on infinite expressions are allowed in any context and are as follows: 
\[\begin{array}{llcl}
(\rbetaTr) &  ((\lambda x. S_1)~S_2)   \to    S_1[S_2/x] \\
(\rseqTr)  &  (\tseq~S_1~S_2) ~~\to~~S_2 \hspace*{0,5cm} \mbox{if $S_1$ is an $\LTREE$-answer}\\
(\rcaseTr)  & (\tcase_T~(c~S_1 \ldots S_n)~\tof~\ldots (c~x_1 \ldots x_n \to S')\ldots) ~~\to~~ S'[S_1/x_1,\ldots,S_n/x_n] \\ 
\end{array}\]
If $S = R[S_1]$ for a  $\IECtxts$-context $R$, and $S_1 \xrightarrow{a} S_2$ for $a \in\{(\rbetaTr), (\rcaseTr),\mbox{ or } (\rseqTr)\}$,  
then we say $S \xrightarrow{\TREE} S' = R[S_2]$ is a {\em normal order reduction} ($\TREE$-reduction) on infinite trees.
Here $S_1$ is the {\em tree-redex} of the tree-reduction. 
We also use the convergence predicate $\maycon_{\TREE}$ for infinite trees defined as: $S\maycon_{\TREE}$ iff $S \xrightarrow{\TREE,*} S'$ and
$S'$ is an $\LTREE$-WHNF.
\end{definition}

Note that  $\xrightarrow{\TREE,\rbetaTr}$ and $\xrightarrow{\TREE,\rcaseTr}$ only reduce a single redex, 
but may modify infinitely many positions, since there may be infinitely many positions of a replaced variable $x$.  
\Eg, a \mbox{($\TREE$,\rbetaTr)} of $\ITT((\lambda x. \tletrx{z = (z~ x)}{z}) ~r)  =  (\lambda x. ((\ldots ~(  \ldots ~ x) ~x)~ x))~ r$
$\to~  ((\ldots ~(  \ldots ~ r) ~r)~ r)$
replaces the infinite number of occurrences of $x$ by $r$. 

Concluding, the calculus $\LTREE$ is defined by the tuple $(\IEXPR,\INFCTXT,\xrightarrow{\TREE},{\ANSWERS}_{\TREE})$
where ${\ANSWERS}_{\TREE}$ are the $\LTREE$-WHNFs.

In the following we  use a variant of infinite outside-in developments \cite{barendregt:84,kennaway-klop:97} as a reduction on trees 
that may reduce infinitely many redexes in one step.
The motivation is that the infinite trees corresponding to finite expressions may  
require the reduction of
infinitely many redexes of the trees for one $\xrightarrow{\LR}$- or $\xrightarrow{\LNAME}$-reduction, respectively.

\begin{definition}\label{def:one-reduction} 
We define an infinite variant of Barendregt's 1-reduction:
Let $S \in \IEXPR$ be an infinite  tree.  
Let $\dagger$ be a special label and $M$ be a set of (perhaps infinitely many) positions of $S$, which must be redexes
\wrt\ the same reduction $a \in\{(\rbetaTr), (\rcaseTr), \mbox{ or } (\rseqTr)\}$.
Now exactly all  positions $m \in M$ of $S$ are labeled with $\dagger$. 
By $S \xrightarrow{I,M} S'$ we denote the (perhaps infinite) development top down, defined as follows:
\begin{itemize}
  \item  Let $S_0 = S$ and $M_0 = M$. 
  \item Iteratively compute $M_{i+1}$ and $S_{i+1}$ from $M_i$ and $S_i$ 
for $i=0,1,2,\ldots$ as follows: \\
Let $d$ be the length of the shortest position in $M_i$, and $M_{i,d}$ be the finite set of positions 
that are the shortest ones in $M_i$. 

For every $p \in M_{i,d}$ construct an infinite tree $T_p$ from ${S_i}|_{p}$ by iterating the following reduction
until the root of ${S_i}|_{p}$ is not labeled: remove the  label from the top of ${S_i}|_{p}$, then perform a labeled reduction inheriting all 
the labels. 
If this iteration does not terminate, because the root of  ${S_i}|_{p}$ gets labeled in every step, then the result is 
$T_p := \tBot$  (unlabeled), otherwise a result $T_p$ is computed after finitely many reductions. 

Now construct $S_{i+1}$ by replacing every subtree at a position $p \in M_{i,d}$ in $S_{i}$ by $T_p$:
for the positions $p$ of $S_i$ that do not have a prefix that is in $M_{i,d}$, we set $p(S_{i+1}) := p(S_i)$ and
for $p \in M_{i,d}$ we set $S_i|_p := T_p$.

Let $M_{i+1}$ be the set of positions in $S_{i+1}$ which carry a label $\dagger$. 
The length of the shortest position is now at least $d+1$. 
Then iterate again with $M_{i+1}, S_{i+1}$.
\item $S'$ is defined as the result after (perhaps infinitely many) construction steps $S_1,S_2,\ldots$
\end{itemize}
If the initial set $M$ does not contain a reduction position then we write $S \xrightarrow{I,M,\NTREE} S'$.
We write $S \xrightarrow{I,\NTREE} S'$ ($S \xrightarrow{I} S'$, resp.)
if there exists a set $M$ such that $S \xrightarrow{I,M,\NTREE} S'$ ($S \xrightarrow{I,M} S'$, resp.). 
\end{definition}

\begin{example}\label{example:inf-tree-reductions-std} We give two examples of standard reduction and $\xrightarrow{I,M}$-reductions.\\
An $\xrightarrow{\LR}$-reduction on expressions corresponds to an $\xrightarrow{I,M}$-reduction on infinite trees
and perhaps corresponds to an infinite sequence of infinite $\TREE$-reductions.
Consider $\tletrec~y=(\lambda x.y)~a~\tin~y$. The $(\LR,\rlbeta)$-reduction 
with a subsequent $(\LR,\rllet)$ reduction results in 
$\tletrec~y=y,x=a~\tin~y$.  The corresponding infinite tree of $\tletrec~y=(\lambda x.y)~a~\tin~y$ is
$S = ((\lambda x_1.((\lambda x_2.((\lambda x_3.(\ldots~a))~a))~a))~a)$.
The $(\TREE,\rbetaTr)$-reduction-sequence is infinite.
let  $M$ be the infinite set of positions of all the applications in $S$, \ie\ $M = \{\varepsilon,1.1,  1.1.1.1, \ldots\}$. Then in the (infinite) development described in  Def. \ref{def:one-reduction} all intermediate trees have a label at the top,
and thus we have $S \xrightarrow{I,M} \tBot$. For a set $M$ without $\varepsilon$, the result will be a value tree.  

  For the expression  $\tletrec~y = (\tseq ~\ttrue ~(\tseq~y~\tfalse))~\tin~y$ the $\xrightarrow{\LR}$-reduction results in
the expression $\tletrec~y=(\tseq~y~\tfalse)~\tin~y$ which diverges. The corresponding infinite tree is
   $(\tseq ~\ttrue~ (\tseq~((\tseq ~\ttrue~ (\tseq~(\ldots)~\tfalse))~\tfalse)))$, which has an infinite number of $\TREE$-reductions,
   at an infinite number of deeper and deeper positions. 
   Let $M = \{\varepsilon, 1.2, 1.2.1.2, \ldots\}$ be the set consisting of all those positions.
Then $S \xrightarrow{I,M}  (\tseq~(\tseq ~(\tseq \ldots~\tfalse)~\tfalse)~\tfalse)$. 
\end{example}

There may be $S,S'$ such that $S \xrightarrow{I,M} S'$ as well as $S \xrightarrow{I,M'} S'$ for some sets $M,M'$ where $M$ contains a reduction position, 
but $M'$ does not contain a reduction position. For example $S = (\lambda x_1.x_1)~((\lambda x_2.x_2)~((\lambda x_3.x_3) \ldots))$, where a single (\rbetaTr)-reduction   
at the top reproduces $S$,
as well as a single (\rbetaTr)-reduction of the  argument.

\subsection{Standardization of Tree Reduction}\label{subsec-inf-proc-standardization}
Before considering the concrete calculi $\LLR$ and $\LNAME$ and their correspondence to the calculus with infinite
trees, we show that for an arbitrary reduction sequence on infinite trees resulting
in an answer we can construct a $\TREE$-reduction sequence that results in an $\LTREE$-WHNF.

\begin{lemma}\label{lemma:T-successful-nsr-base-case} Let $T$ be an infinite expression.
 If $T \xrightarrow{I,M,\NTREE} T'$  for some $M$, where $T'$ is an answer, then $T$ is also an answer.
 \end{lemma}
\begin{proof}
This follows since an answer cannot be generated by $\xrightarrow{I,M,\NTREE}$-reductions, since neither abstractions
nor constructor expressions can be generated at the  top  position.
\end{proof}

\begin{lemma}\label{lemma-R-forking} 
Any overlapping between a $\xrightarrow{\TREE}$-reduction and a $\xrightarrow{I,M}$-reduction 
 can be closed as follows. The trivial case that both given reductions
are identical is omitted.  

\begin{tabular}{@{}lll@{}}
  \begin{minipage}{0.3\textwidth}
  \[
\xymatrix@R=6mm@C=6mm{
 T \ar@{->}[r]^{I,M}  \ar@{->}[d]_{\TREE} &   \cdot \ar@{-->}[d]_{\TREE}   \\
\cdot  \ar@{-->}[r]^{I,M'} &  \cdot    \\         
}
\]
  \end{minipage}
&
  \begin{minipage}{0.3\textwidth}
  \[
\xymatrix@R=6mm@C=6mm{
  T \ar@{->}[r]^{I,M}  \ar@{->}[d]_{\TREE} & \cdot   \\
   \cdot  \ar@{-->}[ru]_{I,M'}  \\         
}
\]
  \end{minipage}
&
  \begin{minipage}{0.3\textwidth}
  \[
\xymatrix@R=6mm@C=6mm{
  T \ar@{->}[r]^{I,M}  \ar@{->}[d]_{\TREE} & \cdot \ar@{-->}[ld]^{\TREE}    \\
   \cdot  \\         
}
\]
  \end{minipage}
 \end{tabular}

 \end{lemma}
 \begin{proof}
 This follows by checking the  overlaps of $\xrightarrow{I}$ with $\TREE$-reductions.   
The third diagram applies if the positions of $M$ are removed by the $\TREE$-reduction.
The second diagram applies if the $\TREE$-redex is included in $M$ and the first diagram is applicable
in all other cases.
 \end{proof}

 \begin{lemma}\label{lemma-R-star-forking}
Let $T$ be an infinite  tree such that there is
a $\TREE$-reduction sequence of length $n$ to a WHNF $T'$, and let $S$ be an infinite tree with 
$T \xrightarrow{I,M} S$.
Then $S$  has a $\TREE$-reduction sequence  of length $\le n$ to a WHNF $T''$.  
 \end{lemma}
 \begin{proof}
This follows from Lemma \ref{lemma-R-forking} by induction on $n$.
 \end{proof}

\begin{lemma}\label{lemma-red-triangle} 
Consider two  reductions  $\xrightarrow{I,M_1}$ and  $\xrightarrow{I,M_2}$ of the same type   (\rbetaTr), (\rcaseTr) or (\rseqTr). 
For all trees $T,T_1,T_2$: if $T  \xrightarrow{I,M_1}  T_1$, 
and $T  \xrightarrow{I,M_2}  T_2$, and $M_2 \subseteq M_1$, then there is a set $M_3$ of positions, such that  $T_2 \xrightarrow{I,M_3}  T_1$. 
\[
\xymatrix@C=20mm@R=6mm{
T  \ar@{->}[dr]_{I,M_2}  \ar@{->}[rr]^{I,M_1} & & T_1   \\ 
& T_2 \ar@{-->}[ur]_{I,M_3} &  \\  
}
\]
\end{lemma}
\begin{proof} 
The argument is that the set $M_3$ is computed by labeling the positions in $T$ using $M_1$, and then by performing the infinite development using
the set of redexes $M_2$, where we assume that the $M_1$-labels are inherited.
The set of positions of marked redexes in $T_2$ that remain and are not reduced by  $T_1 \xrightarrow{I,M_2}  T_2$ 
 is exactly the set $M_3$. \qedhere
\end{proof}

Consider a reduction $T \xrightarrow{I,M} T'$ of type (\rbetaTr), (\rcaseTr) or (\rseqTr).
This reduction may include a redex of a normal order \TREE-reduction. Then the reduction can be split into $T \xrightarrow{\TREE} T_1 \xrightarrow{I} T'$, 
and splitting of the reduction can be iterated
as long as the remaining $T_1 \xrightarrow{I} T'$ has a \TREE-redex. 
It may happen that this process does not terminate. 

We consider this non-terminating case, \ie\ let $T_0 \xrightarrow{I,M} T'$ and we can assume that there exist infinitely many 
$T_1,T_2,\ldots$ and $M_1,M_2,\ldots,$ such that for any $k$: $T_0 \xrightarrow{\TREE,k} T_k$ and $T_k \xrightarrow{I,M_k} T'$.
By induction we can show for every $k \geq 1$: 
$T_{k-1} = R_{k-1}[S_{k-1}] \to R_{k-1}[S_k] = T_k$ for a reduction context $R_k$ and 
where $S_{k-1}$ is the redex and $S_{k}$ is the contractum of $T_{k-1} \to T_{k}$ and the normal order 
$\TREE$-redex of $M_k$ labels  a subterm of $S_{k}$.
This holds, since the infinite development for $T \xrightarrow{I,M} T'$  is performed top down.

This implies that the infinite $\TREE$-reduction goes deeper and deeper along one path of the tree,
or at some point all remaining $\TREE$-reductions are performed at the same position.

\begin{lemma}\label{lemma:iteration-stops-maycon}
Let $T \xrightarrow{I,M} T'$ such that $T'\maycon_{\TREE}$ and $M$ labels the normal order redex of $T$.
Then there exists $T''$ and $M'$ such that $T \xrightarrow{\TREE,*} T'' \xrightarrow{M',\NTREE} T'$.
\end{lemma}
\begin{proof}
Let $T = T_0 \xrightarrow{\TREE,k} T_k$, $T_k \xrightarrow{I,M_k} T'$ where $M_k$ labels a normal order redex.
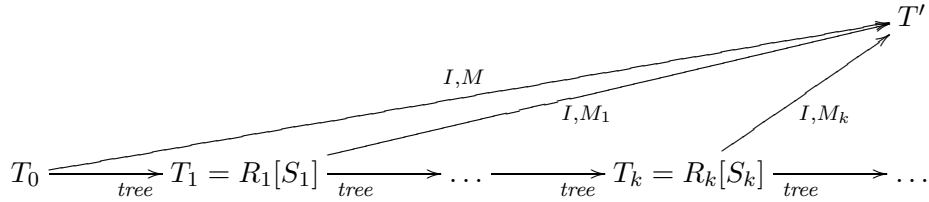
\[
\xymatrix@C=15mm@R=15mm{
  & & & &  T'  \\
T_0  \ar@{->}[urrrr]^{I,M}    
 \ar@{->}[r]_{\TREE}    & T_1=R_1[S_1] \ar@{->}[r]_{\TREE} \ar@{->}[urrr]_{I,M_1} & \ldots \ar@{->}[r]_{\TREE} 
 & T_k=R_k[S_k] \ar@{->}[r]_{\TREE} \ar@{->}[ur]_{I,M_k}   & \ldots
    \\  
}
\]
We have  $T_k = R_k[S_k]$ where $R_k$ is a reduction context, and $M_k$ labels the hole of $R_k$, 
which is the normal order redex. The normal order
reduction is $T_k = R_k[S_k] \xrightarrow{\TREE} R_k[S_k'] =: T_{k+1}$. 
Let $p_k$ be the path of the hole of $R_k$, together with the 
constructors and symbols (\tcase, \tseq, constructors and $@$) on the path. 
Also let $M_k = M_{k,1} \dotcup M_{k,2}$, (where $\dotcup$ means disjoint union) where the labels 
of $M_{k,1}$ are in $R_k$, and the labels $M_{k,2}$ are
in $S_k$. 
Lemma \ref{lemma-red-triangle},  the structure of the expressions and the properties of the infinite top down developments
show  that the normal order redex can only stay or descend, \ie\
  $h > k$ implies that $p_k$ is a prefix of $p_h$.

 Also, we have $R_k[S_k] \xrightarrow{I,M_k} R_k'[S']$, where $R_k[\cdot] \xrightarrow{M_{k,1}} R_k'[\cdot]$,
 and $S_k  \xrightarrow{I}  S'$.
 
 \noindent There are three cases: 
 \begin{itemize}
    \item The normal order reduction of $T_0$ halts, \ie, there is a maximal $k$. 
        Then obviously $T \xrightarrow{\TREE,*} T_k \xrightarrow{M_k,\neg\TREE} T'$.
    \item There is some $k$, such that $R_k = R_h$ for all $h \ge k$. In this case, 
    $T' = R_k'[s']$. The infinite development $T_0 \xrightarrow{I,M} T'$ will reduce infinitely
      often at the position of the hole, hence it
       will plug a $\tbot$ at position $p_k$  of $T'$, and so $T' = R_k'[\tbot]$.  But then $T'$ cannot converge,
         and so this case is not possible.
    \item The positions $p_k$ of the reduction contexts $R_k$ will grow indefinitely. Then there
     is an infinite path (together with the constructs and symbols) $p$ 
      such that  $p_k$ is a  prefix  of $p$ for every $k$.
      Moreover, $p$ is a position of $T'$. 
      The sets $M_{k,1}$ are an infinite ascending set \wrt\ $\subseteq$, hence there is a limit tree
      $T_{\infty}$ with $T \xrightarrow{\TREE,\infty} T_{\infty}$, which is exactly
      the limit of the contexts $R_k$ for $k \to \infty$. 
      There is a  reduction $T_{\infty} \xrightarrow{I,M'}  T'$  which is exactly $M' = \bigcup_{k} M_{k,1}$.
      Hence $T'$ has the  path $p$,     
        and we see that the tree $T'$ cannot have a normal order redex, since the search for such a redex 
         goes along $p$ and thus does not terminate. 
         This is a contradiction, and hence this case is not possible.\qedhere
 \end{itemize}
 \end{proof}

\begin{lemma}\label{lemma:commute-one-nsr-base} Let $T  \xrightarrow{I,M,\NTREE}  T_1    \xrightarrow{\TREE} T'$. 
Then the reduction can be commuted to $T \xrightarrow{\TREE}  T_3  \xrightarrow{I,M'} T'$ for some $M'$.
\end{lemma}
\begin{proof}
This follows since the $\xrightarrow{I,M,\NTREE}$-reduction cannot generate a new normal order $\TREE$-redex.
Hence, the normal order redex of $T_1$ also exists in $T$. The set $M'$ can be found by labeling $T$ with $M$, then performing
the $\TREE$-reduction where all labels of $M$ are kept and inherited by the reduction, except for those positions which are
removed by the reduction.
\end{proof}

\begin{lemma}\label{lemma:commute:ntree-maycon}
 Let $T \xrightarrow{I,\NTREE} T'$ and $T'\maycon_{\TREE}$. Then $T\maycon_{\TREE}$.
\end{lemma}
\begin{proof}
 We show by induction on $k$ that whenever  $T \xrightarrow{I,\NTREE} T'\xrightarrow{\TREE,k} T''$ where $T''$ is an $\LTREE$-WHNF,
 then $T\maycon_{\TREE}$.
The base case is $k=0$ and it holds by Lemma~\ref{lemma:T-successful-nsr-base-case}.
For the induction step let $T \xrightarrow{I,\NTREE} T' \xrightarrow{\TREE} T_0 \xrightarrow{\TREE,k} T''$.
We apply Lemma~\ref{lemma:commute-one-nsr-base} to $T \xrightarrow{I,\NTREE} T' \xrightarrow{\TREE} T_0$
and thus have $T \xrightarrow{\TREE} T_1 \xrightarrow{I,M} T_0 \xrightarrow{\TREE,k} T''$ for some $M$.

This situation can be depicted by the following diagram where the dashed reductions follow by Lemma~\ref{lemma:commute-one-nsr-base}:
\[
\xymatrix@R=6mm{
  T \ar@{-->}[d]_{\TREE}\ar[r]^{I,\NTREE} &T' \ar[r]^{\TREE} &T_0\ar[rr]^{\TREE,k} &&T''\\
  T_1\ar@{-->}[urr]_{I,M}
}
\]

If $M$ does not contain a normal order redex, then the induction hypothesis shows that $T_1\maycon_{\TREE}$ and thus
also $T\maycon_{\TREE}$. 
Now assume that $M$ contains a normal order redex. Then we apply Lemma~\ref{lemma:iteration-stops-maycon} to
$T_1 \xrightarrow{I,M} T_0$ (note that $T_0\maycon_{\TREE}$ and hence the lemma is applicable).
This shows that $T_1 \xrightarrow{\TREE,*} T_0''\xrightarrow{I,\NTREE} T_0$:

\[
\xymatrix@R=6mm{
  T \ar@{->}[d]_{\TREE}\ar[r]^{I,\NTREE} &T' \ar[r]^{\TREE} &T_0\ar[rr]^{\TREE,k} &&T''\\
  T_1\ar@{-->}[d]_{\TREE,*}\ar@{->}[urr]_{I,M}\\
  T_0''\ar@{-->}[uurr]_{I,\NTREE}
}
\]
Now we can apply the induction hypothesis
to $T_0''\xrightarrow{\NTREE} T_0 \xrightarrow{\TREE,k} T''$ and have $T_0''\maycon_{\TREE}$ which also shows $T\maycon_{\TREE}$.
\end{proof}

\begin{proposition}[Standardization]\label{prop:I-sequence-gives-maycon}
Let $T_1,\ldots,T_k$ be infinite trees such that 
$T_k \xrightarrow{I,M_{k-1}} T_{k-1} \xrightarrow{I,M_{k-2}} T_{k-2} \ldots \xrightarrow{I,M_1} T_1$,
where $T_1$ is an $\LTREE$-WHNF.
Then $T_k\maycon_{\TREE}$
\end{proposition}
 \begin{proof}
We use induction on $k$. If $k=1$ then the claim obviously holds since $T_k = T_1$ is already an $\LTREE$-WHNF.
For the induction step assume that $T_i \xrightarrow{I,M_{i-1}} T_{i-1}~\ldots \xrightarrow{I,M_1}~T_1$ 
 and $T_i\maycon_{\TREE}$. Let $T_{i+1} \xrightarrow{I,M_i} T_i$.
If $M_i$ contains a normal order redex, then we apply Lemma~\ref{lemma:iteration-stops-maycon} and have the following situation
 \[
 \xymatrix@C=15mm@R=6mm{
  T_{i+1}\ar[d]_{\TREE,*}\ar[rr]^{I,M_i}    && T_i \ar[rr]^{I,*}\ar[d]^{\TREE,*} && T_1
  \\
  T_{i+1}'\ar[urr]_{I,\NTREE}                &      & T_i'
  }
  \]             
where $T_i'$ is an $\LTREE$-WHNF. We apply Lemma~\ref{lemma:commute:ntree-maycon} to $T_{i+1}' \xrightarrow{I,\NTREE} T_i \xrightarrow{\TREE,*} T_i'$
which shows that $T_{i+1}'\maycon_{\TREE}$ and thus also $T_{i+1}\maycon_{\TREE}$.

\noindent If $M_i$ contains no normal order redex, we have
 \[
 \xymatrix@R=6mm{
  T_{i+1}\ar[r]^{I,\NTREE} & T_i \ar[rr]^{I,*}\ar[d]^{\TREE,*} && T_1
  \\
                                           & T_i'
  }
\]
where $T_i'$ is an $\LTREE$-WHNF. We apply Lemma~\ref{lemma:commute:ntree-maycon} to $T_{i+1} \xrightarrow{I,\NTREE}T_i \xrightarrow{\TREE,*} T_i'$
and have $T_{i+1}\maycon_{\TREE}$.
 \end{proof}

\subsection{\texorpdfstring{Equivalence of Tree-Convergence and $\LLR$-Convergence}{Equivalence of Tree-Convergence and LR-Convergence}}
In this section we will show that $\LLR$-convergence for finite expressions $s \in {\LETRECEXPR}$
coincides with convergence for the corresponding infinite tree $\ITT(s)$.

\begin{lemma}\label{lemma-cp-ll-inftreeeq} Let $s_1,s_2 \in {\LETRECEXPR}$ be finite expressions and $s_1 \to s_2$ by a rule (\rcp), or (\rlll).
Then $\ITT(s_1) = \ITT(s_2)$. \qed
\end{lemma}

\begin{lemma}\label{lemma-red-base-case} Let $s$ be a finite expression. If $s$ is an $\LLR$-WHNF then  $\ITT(s)$ is an answer.
If   $\ITT(s)$ is an answer, then $s \maycon_{\LR}$.
\end{lemma}
 \begin{proof}
If  $s$ is an $\LLR$-WHNF, then obviously, $\ITT(s)$ is a answer.
If $\ITT(s)$ is an answer, then the label computation of the infinite tree for the empty position using $s$, \ie\ $s\LABELCOMP_{\varepsilon}$,
must be $\lambda x$ or $c$ for some constructor. If we consider all the cases where the label computation for $s\LABELCOMP_{\varepsilon}$ ends
with such a label, we see that $s$ must be of the form $\NL[v]$, where $v$ is an $\LLR$-answer and the contexts $\NL$ are constructed
according to the grammar:
\[\begin{array}{ll}
\NL ::= &[\cdot] \OR \tletrec~\iEnv~\tin~\NL \\
        &\OR \tletrec~x_1 = \NL[\cdot], \{{x_i}={\NL[x_{i-1}]}\}_{i=2}^n,\iEnv~\tin~\NL[x_n]
\end{array}
\]

We show by induction that every expression $\NL[v]$, where $v$ is a value, can be reduced by normal order (\rcp)- and (\rllet)-reductions
to a WHNF in $\LLR$. We use the following induction measure $\mu$ on $\NL[v]$:
\[
\begin{array}{l@{\,}l@{\,}l}
\mu(v) &:= &0\\
\mu(\tletrec~\iEnv~\tin~\NL[v]) &:= & 1+\mu(\NL[v])\\
\mu(\tletrec~x_1 = \NL_1[v], \{{x_i}={\NL_i[x_{i-1}]}\}_{i=2}^n,\iEnv~\tin~\NL_{n+1}[x_n]) &:= &\\
\multicolumn{3}{r}{\mu(\NL_1[v]) + \mu(\tletrec~x_2 = \NL_2[v], \{{x_i}={\NL_i[x_{i-1}]}\}_{i=3}^n,\iEnv~\tin~\NL_{n+1}[x_n])}\\
\end{array}
\] 
The base case obviously holds, since $v$ is already an $\LLR$-WHNF. For the induction step assume that $\NL[v']\xrightarrow{\LR,\rcp \vee \rllet,*} t$, where $t$ is 
an $\LLR$-WHNF for every $\NL[v']$ with $\mu(\NL[v']) < k$. Let $\NL$, and $v$ be fixed, such that $\mu(\NL[v]) = k \geq 1$.
There are two cases:
\begin{itemize}
 \item $\NL[v] =\tletrec~\iEnv~\tin~\NL'[v]$. 
 If $\NL'$ is the empty context, then $\NL[v]$ is an $\LLR$-WHNF.
 Otherwise $\NL'[v]$ is a $\tletrec$-expression. Thus we can apply an $(\LR,(\mbox{\rlletin}))$-reduction to $\NL[v]$, where the measure $\mu$ is 
 decreased by one. The induction hypothesis shows the claim.
\item $\NL[v]=\tletrec~x_1 = \NL_1[v], \{{x_i}={\NL_i[x_{i-1}]}\}_{i=2}^n,\iEnv~\tin~\NL_{n+1}[x_n]$.
If $\NL_{n+1}[x_n]$ is a $\tletrec$-expression, then we can apply an $(\LR,\rlletin)$-reduction to $\NL[v]$ and the measure $\mu$ is
 decreased by 1. 
If $\NL_{n+1}$ is the empty context, and there is some $i$ such that $\NL_i$ is not the empty context, then we can choose the largest number $i$
and apply an $(\LR,\rllete)$-reduction to $\NL[v]$. Then the measure $\mu$ is strictly decreased and we can use the induction hypothesis.
If all the contexts $\NL_i$ for $i=1,\ldots,n+1$ are empty contexts, then either $\NL[v]$ is an $\LLR$-WHNF (if $v$ is a constructor application)
or we can apply an $(\LR,\rcp)$ reduction to obtain an $\LLR$-WHNF. \qedhere
\end{itemize}
\end{proof}

\begin{lemma}\label{lemma:single-red-inf}
Let $s\in\LETRECEXPR$ such that $s \xrightarrow{\LR,a} t$.  
If the reduction $a$ is (\rcp) or (\rlll) then $\ITT(s) = \ITT(t)$.  
If the reduction $a$ is (\rlbeta), (\rcasec), (\rcasein), (\rcasee) or (\rseqc), (\rseqin),(\rseqc) then 
$\ITT(s) \xrightarrow{I,M,a'} \ITT(t)$ for some $M$, where $a'$ is (\rbetaTr), (\rcaseTr), or (\rseqTr), respectively, and the set 
$M$ contains normal order redexes.
\end{lemma}
\begin{proof}
Only the latter needs a justification. Therefore, we label every redex in $\ITT(s)$ that is derived from the redex $s \xrightarrow{\LR} t$ by $\ITT(.)$.
This results in the set $M$ for $\ITT(s)$.
There will be at least one position in $M$ that is a normal order redex of $\ITT(s)$. 
\end{proof}

 \begin{proposition}\label{prop:maycon-expr-to-infexpr}
Let $s\in\LETRECEXPR$ such that $s\maycon_{\LR}$. Then $\ITT(s)\maycon_{\TREE}$.
 \end{proposition}
 \begin{proof}
We assume that  $s \xrightarrow{\LR,*} t$, where $t$ is a WHNF.
Using Lemma \ref{lemma:single-red-inf}, we see that there is
a finite sequence of reductions $\ITT(s) \xrightarrow{I,*} \ITT(t)$. 
Lemma~\ref{lemma-red-base-case} shows that $\ITT(r)$ is an $\LTREE$-WHNF.
Now Proposition~\ref{prop:I-sequence-gives-maycon} shows that $\ITT(s)\maycon_{\TREE}$.
 \end{proof}

We now consider the other direction and show that for every expression $s$: if $\ITT(s)$ converges, then $s$ converges, too.

\begin{lemma}\label{lemma-partition-R-red} 
Let $R$ be some reduction context, \st\ $\ITT(s) = R[T] \xrightarrow{\TREE,a'} R[T']$.
Then for $(a',a) \in \{(\mbox{\rbetaTr},\mbox{\rlbeta}), (\mbox{\rcaseTr}, \mbox{\rcase}), (\mbox{\rseqTr}, \mbox{\rseq})\}$
there exist expressions $s_1,s_2,s_3$ and an infinite tree $T'$ with $s \xrightarrow{\LR,\rlll,*} s_1 \xrightarrow{\LR,\rcp,0\vee 1} s_2 \xrightarrow{\LR,a} s_3 \text{ with } R[T']\xrightarrow{I,M} \ITT(s_3)$.

 \[
\xymatrix@R=6mm{
  s \ar@{->}[rr]^{\ITT(\cdot)}  \ar@{-->}[d]_{\LR,\rlll,*} &&  \ITT(s) = R[T] \ar@{->}[dd]^{\TREE,a'} \ar@/^70pt/@{-->}[ddd]^{I,M,a'}  \\
  s_1 \ar@{-->}[d]_{\LR,\rcp,0\vee 1}\\
  s_2 \ar@{-->}[uurr]^{\ITT(\cdot)}  \ar@{-->}[d]_{\LR,a} &&  R[T'] \ar@{-->}[d]^{I,M',a'}  \\
  s_3 \ar@{-->}[rr]^{\ITT(\cdot)}  && \ITT(s_3)  \\         
}
\]
\end{lemma}
\begin{proof}
Let $p$ be the position of the hole of $R$. We follow the label computation to $T$ along $p$ inside $s$
and show that the redex corresponding to $T$ can be found in $s$ after some $(\rlll)$ and $(\rcp)$ reductions. 
For applications, \tseq-expressions, and \tcase-expressions there is a one-to-one correspondence.
If the label computation shifts a position into a ``deep'' $\tletrec$, \ie\ 
$C[(\tletrec~\iEnv~\tin~s)\LABELCOMP_p] \mapsto C[(\tletrec~\iEnv~\tin~s\LABELCOMP_p)]$ where $C$ is non-empty, 
then a sequence of normal order (lll)-reduction moves the environment $\iEnv$ to the top of the expression, where
perhaps it is joined with a top-level environment of $C$. Let $s \xrightarrow{\LR,\rlll,*} s'.$
Lemma~\ref{lemma-cp-ll-inftreeeq} shows that $\ITT(s') = \ITT(s)$ and the label computation along $p$ for $s'$ requires fewer steps than the computation for $s$.
Hence this construction can be iterated and terminates. This yields a reduction sequence $s\xrightarrow{\LR,\rlll,*} s_1$ such that
the label computation along $p$ for $s_1$ does not shift the label into deep $\tletrec$s and where $\ITT(s) = \ITT(s_1)$ (see Lemma~\ref{lemma-cp-ll-inftreeeq}).
Now there are two cases: Either the redex corresponding to $T$ is also a normal order redex of $s_1$, or $s_1$ is of
the form $\tletrec~x_1=\lambda x.s', x_2 =x_1,\ldots,x_m=x_{m-1},\ldots R'[x_m] \ldots$. For the latter case an $(\LR,\rcp)$ reduction
is necessary before the corresponding reduction rule can be applied. Again Lemma~\ref{lemma-cp-ll-inftreeeq} assures that the infinite tree remains unchanged.
After applying the corresponding reduction rule, \ie\ $s_2 \xrightarrow{\LR,a} s_3$, the normal order reduction may have changed
infinitely many positions of $\ITT(s_3)$, while $R[T] \xrightarrow{\TREE,a'} R[T']$ does not change all these positions, but nevertheless
Lemma \ref{lemma:single-red-inf} shows that  there is a reduction $R[T]  \xrightarrow{I,M,a'} \ITT(s_3)$, and Lemma~\ref{lemma-red-triangle} 
shows that also $R[T'] \xrightarrow{I,M',a'} \ITT(s_3)$ for some $M'$. 
\end{proof} 

\begin{example}
An example for the proof of the last lemma is the expression $s$ defined as  $s:=\tletrec~x=(\lambda y.y)~x~\tin~ x$. Then
$\ITT(s) = (\lambda y.y)~((\lambda y.y)~((\lambda y.y) \ldots))$. The $\TREE$-reduction for $\ITT(s)$ is
$\ITT(s) \xrightarrow{\TREE,\rbetaTr} \ITT(s)$. On the other hand the normal order reduction of $\LLR$ reduces
to $s':=\tletrec~x=(\tletrec~y=x~\tin~y)~\tin~x$ and $\ITT(s') = \tBot$. To join the reductions we perform an $\xrightarrow{I,M}$-reduction
for $\ITT(s)$ where all redexes are labeled in $M$, which also results in $\tBot$.
\end{example}

\begin{proposition}\label{prop:tree-conv-impl-term-conv} Let $s$ be an  expression such that $\ITT(s)\maycon_{\TREE}$. 
Then  $s\maycon_{\LR}$.     
\end{proposition} 
\begin{proof} 
The precondition  $\ITT(s) \maycon_{\TREE}$ implies
 that there is a $\TREE$-reduction sequence of  $\ITT(s)$ to an $\LTREE$-WHNF.
The base case, where no $\TREE$-reductions are necessary, is treated in Lemma \ref{lemma-red-base-case}. 
In the general case, let $T \xrightarrow{\TREE,a'} T'$ be a $\TREE$-reduction.
Lemma \ref{lemma-partition-R-red} shows that  there are expressions $s',s''$ with 
$s \xrightarrow{\LR,\rlll,*} \xrightarrow{\LR,\rcp,0\vee 1}  s' \xrightarrow{\LR,a}  s''$,  
and  $T' \xrightarrow{I,M}  \ITT(s'')$.
Lemma \ref{lemma-R-star-forking} shows that $\ITT(s'')$ has a normal order $\TREE$-reduction to a WHNF 
where the number of $\TREE$-reductions is strictly smaller than the  number of $\TREE$-reductions of $T$ to a WHNF.
Thus  we can use induction on this length and obtain a normal order $\LR$-reduction of $s$ to a WHNF.
\end{proof}
\noindent Propositions \ref{prop:maycon-expr-to-infexpr} and \ref{prop:tree-conv-impl-term-conv} imply the theorem
\begin{theorem}\label{thm:may-equivalence}
Let $s$ be an $\LETRECEXPR$-expression. Then $s\maycon_{\LR}$ if and only if $\ITT(s)\maycon_{\TREE}$.  \qed
\end{theorem}
\subsection{\texorpdfstring{Equivalence of Infinite Tree Convergence and $\LNAME$-convergence}{Equivalence of Infinite Tree Convergence and LNAME-convergence}}
It is easy to observe that several reductions of $\LNAME$ do not change the infinite trees \wrt\ the 
translation $\ITT(\cdot)$:

\begin{lemma}\label{lemma:lapp-cp-it-identical}
Let $s_1,s_2 \in \LETRECEXPR$. Then 
  $s_1 \xrightarrow{\NAME,a} s_2$ for $a \in\{\rgcp,\rlapp,\rlcase,\rlseq\}$ implies $\ITT(s_1) = \ITT(s_2)$.\qed
\end{lemma}

\begin{lemma}\label{lemma:name-beta-is-tree-beta}
For $(a,a') \in \{(\mbox{\rbeta},\mbox{\rbetaTr}), (\mbox{\rcase},\mbox{\rcaseTr}), (\mbox{\rseq},\mbox{\rseqTr})\}$ it holds:

\noindent If $s_1  \xrightarrow{\NAME,a} s_2$ for $s_i \in \LETRECEXPR$, then  $\ITT(s_1) \xrightarrow{\TREE,a'} \ITT(s_2)$.
\end{lemma}
\begin{proof}
Let $s_1 := R_{\NAME}[s_1'] \xrightarrow{\NAME,a} R_{\NAME}[s_2'] = s_2$ where $s_1'$ is the redex of the $\xrightarrow{\NAME}$-reduction
and $R_{\NAME}$ is an $\LNAME$-reduction context. First one can observe that 
the redex $s_1'$ is mapped by $\ITT$ to a unique tree position within a tree reduction context in $\ITT(s_1)$.  

We only consider the (\rbeta)-reduction, since for a (\rcase)- or a (\rseq)-reduction the reasoning is completely analogous.
So let us assume that $s_1' = ((\lambda x.s''_1)~s''_2)$. Then $\ITT$ transforms $s_1'$ into a subtree 
$\sigma((\lambda x.\ITT(s''_1))~\ITT(s''_2))$ where $\sigma$ is a substitution replacing variables by infinite trees.
The tree reduction replaces $\sigma((\lambda x.\ITT(s''_1))~\ITT(s''_2))$ by $\sigma(\ITT(s''_1))[\sigma(\ITT(s''_2))/x]$, 
hence the lemma holds.
\end{proof}

\begin{proposition}
 Let $s\in\LETRECEXPR$ be an expression with $s\maycon_{\NAME}$. Then $\ITT(s)\maycon_{\TREE}$.
\end{proposition}
\begin{proof}
This follows by induction on the length of a normal order reduction of $s$. 
The base case holds since $\ITT(L[v])$, where $v$ is an $\LNAME$-answer is always an $\LTREE$-answer.
For the induction step we consider the first reduction of $s$, say $s \xrightarrow{\NAME} s'$. The induction hypothesis
shows $\ITT(s')\maycon_{\TREE}$. If the reduction $s \xrightarrow{\NAME}s'$
is  ($\NAME$,\rgcp), ($\NAME$,\rlapp), \mbox{($\NAME$,\rlcase),} or ($\NAME$,\rlseq), then Lemma~\ref{lemma:lapp-cp-it-identical}
implies $\ITT(s)\maycon_{\TREE}$. If $s \xrightarrow{\NAME,a} s'$ for $a \in \{\mbox{(\rbeta)},\mbox{(\rcase)},\mbox{(\rseq)}\}$, 
then Lemma~\ref{lemma:name-beta-is-tree-beta} shows $\ITT(s) \xrightarrow{\TREE} \ITT(s')$ and thus $\ITT(s)\maycon_{\TREE}$.
\end{proof}
\noindent Now we show the other direction:

\begin{lemma}\label{lem:name-lapp-cp=gleich}
Let $s\in\LETRECEXPR$  such that $\ITT(s)={\cal R}[T]$, where ${\cal R}$ is a tree reduction context and $T$ is a value or a redex.   
Then there are expressions  $s',s''$ such that $s \xrightarrow{\NAME,\rlapp\vee\rlcase\vee\rlseq\vee\rgcp,*} s'$, $\ITT(s') = \ITT(s)$,  $s' = R[s'']$,
$\ITT(L[s''] ) = T$, where $R = L[A[\cdot]]$ is a reduction context for some ${\cal L}$-context $L$ and some ${\cal A}$-context $A$,
 $s''$ may be an abstraction, a constructor application, or a  \rbeta-, \rcase- or \rseq-redex 
   iff $T$ is an abstraction, a constructor application, or a  \rbetaTr-, \rcaseTr- or \rseqTr-redex, respectively,
and the position $p$ of the hole in ${\cal R}$ is also 
the position of the hole in $A[\cdot]$.
\end{lemma}
\begin{proof*}
The tree $T$ may be an abstraction, a constructor application, an application, or a \rbetaTr-, \rcaseTr- or \rseqTr-redex in $R[T]$. 
Let $p$ be the position of the hole of ${\cal R}$. 
We will show by induction on the label-computation for $p$ in $s$ that there is a  reduction 
$s \xrightarrow{\NAME,\rlapp\vee\rlcase\vee\rlseq\vee\rgcp,*} s'$, where $s'$ is as claimed in the lemma. \\
We consider the label-computation for $p$ to explain the induction
measure, where we use the numbers of  the rules given in  \FIGURE~\ref{figure-inftree}.
Let $q$ be such that the label computation for $p$ is of the form $(10)^*.q$ and $q$ does not start with $(10)$. 
The measure for induction is  a tuple $(a,b)$, where $a$ is the  
length of $q$, and  $b \geq 0$ is
 the maximal number with   $q = (2\vee 4 \vee 6)^b.q'$.
The base case is $(a,a)$: Then the label computation is of the form $(2\vee 4 \vee 6)^*$ and 
indicates that $s$ is of the form $L[A[s'']]$ and satisfies the claim of the lemma.
For the induction step  we have to check several cases: 
\begin{enumerate}
\item The label computation starts with $(10)^*(2\vee 4 \vee 6)^+(10)$. 
 Then a normal-order (lapp), (lcase), or (lseq)
can be applied to $s$ resulting in 
$s_1$. The label-computation for $p$ \wrt\ $s_1$ is of the same
  length, and only applications of (10)
  and $(2\vee 4\vee 6)$ are interchanged, hence the second component of the measure is strictly decreased.
\item  The label computation starts with $(10)^*(2\vee 4 \vee 6)^*(11)$. 
Then a normal-order (\rgcp) can be applied to $s$ resulting in 
$s_1$. The length $q$ is strictly decreased by 1, and perhaps one (12)-step is changed into a (11)-step. Hence the 
measure is strictly decreased.
\end{enumerate}
In every case the claim on the structure of the contexts and $s'$ can easily be verified.
\qedhere
\end{proof*}

\begin{lemma}\label{lemma:tree-reduction-can-be-simulated-by-name}
 Let $s$ be an expression with $\ITT(s) \xrightarrow{\TREE} T$. 
 Then there is some $s'$ with $s \xrightarrow{\NAME,*} s'$ and $\ITT(s') = T$.
\end{lemma}
\begin{proof}
If $\ITT(s) \xrightarrow{\TREE} T$, then $\ITT(s) = {\cal R}[S]$ where ${\cal R}$ is a reduction context,
 $S$ a tree-redex with $S \xrightarrow{\TREE} S'$
and $T = {\cal R}[S']$. Let $p$ be the position of the hole of ${\cal R}$ in $\ITT(s)$. 
We  apply 
Lemma \ref {lem:name-lapp-cp=gleich}, which implies 
 that there is  a reduction $s \xrightarrow{\NAME,*} s'$, such that $\ITT(s) =  \ITT(s')$ and  
$s' = R[s'']$ where $R = L[A[\cdot]]$ is a reduction context and $\ITT(L[s''])$ is a beta-, case-, or seq-redex.
It is obvious that $s' = L[A[s'']] \xrightarrow{\NAME,a}  t$.
Now one can verify that $ \ITT(t) = T$ must hold.
\end{proof}

\begin{proposition}
 Let $s$ be an expression with $\ITT(s)\maycon_\TREE$. Then $s\maycon_{\NAME}$.
\end{proposition}
\begin{proof}
We use induction on the length $k$ of a tree reduction $\ITT(s) \xrightarrow{{\TREE},k} T$, where $T$ is an $\LTREE$-answer.
For the base case it is easy to verify that if $\ITT(s)$ is an $\LTREE$-answer, then $s \xrightarrow{\NAME,\rgcp,*} L[v]$ for 
some ${\cal L}$-context $L$ and some $\LNAME$-value $v$. Hence we have $s\maycon_{\NAME}$.
The induction step follows by repeated application of Lemma~\ref{lemma:tree-reduction-can-be-simulated-by-name}.
\end{proof}

\begin{corollary}\label{cor:infinite-tree-conv=name-conv} For all $\LETRECEXPR$-expressions $s$: 
 $s\maycon_{\NAME}$ if, and only if $\ITT(s)\maycon_\TREE$. \qed
\end{corollary}

\begin{theorem}\label{theo:leqneed-equals-leqname}
 $\lec_{\NAME}\ =\ \lec_{\LR}$.
\end{theorem}
\begin{proof}
In Corollary~\ref{cor:infinite-tree-conv=name-conv} we have shown that $\LNAME$-convergence is equivalent to infinite tree convergence. 
In Theorem~\ref{thm:may-equivalence} we have shown that $\LLR$-convergence is equivalent to infinite tree convergence. Hence, 
$\LNAME$-convergence and $\LLR$-convergence are equivalent, which further implies that both contextual preorders and also
the contextual equivalences are identical.
\end{proof}

\begin{corollary}\label{cor:W-fully-abs}
The translation $W$ is convergence equivalent and fully abstract. \qed
\end{corollary}

Since $W$ is the identity on expressions, this  implies:
\begin{corollary}\label{cor:W-iso}
$W$ is an isomorphism according to Definition \ref{def:translation-compo-etal}.\qed
\end{corollary}

A further consequence of our results 
 is that the general copy rule (gcp) is a correct program transformation in $\LLR$.
This is a novel result, since in previous work only special cases were proved correct.
\begin{proposition}\label{prop-gcp-correct}
 The program transformation (gcp) is correct in $\LNAME$ and $\LLR$.
\end{proposition}
\begin{proof}
Correctness of (gcp) in $\LNAME$ holds, since for $s,t \in \LETRECEXPR$ with
 $s \xrightarrow{gcp} t$ and for any context $C$: $\ITT(C[s]) = \ITT(C[t])$. Hence
 Corollary~\ref{cor:infinite-tree-conv=name-conv} implies that
 $C[s] \maycon_{\NAME} \iff C[t] \maycon_{\NAME}$ and
thus $s \simc_{\NAME} t$. Theorem~\ref{theo:leqneed-equals-leqname} finally also shows $s \simc_{\LR} t$.
\end{proof}

%%%%%%%%%%%%%%%%%%%%%%%%%%%%%%%%%%%%%%%%%%%%%%%%%%%%%%%%%%%%%%%%%%%%%%%%%%%%%%
%% SECTION 6: THE TRANSLATION N
\section{\texorpdfstring{The Translation $\transComp:\LNAME \to \LLAZYCC$}{The Translation N}}\label{sec:NAME-to-LAZY}
We use multi-fixpoint combinators as defined in~\cite{goldberg:05} to translate
letrec-expressions $\LETRECEXPR$ of the calculus $\LNAME$ into equivalent ones without a \tletrec. The translated
expressions are $\LAMBDAEXPR$ and belong to the calculus $\LLAZYCC$.

\begin{definition}\label{def:multifixpoint}
Given $n \geq 1$, a family of $n$ fixpoint combinators $\Y_i^n$ for $i = 1,\ldots,n$ can be defined as follows: 
\[\begin{array}{lclll}
\Y_i^n & := & \lambda f_1,\ldots, f_n . (&(\lambda x_1, \ldots,x_n . f_i &(x_1~x_1 ~\ldots x_n)~\ldots~(x_n~x_1 ~\ldots x_n))\\
     & &                               &(\lambda x_1, \ldots,x_n . f_1 &(x_1~x_1 ~\ldots x_n)~\ldots~(x_n~x_1 ~\ldots x_n))\\
   & &                               &    \ldots\\
   & &                               &  (\lambda x_1, \ldots,x_n . f_n &(x_1~x_1 ~\ldots x_n)~\ldots~(x_n~x_1 ~\ldots x_n)))
\end{array}\]
\end{definition}

The idea of the translation is to replace  
 $\tletrx{x_1 = s_1, \ldots, x_n = s_n}{t}$ by $t[B_1/x_1,\ldots, B_n/x_n]$
 where $B_i := \Y_i^n~F_1 \ldots F_n$ 
 and $F_i := \lambda x_1, \ldots, x_n . s_i$.

In this way the fixpoint combinators implement  the generalized fixpoint property: 
$\Y_i^n~F_1 \ldots F_n \sim F_i~(\Y_1^n~F_1 \ldots F_n) \ldots (\Y_n^n~F_1 \ldots F_n)$. However, our translation uses modified expressions, as shown below.

Consider the expression \mbox{$(\Y_i^n~F_1~\ldots~F_n)$}. After expanding the notation for $\Y_i^n$ we obtain the expression
$((\lam f_1, \ldots, f_n. (X_i~X_1~\ldots~X_n))~F_1~\ldots~F_n)$
where $X_i$ can be expanded to $X_i =  \lam x_1 \ldots x_n. (f_i~(x_1~x_1~\ldots~x_n)~\ldots~(x_n~x_1~\ldots~x_n))$.
If we reduce further then we get:
$$
\begin{array}{ll}
 (\lam f_1, \ldots, f_n . (X_i~X_1~\ldots~X_n))~F_1~\ldots~F_n \annRed{\rnbet,*} (X'_i~X'_1~\ldots~X'_n), 
\\
\mbox{where } X'_i = \lam x_1 \ldots x_n. (F_i~(x_1~x_1~\ldots~x_n) \ldots (x_n~x_1~\ldots x_n))
\end{array}
$$

We take the latter expression as the definition of the
multi-fixpoint translation, where we avoid substitutions 
and instead generate (\rnbeta)-redexes
which ensures that contexts are mapped to contexts
\begin{definition}
The translation $\transComp : \LNAME \to \LLAZYCC$ is recursively defined as:
\begin{itemize}
\item
 $\transComp\tletrx{x_1 = s_1, \ldots, x_n = s_n}{t}  = $ \\
  $(\lambda x_1',\ldots,x_n'. (\lambda x_1,\ldots x_n.\transComp(t))~U_1 \ldots~U_n)~X_1' \ldots X_n'$\\[1mm]
$\begin{array}{@{\qquad}lrcl}
\mbox{where }&  \multicolumn{3}{l}{\text{$x_1',\ldots,x_n'$ are fresh variables}}\\
             & U_i   &=& x_i'~x_1' \ldots x_n',
\\
 &  X'_i  &=& \lam x_1 \ldots x_n. F_i (x_1~ x_1 \ldots x_n) \ldots (x_n~x_1 \ldots x_n), 
\\
&   F_i &=& \lambda x_1, \ldots, x_n . \transComp(s_i).

\end{array}
$
\item $\transComp(s~t) = (\transComp(s)~\transComp(t))$
\item $\transComp(\tseq~s~t) = (\tseq~\transComp(s)~\transComp(t))$
\item $\transComp(c~s_1~\ldots~s_{\ari(c)}) = (c~\transComp(s_1)\ldots\transComp(s_{\ari(c)}))$
\item $\transComp(\lambda x.s) = \lambda x.\transComp(s)$
\item $\transComp(\tcase_T~s~\tof~alt_1~\ldots~alt_{|T|})= \tcase_T~\transComp(s)~\tof~\transComp(alt_1)~\ldots~\transComp(alt_{|T|})$
\item for a case-alternative: $\transComp(c~x_1~\ldots~x_{\ari(c)} \to s) = (c~x_1~\ldots~x_{\ari(c)} \to \transComp(s))$ 
\item $\transComp(x) = x$.

\end{itemize}
We extend $\transComp$ to contexts by treating the hole as a constant, \ie\ $\transComp(\chole) = \chole$.
This is consistent, since the hole is not duplicated by the translation.
\end{definition}

\subsection{\texorpdfstring{Convergence Equivalence of $\transComp$}{Convergence Equivalence of N}}\label{appendix:conv-equiv-trans}
In the following we will also use the context class ${\cal B}$, defined as   
${\cal B} = L[{\cal B}] \mid A[{\cal B}] \mid [\cdot]$ (${\cal L}$- and ${\cal A}$-contexts 
are defined as before in Sect.~\ref{subsec:name-calc}).

The proof of convergence equivalence of the translation $\transComp$ may be
performed directly, but it would be complicated due to the additional (\rnbeta)-reductions
required in $\LLAZYCC$. For this technical reason we provide a second translation
$\transN$, which requires a special treatment for the translation of
contexts and uses a substitution function $\sigma$:

\begin{definition}\label{def:simp-fixpoint}
The translation $\transN : \LNAME \to \LLAZYCC$  is recursively defined as:
\begin{itemize}
\item $\transN\tletrx{x_1 = s_1, \ldots, x_n = s_n}{t}  =  \sigma(\transN(t))$,
where 
\[\begin{array}{rcl}
  \sigma &= & \{x_1 \mapsto U_1, \ldots x_n \mapsto U_n\}
 \\ 
   U_i   &=& (X'_i~X'_1~\ldots~X'_n),
\\
   X'_i  &=& \lam x_1 \ldots x_n. F_i (x_1~x_1~\ldots~x_n)~\ldots~(x_n~x_1~\ldots~x_n), 
\\
   F_i &=& \lambda x_1, \ldots, x_n . \transN(s_i).
\end{array}
\]
\item $\transN(s~t) = (\transN(s)~\transN(t))$
\item $\transN(\tseq~s~t) = (\tseq~\transN(s)~\transN(t))$
\item $\transN (c~s_1 \ldots s_n) = (c~\transN(s_1) \ldots \transN(s_n))$
\item $\transN(\lambda x.s) = \lambda x.\transN(s)$
\item $\transN(\tcase_T~s~\tof~alt_1~\ldots~alt_{|T|})= \tcase_T~\transN(s)~\tof~\transN(alt_1)~\ldots~\transN(alt_{|T|})$
\item for a case-alternative: 
$\transN(c~x_1~\ldots~x_{\ari(c)} \to s) = (c~x_1~\ldots~x_{\ari(c)} \to \transN(s))$ 
\item $\transN(x) = x$.
\end{itemize}
The extension of $\transN$ to contexts is done only for ${\cal B}$-contexts and requires an extended
  notion of contexts that are accompanied by an additional substitution,
  \ie\ a ${\cal B}$-context translates into a pair 
   $(C,\sigma)$ of a context $C$ and a substitution $\sigma$ acting as a
  function on expressions. Filling the hole of $(C,\sigma)$ by an
expression $s$ is by definition $(C,\sigma)(s) =
C[\sigma(s)]$.
The translation for ${\cal B}$-contexts is  defined as

\begin{center}
\begin{tabular}{llp{11.5cm}}
$\transN(C)$ &$= (C',\sigma)$, 
 & where $C'$ and $\sigma$ are
 calculated by applying $N'$ to $C$: for calculating $C'$ the hole 
of $C$ is treated 
as a constant,
and $\sigma$ is the combined substitution affecting the hole of $C'$.
\end{tabular}
\end{center}
\end{definition}

This translation does not duplicate holes of contexts.\enlargethispage{\baselineskip}

\begin{lemma}\label{lemma:T-prime-sim-T}
The translation $\transComp$ is equivalent to $\transN$ on expressions, that is for all $\LETRECEXPR$-expressions $s$ the equivalence   
$\transComp(s) \sim_{\LCC} \transN(s)$ holds.
\end{lemma}
\begin{proof}
 This follows from the definitions and correctness of (\rnbeta)-reduction in $\LLAZYCC$ by Theorem \ref{cor:lcc-red-rules-correct}.
\end{proof}
We first prove that the translation $\transN$ is convergence-equivalent.
Due to Lemma~\ref{lemma:T-prime-sim-T} this will also imply that $\transComp$ is convergence-equivalent.
All reduction contexts $L[A[\cdot]]$ in $\LNAME$ translate into reduction contexts  $R_\LCC$ in $\LLAZYCC$
since removing the case of \tletrec\ from the definition of a reduction
context in $\LNAME$ results in the reduction context definition in $\LLAZYCC$.\newpage

\noindent However, this cannot be reversed, since
a translated expression of $\LNAME$ may have a redex in $\LLAZYCC$, but it is not
a normal order redex in $\LNAME$ since (\rlapp), (\rlseq), or (\rlcase) reductions must be performed first to
shift $\tletrec$-expressions out of an application, a $\tseq$-expression, or a $\tcase$-expression.
The lemma below gives a more precise characterization of this relation: 
\begin{lemma}\label{lem:red-context-trans}
If $L[A[\cdot]]$ is a reduction context in $\LNAME$, then
$\transN(L[A[\cdot]]) =  R[\sigma(\cdot)]$, where $R$ is a reduction context in $\LLAZYCC$ 
and $\sigma$ is a substitution.

If $R$ is a reduction context in $\LLAZYCC$, and 
$\transN(C') = (R,\sigma)$ for 
some substitution $\sigma$ and some context $C'$ in $\LNAME$, then  $C'$ is
a ${\cal B}$-context.
\end{lemma}
\begin{proof}
The first claim can be shown by structural induction on  the context $L[A[\cdot]]$. It holds, since applications are translated into applications,
$\tseq$-expressions are translated into $\tseq$-expressions, $\tcase$-expressions are translated into $\tcase$-expressions,
and \tletrec-expressions are translated into substitutions.

The other part can be shown by induction on the number of translation steps.
It is easy to observe that the definition of a reduction context in
$\LNAME$ does not descend into $\tletrec$-expressions below applications, $\tseq$-, and $\tcase$-expressions.
For instance, in $(\tletrx{\iEnv}{((\lambda x. s_1)~ s_2)}~s_3)$ the reduction contexts are $[\cdot]$ and 
$([\cdot]~s_3)$ and the redex is (\rlapp), \ie\ the reduction context does not reach $((\lam x. s_1)~s_2)$. 
In general, applications, \tseq-, and \tcase-expressions in such cases appear in 
 $\cal B$-contexts, as defined above.
By examining the expression definition we observe that these (\rlapp), (\rlseq), and/or (\rlcase)-redexes are the only cases
where non-reduction contexts may be translated into reduction contexts.
\end{proof}

\begin{lemma}\label{lem:translate-values}
Let $\transN(s) = t$. Then:
\begin{enumerate}
\item \label{item:abs}  If $s$ is an abstraction then so is $t$. 
\item \label{item:cons} If $s = (c~s_1 \ldots s_{\ari(c)})$ then $t =    (c~t'_1 \ldots t'_{\ari(c)})$. 
\end{enumerate}
\end{lemma}
\begin{proof}
This follows by examining the translation $\transN$.
\end{proof}

We will now use reduction diagrams to show the correspondence of $\LNAME$-reduction and $\LLAZYCC$-reduction
\wrt\ the translation $\transN$.

\subsubsection*{\texorpdfstring{Transferring $\LNAME$-reductions into $\LLAZYCC$-reductions}{Transferring LNAME-reductions into LCC-reductions}} ~\\
In this section we analyze how normal order reduction in $\LNAME$ can be transferred into
$\LLAZYCC$ via $\transN$. We illustrate this by using reduction diagrams.
For $s \annRed{\NAME} t$ we analyze how the reduction transfers to
$\transN(s)$. The cases are on the rule used in $s\annRed{\NAME} t$:
\begin{itemize}
\item $(\rbeta)$ Let $s = R[(\lambda x. s_1)~s_2]$ be
an expression in $\LNAME$, where $R$ is a reduction context.
We observe that in $\LNAME$:
$
s \annRed{\NAME} t =  R[s_1[s_2/x]]
$.
Let $\transN(R[\cdot]) = (R',\sigma)$. Then the translations for $s$ and $t$ are as follows:
$$
\begin{array}{l}
\transN(s) = R'[\sigma(\transN((\lam x. s_1)~s_2))]
     = R'[(\lam x.\sigma(\transN(s_1)))~\sigma(\transN(s_2))]
\\[1.1ex]
\transN(t) = \transN(R[s_1[s_2/x]]) 
      = R'[\sigma(\transN(s_1[s_2/x]))]
      = R'[\sigma(\transN(s_1))[\sigma(\transN(s_2))/x]]
\end{array}
$$
\begin{figure}
\[
\begin{array}{c}
\begin{array}{ccc}
\xymatrix@R=6mm@C=8mm{
\cdot \ar@{->}[r]^{\transN} \ar@{->}[d]_{\footnotesize\txt{$\NAME$,\\$\rbeta$}} & \cdot \ar@{->}[d]^{\LCC,\rnbet}
\\
\cdot \ar@{->}[r]_{\transN} & \cdot}
&
\xymatrix@R=6mm@C=8mm{
 \cdot \ar@{->}[r]^{\transN}  \ar@{->}[d]_{\footnotesize\txt{$\NAME$,\\ $\rgcp$}} & \cdot \ar@{->}[d]^{\LCC,\rnbet,2n~~} \\
\cdot  \ar@{->}[r]_{\transN} & \cdot 
}
&
\xymatrix@R=6mm@C=8mm{
\cdot \ar@{->}[r]^{\transN} \ar@{->}[d]_{\footnotesize\txt{$\NAME$,\\$\rcase$}} & \cdot \ar@{->}[d]^{\LCC,\rncase}  
 \\ 
\cdot \ar@{->}[r]_{\transN}  & \cdot
}
\\[1.1ex]
\text{(1)} &\text{(2)} &\text{(3)}
\end{array}
\\[1.4ex]
\begin{array}{cccc}
\xymatrix@R=6mm@C=8mm{
 \cdot \ar@{->}[r]^{\transN} \ar@{->}[d]_{\footnotesize\txt{$\NAME$,\\$\rlapp$}} & \cdot \\
\cdot  \ar@{->}[ur]_{\transN} 
}
&
\xymatrix@R=6mm@C=8mm{
\cdot \ar@{->}[r]^{\transN} \ar@{->}[d]_{\footnotesize\txt{$\NAME$,\\$\rlcase$}} & \cdot \\  
\cdot \ar@{->}[ur]_{\transN} 
}
&
\xymatrix@R=6mm@C=8mm{
 \cdot \ar@{->}[r]^{\transN} \ar@{->}[d]_{\footnotesize\txt{$\NAME$,\\$\rlseq$}} & \cdot \\  
\cdot  \ar@{->}[ur]_{\transN} 
}
&
\xymatrix@R=6mm@C=8mm{
 \cdot \ar@{->}[r]^{\transN}  \ar@{->}[d]_{\footnotesize\txt{$\NAME$,\\$\rseq$}} & \cdot \ar@{->}[d]^{\LCC,\rnseq}\\
  \cdot \ar@{->}[r]_{\transN} & \cdot
}
\\[1.1ex]
\text{(4)} &\text{(5)} &\text{(6)}&\text{(7)}
\end{array}
\end{array}
\]
\caption{Diagrams for transferring reductions between $\LNAME$ and $\LLCC$}\label{fig:N-diagrams}
\end{figure}

Since $R'$ is a reduction context in $\LLAZYCC$, this shows
$\transN(s) \annRed{\LCC,\rnbet} \transN(t)$.
Thus we have the diagram (1)  in Figure~\ref{fig:N-diagrams}. 
\item $(\rgcp)$
Consider the (\rgcp) reduction. Without loss of generality we assume that $x_1$
is the variable that gets substituted: 
$$
\begin{array}{l}
s = L[\tletrec~x_1=s_1,\ldots,x_n=s_n~\tin~R[x_1]]
\annRed{\NAME,\rgcp}
\\
t = L[\tletrec~x_1=s_1,\ldots,x_n=s_n~\tin~R[s_1]]
\end{array}
$$
Let 
$\transN(L) = ([\cdot],\sigma_L)$, 
$\transN(\tletrec~x_1=s_1,\ldots,x_n=s_n~\tin~[\cdot]) = ([\cdot],\sigma_{\iEnv})$, and
$\transN(R) = (R',\sigma_R)$ where $R'$ is a reduction context.
Then
$$\begin{array}{ll}
\transN(s)  &= \sigma_L(\sigma_{\iEnv}(R'[\sigma_{R}(x_1)]))
      = \sigma_L(\sigma_{\iEnv}(R'))[\sigma_L(\sigma_{\iEnv}(\sigma_{R}(x_1)))]
\\
      &= \sigma_L(\sigma_{\iEnv}(R'))[\sigma_L(\sigma_{\iEnv}(x_1))]
  \end{array}$$
where the last step follows, since $x_1$ cannot be substituted by $\sigma_{R}$,
and
$$\transN(t) = \sigma_L(\sigma_{\iEnv}(R'))[\sigma_L(\sigma_{\iEnv}(\transN(s_1)))]$$ where it is again necessary
to observe that $\sigma_{R}(s_1) = s_1$ must hold. The context $R''= \sigma_L(\sigma_{\iEnv}(R'))$ must be 
a reduction context, since $R'$ is a reduction context. This means that we need to show
that $R''[\sigma_L(\sigma_{\iEnv}(x_1))] \xrightarrow{\LCC,*} R''[\sigma_L(\sigma_{\iEnv}(\transN(s_1)))]$ holds.

By Definition~\ref{def:simp-fixpoint} of the translation~$\transN$ we have
$\sigma_L(\sigma_{\iEnv}(x_1)) = U_1 = (X'_1 X'_1 \ldots
  X'_n)$,
where $X'_i = \lam x_1 \ldots x_n. F_i (x_1 x_1  \ldots x_n) \ldots (x_n x_1
  \ldots x_n)$, and $F_i = \lam x_1, \ldots, x_n. \sigma_L(\transN(s_i))$, 
\ie,  $\transN(t) = R''[U_1]$.

Performing the applications, we transform $U_1$ in $2n$ steps as
$$
\begin{array}{ll}
&(\lam x_1, \ldots, x_n. (F_1 (x_1 x_1 \ldots x_n) \ldots (x_n x_1 \ldots x_n)))~X'_1~\ldots~X'_n 
\\
\annRed{\rnbet,n}
&F_1~(X'_1 X'_1 \ldots X'_n)~\ldots~(X'_n X'_1 \ldots X'_n) 
\\
=&
(\lam x_1, \ldots, x_n. \sigma_L(\transN(s_1)))~(X'_1 X'_1 \ldots X'_n)~\ldots~(X'_n X'_1 \ldots X'_n)
\\
\annRed{\rnbet,n}
&
\sigma_L(\transN(s_1))[U_1/x_1, \ldots, U_n/x_n].
\end{array}
$$
Obviously, for all reduction contexts in $\LLAZYCC$ holds: $r_1 \xrightarrow{\LCC} r_2$ implies $R[r_1] \xrightarrow{\LCC} R[r_1]$.
Hence $\transN(s) \xrightarrow{\LCC,\rnbet,2n} R''[\sigma_L(\transN(s_1))[U_1/x_1, \ldots, U_n/x_n]]$  and
since $x_1, \ldots, x_n$ cannot occur free in $L$, the last expression is 
the same as $R''[\sigma_L(\sigma_{\iEnv}(\transN(s)))]$.
Thus we obtain the diagram (2)  in Figure~\ref{fig:N-diagrams}, 
where $n$ is the number of bindings in the $\tletrec$-subexpression where the copied binding is.
\item $(\rcase)$ The diagram for this case is  marked (3) in Figure~\ref{fig:N-diagrams}.
The case is similar to  (\rbeta):
\mbox{$s =R[\tcase_T~(c~\vect{s_i}) \ldots ((c~\vect{x_i}) \to r) \ldots ]$} 
$\annRed{\NAME} R[r[s_1/x_1, \ldots, s_{\ari(c)}/x_{\ari(c)} ]] = t.$
Let $\transN(R[\cdot]) = (R',\sigma)$. Then the translations for $s$ and $t$ are as follows:
\[
\begin{array}{@{\quad\qquad}l@{~}l@{}}
\transN(s) &= R'[\sigma(\transN(\tcase_T~(c~s_1 \ldots s_{\ari(c)}) \ldots ((c~x_1 \ldots x_{\ari(c)}) \to r) \ldots ))] 
\\
            &= R'[\tcase_T~(c~\sigma(\transN(s_1))) \ldots \sigma(\transN(s_{\ari(c)})) \ldots ((c~x_1 \ldots x_{\ari(c)}) \to \sigma(\transN(r))) \ldots]
\end{array}
\]
\[\begin{array}{@{\quad\qquad}l@{~}l@{}}
\transN(t)  &= \transN(R[r[s_1/x_1, \ldots, s_{\ari(c)}/x_{\ari(c)} ]]) 
\\
 &=     R'[\sigma(\transN(r[s_1/x_1, \ldots, s_{\ari(c)}/x_{\ari(c)} ]))] 
\\
 &=     R'[\sigma(\transN(r))[\sigma(\transN(s_1))/x_1, \ldots, \sigma(\transN(s_{\ari(c)}))/x_{\ari(c)} ]]
\end{array}\qquad\qquad
\]

Since $R'$ is a reduction context in $\LLAZYCC$, this shows
$\transN(s) \annRed{\LCC} \transN(t)$.
\item $(\rlapp)$ The reduction is $
R[\tletrx{\iEnv}{s_1}~s_2] \annRed{\NAME} R[\tletrx{\iEnv}{(s_1~s_2)}]$.
Since free variables of $s_2$ do not depend on $\iEnv$, the translation of $s_2$
does not change by adding $\iEnv$.
\Ie, for $\transN(R) =(R',\sigma_R)$ and $\transN\tletrx{\iEnv}{[\cdot]} = ([\cdot],\sigma_{\iEnv})$
we have 
\[\eqalign{
  \transN(R[(\tletrec~\iEnv~\in~s_1)~s_2]) 
&= R'[\sigma_R(\sigma_{Env}(\transN(s_1))~\transN(s_2))]\cr
&= R'[\sigma_R(\sigma_{\iEnv}(\transN(s_1~\transN(s_2))))]\cr
&= \transN(R[(\lr \iEnv \inn (s_1~s_2))])
}
\]
and thus the diagram for this case is as the one marked (4) in Figure~\ref{fig:N-diagrams}.
\item $(\rlcase)$
The case is analogous to that of (\rlapp), with the diagram marked as (5) in Figure~\ref{fig:N-diagrams}.
\item $(\rlseq)$ The case is analogous to (\rlapp) and (\rlcase), with the diagram (6) in Figure~\ref{fig:N-diagrams}.
\item $(\rseq)$
$s = R[\tseq~v~s_1] \annRed{\NAME} R[s_1] = t$ where $v$ is an abstraction or a constructor application

Let $\transN(R[\cdot]) = (R',\sigma)$. Then the translations for $s$ and $t$ are as follows:
$$
\begin{array}{l}
\transN(s) = R'[\sigma(\transN(\tseq~v~s_1))]=R'[\tseq~\sigma(\transN(v))~\sigma(\transN(s_1))] 
\\[1.1ex]
\transN(t) = R'[\sigma(\transN(s_1))]
\end{array}
$$

By Lemma~\ref{lem:translate-values}~$\transN(v)$ is a value in $\LLAZYCC$ (which cannot be changed by the substitution $\sigma$)
and thus $\transN(s) \xrightarrow{\LCC,\rnseq} \transN(t)$.
The diagram for this case is (7) in Figure~\ref{fig:N-diagrams}.
\end{itemize}

We inspect how WHNFs and values  of both calculi are related \wrt\ $\transN$:

\begin{lemma}\label{lemma:irreduble-in-lname-non-whnf-implies-irreducible-non-whnf-in-lcc}
Let $s$ be irreducible in $\LNAME$, but not an $\LNAME$-WHNF. 
Then $\transN(s)$ is irreducible in $\LLAZYCC$ and also not an $\LLAZYCC$-WHNF.
\end{lemma}
\begin{proof}
Assume that expression $s$ is irreducible in $\LNAME$ but not an $\LNAME$-WHNF. There are three cases:
\begin{enumerate}
 \item Expression $s$ is of the form $R[x]$ where $x$ is a free variable in $R[x]$,
then let $\transN(R)=(R',\sigma)$ and thus $\transN(s) = R'[\sigma(x)]$.  Since $\sigma$ only substitutes bound variables, we get
$\sigma(x) = x$ and thus $\transN(s) = R'[x]$ where $x$ is free in $R'[x]$. Hence $\transN(s)$ cannot be an $\LLAZYCC$-WHNF and it is irreducible in $\LLAZYCC$.
\item Expression $s$ is of the form $R[\tcase_T~(c~s_1~\ldots~s_{\ari(c)})~\tof~alts]$, but $c$ is not of type $T$.
Let $\transN(R)=(R',\sigma)$. Then $\transN(s) = R'[\tcase_T~(c~\sigma(\transN(s_1))\ldots\sigma(\transN(s_{\ari(c)})))~\tof~alts']$
which shows that $\transN(s)$ is not an $\LLAZYCC$-WHNF and irreducible in $\LLAZYCC$.
\item Expression $s$ is of the form $R[((c~s_1~\ldots~s_{\ari(c)})~r)]$. Then again $\transN(s)$ is not an $\LLAZYCC$-WHNF and irreducible. \qedhere
\end{enumerate}
\end{proof}
\begin{lemma}\label{lem:T-WHNF} 
Let $s\in\LETRECEXPR$. Then  $s$ is an $\LNAME$-WHNF iff $\transN(s)$ is an $\LLAZYCC$-WHNF.
\end{lemma}
\begin{proof}
If $s = L[\lam x. s']$ or $s=L[(c~s_1 \ldots s_{\ari(c)})]$ then $\transN(s) = \lam x. \sigma(\transN(s'))$ 
or $\transN(s) = (c~\sigma(\transN(s_1)) \ldots \sigma(\transN(s_{\ari(c)})))$ respectively, both of which are
$\LLAZYCC$-WHNFs. 

For the other direction assume that $\transN(s)$ is an abstraction or a constructor application.
The analysis of the reduction correspondence in the previous paragraph
shows that $s$ cannot have a normal order redex in $\LNAME$, since otherwise $\transN(s)$ cannot be an $\LLAZYCC$-WHNF. 
Lemma \ref{lemma:irreduble-in-lname-non-whnf-implies-irreducible-non-whnf-in-lcc} shows that $s$ cannot be irreducible in $\LNAME$, but
not an $\LNAME$-WHNF. Thus $s$ must be an $\LNAME$-WHNF.
\end{proof}

\subsubsection*{\texorpdfstring{Transferring $\LLAZYCC$-reductions into $\LNAME$-reductions}{Transferring LCC-reductions into LNAME-reductions}}~\\
We will now analyze how normal order reductions for $\transN(s)$ can be transferred
into normal order reductions for $s$ in $\LNAME$.

Let $s$ be an $\LETRECEXPR$-expression and $\transN(s) \xrightarrow{\LCC} t$.
We split the argument into three cases based on whether or not a normal order reduction is applicable to $s$:
\begin{itemize}
 \item If $s \xrightarrow{(\NAME)} r$, then we can use the already developed diagrams, since normal-order reduction
in both calculi is unique.
 \item $s$ is a WHNF. This case cannot happen, since then $\transN(s)$ would also be a WHNF (see Lemma~\ref{lem:T-WHNF}) and thus irreducible.
 \item $s$ is irreducible but not a WHNF. Then Lemma~\ref{lemma:irreduble-in-lname-non-whnf-implies-irreducible-non-whnf-in-lcc} implies
that $\transN(s)$ is irreducible in $\LLAZYCC$ which contradicts the assumption $\transN(s) \xrightarrow{\LCC} t$. Thus this case is impossible.\medskip
\end{itemize}

\noindent We summarize the diagrams in the following lemma:
\begin{lemma}\label{lem:T-diagrams}
Normal-order reductions in $\LNAME$ can be transferred into
reductions in $\LLAZYCC$, and vice versa, by the diagrams in Figure~\ref{fig:N-diagrams}. \qed
\end{lemma}

\begin{proposition}\label{prop:TY-conv-equiv-appendix}\label{prop:TY-conv-equiv}
$\transN$  and $\transComp$ are convergence equivalent, \ie\ for all $\LETRECEXPR$-expressions $s$:  
$s\maycon_{\NAME} \iff \transN(s) \maycon_{\LCC}$ (
$s\maycon_{\NAME} \iff \transComp(s) \maycon_{\LCC}$, resp.).
\end{proposition}
\begin{proof}
We first prove convergence equivalence of $\transN$:
Suppose $s\maycon_{\NAME}$. Let $s \xrightarrow{\NAME,k} s_1$ where $s_1$ is a WHNF.
We show that there exists an $\LLAZYCC$-WHNF $s_2$ such that $\transN(s) \annRedM{\LCC} s_2$ 
by induction on $k$. The base case follows from Lemma~\ref{lem:T-WHNF}. 
The induction step follows by applying a diagram from Lemma~\ref{lem:T-diagrams}
and then using the induction hypothesis.

For the other direction we assume that $\transN(s) \maycon_{\LCC}$, \ie\ there exists a
WHNF $s_1 \in \LLAZYCC$ \st\ $\transN(s) \xrightarrow{\LCC,k} s_1$. 
By induction on $k$ we show that there exists a
$\LNAME$-WHNF $s_2$ such that $s \annRedM{\NAME} s_2$. The base case is covered by
Lemma~\ref{lem:T-WHNF}. The induction step uses the diagrams.
Here it is necessary to observe that the diagrams for the reductions (\rlapp), (\rlcase), and (\rlseq) cannot be applied
infinitely often without 
being interleaved with other reductions.
This holds, since let-shifting by (\rlapp), (\rlcase), and (\rlseq) moves \tletrec-symbols to the top of the expressions, and thus
 there are no infinite sequences of these reductions.

It remains to show convergence equivalence of $\transComp$: 
Let $s\maycon_{\NAME}$ then $\transN(s)\maycon_{\LCC}$, since $\transN$ is
convergence equivalent. Lemma~\ref{lemma:T-prime-sim-T} implies $\transN(s)\sim_{\LCC}\transComp(s)$ 
and thus $\transComp(s)\maycon_{\LCC}$ must hold.
For the other direction Lemma~\ref{lemma:T-prime-sim-T} shows that $\transComp(s)\maycon_{\LCC}$ implies 
$\transN(s)\maycon_{\LCC}$. Using convergence equivalence of $\transN$ yields
$s\maycon_{\NAME}$.
\end{proof}

\begin{lemma}\label{lemma:T-prime-comp}
The translation $\transComp$ is compositional, \ie\ for all expressions $s$ and all contexts $C$: $\transComp(C[s]) = \transComp(C)[\transComp(s)]$.
\end{lemma}
\begin{proof}
This easily follows by  structural induction on the definition.
\end{proof}

\begin{proposition}\label{prop:T-adequ}
 For all  $s_1,s_2 \in \LETRECEXPR$: $\transComp(s_1) \lec_{\LCC} \transComp(s_2) 
 \implies s_1 \lec_{\NAME} s_2$, \ie\ $\transComp$ is adequate.
\end{proposition}
\begin{proof}
Since $\transComp$ is convergence-equivalent (Proposition~\ref{prop:TY-conv-equiv}) 
and compositional by Lemma \ref{lemma:T-prime-comp}, 
we derive that $\transComp$ is adequate (see \cite{schmidt-schauss-niehren-schwinghammer-sabel-ifip-tcs:08} 
and Section \ref{sec:common}).
\end{proof}

\begin{lemma}\label{lemma:leq-transfers-from-name-to-lazy-for-letrec-free-expr}
 For $\tletrec$-free expressions  $s_1,s_2 \in \LAMBDAEXPR$ the following holds: 
 $s_1 \lec_{\NAME} s_2 \implies s_1 \lec_{\LCC} s_2$.
\end{lemma}
\begin{proof}
Note that the claim only makes sense since clearly $\LAMBDAEXPR \subseteq \LETRECEXPR$.
Let $s_1,s_2$ be $\tletrec$-free such that $s_1 \lec_{\NAME} s_2$.
Let $C$ be an $\LLAZYCC$-context such that $C[s_1]\maycon_{\LCC}$, \ie\ $C[s_1] \xrightarrow{\LCC,k} \lambda x.s_1'$.
By comparing the reduction strategies in $\LNAME$ and $\LLAZYCC$, we obtain
that $C[s_1] \xrightarrow{\NAME,k} \lambda x.s_2'$ (by the identical reduction sequence)
since $C[s_1]$ is $\tletrec$-free. Thus, $C[s_1]\maycon_{\NAME}$ and also $C[s_2]\maycon_{\NAME}$, \ie\ there is a normal order reduction in $\LNAME$ for $C[s_2]$ to a WHNF. 
Since $C[s_2]$ is $\tletrec$-free, we can
perform the identical reduction in $\LLAZYCC$ and obtain $C[s_2]\maycon_{\LCC}$.
\end{proof}

\noindent The language $\LLAZYCC$ is embedded into $\LNAME$ (and also $\LLR$) by $\iota(s) = s$.    

\begin{proposition}\label{proposition:s-Ns-neu}
 For all $s \in \LETRECEXPR$: $s \sim_{\NAME} \iota(\transComp(s))$. 
\end{proposition}
\begin{proof}
We first show that for all expressions $s\in\LETRECEXPR$: $s \sim_{\NAME} \iota(\transComp(s))$.
Since $\transComp$ is the identity mapping on $\tletrec$-free expressions of $\LNAME$ and $\transComp(s)$ is $\tletrec$-free,
we have
$\transComp(\iota(\transComp(s))) = \transComp(s)$. 
Hence adequacy of $\transComp$ (Proposition~\ref{prop:T-adequ}) implies $s \sim_{\NAME} \iota(\transComp(s))$.
\end{proof}

\begin{proposition}\label{prop:missing-part:N-fully-abs}
 For all $s_1,s_2 \in \LETRECEXPR$: $s_1 \lec_{\NAME} s_2 \implies \transComp(s_1) \lec_{\LCC} \transComp(s_2)$.
\end{proposition}
\begin{proof}
For this proof it is necessary to observe that $\LAMBDAEXPR \subseteq \LETRECEXPR$, thus we can  treat $\LLAZYCC$ expressions as $\LNAME$ expressions.
Let $s_1,s_2 \in \LETRECEXPR$ and $s_1 \lec_{\NAME} s_2$. 
By  Proposition~\ref{proposition:s-Ns-neu}:
$\transComp(s_1) \sim_{\NAME} s_1 \lec_{\NAME} s_2 \sim_{\NAME} \transComp(s_2)$,  thus $\transComp(s_1) \lec_{\NAME} \transComp(s_2)$.
Since $\transComp(s_1)$ and $\transComp(s_2)$ are $\tletrec$-free, we can apply Lemma~\ref{lemma:leq-transfers-from-name-to-lazy-for-letrec-free-expr}
and thus have $\transComp(s_1) \lec_{\LCC} \transComp(s_2)$.
\end{proof}

Now we put all parts together, where $(\transComp \circ W) (s)$ means $\transComp(W(s))$:
\begin{theorem}\label{theorem:N-fully-abstract}
$\transComp$ and $\transComp \circ W$ are fully-abstract, \ie\ for all expressions $s_1,s_2\in\LETRECEXPR$: \mbox{$s_1 \lec_{\LR} s_2 \iff 
{N(W(s_1))} \lec_{\LCC} {N(W(s_2))}$}.
\end{theorem}
\begin{proof}
Full-abstractness of $\transComp$ follows from Propositions~\ref{prop:T-adequ} and~\ref{prop:missing-part:N-fully-abs}.
Full-abstractness of $\transComp \circ W$ thus holds, since $W$ is fully-abstract (Corollary~\ref{cor:W-fully-abs}).
\end{proof}

Since $N$ is surjective, this and Corollary \ref{cor:N-iso} imply:
\begin{corollary}\label{cor:N-iso}
$\transComp$ and $\transComp \circ W$ are isomorphisms according to Definition \ref{def:translation-compo-etal}.\qed
\end{corollary}

The results also allow us to transfer the characterization of expressions in $\LLAZYCC$ into 
$\LLR$. With $\cBot_\LR$ we denote the set of $\LETRECEXPR$-expressions $s$ with the property that for all 
substitutions $\sigma$: if $\sigma(s)$ is closed, then $\sigma(s) \Uparrow_\LR$.

\begin{proposition}\label{prop-classification-lr}
Let $s$ be a closed $\LETRECEXPR$-expression. 
Then there are three cases: $s \sim \Omega$, $s \sim_{\LR} \lambda x.s'$ for some $s'$,  
$s \sim_{\LR} c~s_1 \ldots s_n$
for some terms $s_1,\ldots,s_n$ and constructor $c$. Moreover, the three cases are disjoint.   
For two closed $\LETRECEXPR$-expressions $s,t$ with $s \lec_{\LR} t$: Either $s \sim_{\LR} \Omega$, or $s \sim_{\LR} c~s_1 \ldots s_n$, 
$t \sim c~t_1 \ldots t_n$ and $s_i \lec_{\LR} t_i$ for all $i$ for some terms 
  $s_1,\ldots,s_n,t_1,\ldots,t_n$ and constructor $c$,
or  $s \simc_\LR \lambda x.s'$ and $t  \simc_\LR \lambda x.t'$ for some expressions $s',t'$ with $s' \lec_{\LR} t'$, or  
$s \simc_\LR \lambda x.s'$ and $t  \simc_\LR  c~t_1 \ldots t_n$ 
for some term $s' \in \cBot_\LR$, expressions $t_1,\ldots,t_n$ and constructor $c$. 
\end{proposition} 
\begin{proof}
This follows by Proposition~\ref{prop-classification-lcc} and since $N \circ W$ is surjective, compositional and fully abstract, and the identity on constructors.
\end{proof}

%%%%%%%%%%%%%%%%%%%%%%%%%%%%%%%%%%%%%%%%%%%%%%%%%%%%%%%%%%%%%%%%%%%%%%%%%%%%%%
%% SECTION 7: ON SIMILARITY IN LR
\section{\texorpdfstring{On Similarity  in  $\LLR$}{On Similarity in LR}}\label{sec-simulations}
In this section we will explain co-inductive and inductive (bi)similarity for $\LLR$.
Our results of the previous sections then enable us to show that these bisimilarities coincide  
with contextual equivalence in $\LLR$.

\subsection{\texorpdfstring{Overview of soundness and completeness proofs for similarities in $\LLR$}{Overview of soundness and completeness proofs for similarities in LR}}
Before we give details of the proof for lifting soundness and completeness of similarities from $\LLAZYCC$ to $\LLR$, 
we show an outline of the proof in \FIGURE~\ref{fig:proof-structure}.
The diagram shows fully abstract translations between the calculi $\LLR, \LNAME$, and $\LLAZYCC$ defined and studied in 
Sections~\ref{sec-translation-NEED-NAME} and~\ref{sec:NAME-to-LAZY}, where Corollary~\ref{cor:W-fully-abs} 
and Theorem~\ref{theorem:N-fully-abstract} show full abstractness for $W$ and $N$, respectively. 
These fully-abstract translations that are also surjective, and the identity on letrec-free expressions, allow us to 
prove that $s_1 \lec_{\LR} s_2 \iff N(W(s_1)) \lec_{\LCC} N(W(s_2))$.  
By Theorem~\ref{thm:llc-Q-equivalence} in  $\LLAZYCC$,
this is equivalent to  $N(W(s_1)) \leb_{\LCC}^o N(W(s_2))$. 
 The proof is completed by using the translations by transferring the equations back and forth between $\LLR$ and $\LLAZYCC$
in this section  in order to finally show that $s_1 \lec_{\LR} s_2 \iff s_1 \leb_{\LR,Q_{\mathit{CE}}}^o s_2$ in Theorem \ref{thm:maintheorem}.

\begin{figure}
\begin{tikzpicture}[yscale=0.6,xscale=1]
 \node (Lneedeq) at (0.5,-3) [] {\scalebox{1}{$s_1 \lec_{\LR} s_2$}};
 \draw [-,dotted] (Lneedeq) to node [] {} node [] {} (0.5,-.3);
 \node (LneedeqLRQce) at (-1.5,0) [] {\scalebox{1}{$s_1 \lec_{{\LR,Q_{\mathit{CE}}}} s_2$}};
 \node (Lneedbeq) at (0.5,3) [] {\scalebox{1}{$s_1 \leb_{\LR,Q_{\mathit{CE}}} s_2$}};
 \draw [-,dotted] (Lneedbeq) to node [] {} node [] {} (0.5,.3);
 \node (Lnameeq) at (5,-3) [] {\scalebox{1}{$W(s_1) \lec_{\NAME} W(s_2)$}};
 \draw [-,dotted] (Lnameeq) to node [] {} node [] {} (5,-.3);
 \node (Llazyeq) at (9.8,-3) [] {\scalebox{1}{$N(W(s_1)) \lec_{\LCC} N(W(s_2))$}};
 \draw [-,dotted] (Llazyeq) to node [] {} node [] {} (10,-.3);
 \node (LlazyQbeq) at (9.8,3) [] {\scalebox{1}{$N(W(s_1)) \leb_{\LCC,Q_{\mathit{CE}}} N(W(s_2))$}};   
 \draw [-,dotted] (LlazyQbeq) to node [] {} node [] {} (10,.3);
 \node (Lneedcircle) at (0.5,0) [line width=1pt,shape=ellipse,draw,fill=black!15!white] {\rule{0mm}{8mm}\rule{8mm}{0mm}};
 \node (Lneedtext) at (0.5,0) [] {\scalebox{1}{$\LLR$}};
 \node (Lnamecircle) at (5,0) [line width=1pt,shape=ellipse,draw,fill=black!10!white] {\rule{0mm}{8mm}\rule{8mm}{0mm}};
 \node (Lnametext) at (5,0) [] {\scalebox{1}{$\LNAME$}};
 \node (Llazycircle) at (9.8,0) [line width=1pt,shape=ellipse,draw,fill=black!10!white] {\rule{0mm}{8mm}\rule{8mm}{0mm}};
 \node (Llazytext) at (9.8,0) [] {\scalebox{1}{$\LLCC$}};
 \draw[->,line width=1pt] (Lneedcircle) to node [above] {$W$} node [] {} (Lnamecircle);
 \draw[->,line width=1pt,bend left =30] (Lneedcircle) to node [above] {$N \circ W$} node [] {} (Llazycircle);
 \draw[->,line width=1pt] (Lnamecircle) to node [above] {$N$} node [] {} (Llazycircle);
 \draw[<->,double,double distance=1.2pt,line width=.5pt] (Lneedeq) to node [above=3pt] {Cor}  node [below=3pt] {\ref{cor:W-fully-abs}}(Lnameeq);
 \draw[<->,double,double distance=1.2pt,line width=.5pt] (Lnameeq) to node [above=3pt] {Thm} node [below=3pt] {\ref{theorem:N-fully-abstract}} (Llazyeq);
 \draw[<->,double,double distance=1.2pt,line width=.5pt, bend right=40] (Llazyeq) to node [right] {\parbox{1cm}{Thm\\\ref{theorem:bisim-alternatives}}} node [] {} (LlazyQbeq);
 \draw[<->,double,double distance=1.2pt,line width=.5pt, bend left=20] (Lneedeq) to node [left] {\parbox{1cm}{Prop\\\ref{prop:finite-simulation-LLR}}} node [] {} (LneedeqLRQce);
 \draw[<->,double,double distance=1.2pt,line width=.5pt, bend left=20] (LneedeqLRQce) to node [left] {\parbox{1cm}{Thm\\\ref{theo:closed-need-bisim-is-csim}}} node [] {} (Lneedbeq);
 \draw[<->,double,double distance=1.2pt,line width=.5pt] (LlazyQbeq) to node [above] {see Proof of Prop.  \ref{prop:finite-simulation-LLR}} node [] {} (Lneedbeq);
\end{tikzpicture}
\caption{The structure of the reasoning for the similarities in $\LLR$ for closed expressions.}\label{fig:proof-structure}
\end{figure}

\subsection{\texorpdfstring{Similarity in $\LLR$}{Similarity in LR}}\label{subsec:sim-lr-def}
The definition of $\LLR$-WHNFs implies that  they are of the form $R[v]$, where $v$ is either 
an abstraction $\lambda x.s$ or a constructor application $(c~s_1~\ldots~s_{\ari(c_i)})$,
and where $R$ is an {\em $\LLR$-AWHNF-context} according 
to the grammar $R ::= [\cdot]~|~\tletrx{\iEnv}{[\cdot]}$ if $v$ is an abstraction, 
 and   $R$ is an {\em $\LLR$-CWHNF-context} according 
to the grammar $R ::= [\cdot]~|~\tletrx{\iEnv}{[\cdot]}~|~ \tletrx{x_1 = [\cdot], \bchainXInd{2}{m},\iEnv}{x_{m}}$
if $v$ is a constructor application. Note that  $\LLR$-AWHNF-contexts and $\LLR$-CWHNF-contexts are
special $\LLR$-reduction contexts, also called {\em $\LLR$-WHNF-contexts}.
 
First we show that finite simulation (see \cite{schmidt-schauss-machkasova-rta:08}) is correct for $\LLR$:

\begin{definition}
Let $\lec_{\LR,Q_{\mathit{CE}}}$ be defined for $\LLR$ as instantiating 
the relation $\leq_{\cal Q}$ in Definition~\ref{def:Q-simpl-preorder} with the closed subcalculus of the calculus $\LLR$ and the set ${\cal Q}$
with $Q_{\mathit{CE}}$ from Definition~\ref{def:Q-Omega-expr-contexts}.

The relation $\leb_{\LR,Q_{\mathit{CE}}}$ is ${\cal Q}$-similarity
(Definition~\ref{def:Q-gfp-preorder}) instantiated for the calculus $\LLR$ with the set of contexts $Q_{\mathit{CE}}$ (Definition~\ref{def:Q-Omega-expr-contexts}).
Its open extension is denoted with $\leb_{\LR,Q_{\mathit{CE}}}^o$.
\end{definition}

\begin{proposition} \label{prop:finite-simulation-LLR}
Let  $s_1,s_2$ be closed $\LETRECEXPR$-expressions. Then $s_1 \lec_\LR  s_2$ iff $s_1 \lec_{\LR,Q_{\mathit{CE}}} s_2$.
\end{proposition}
\begin{proof} 
The $\Rightarrow$ direction is trivial.
We show $\Leftarrow$, the nontrivial part:
Assume that the inequation $s_1 \lec_{\LR,Q_{\mathit{CE}}} s_2$ holds.
Then $\transComp(W(s_1)) \lec_{\LCC,Q_{\mathit{CE}}} \transComp(W(s_2))$, since for every $n \ge 0$ 
and context $Q = Q_n(\ldots(Q_2(Q_1[ ])\ldots))$
with $Q_i \in Q_{\mathit{CE}}$, we have $\transComp(W(Q)) = Q$, and also $Q(s_i)\downarrow_\LR \iff Q(s_i)\downarrow_\LCC$, since $\transComp\circ W$
is convergence-equivalent and compositional, and the identity on $\tletrec$-free expressions. 
Now Theorem \ref{theorem:bisim-alternatives} shows   $\transComp(W(s_1)) \lec_{\LCC} \transComp(W(s_2))$, 
and then  Theorem \ref{theorem:N-fully-abstract} shows   $s_1 \lec_\LR s_2$.\qedhere
 \end{proof}
\noindent The following  lemma is helpful in applying Theorem \ref{thm:conv-admissibile-implies}.

\begin{lemma}\label{lem:w-condition-need}
 The closed part of the calculus $\LLR$  is convergence-admissible: 
 For all contexts $Q \in Q_{\mathit{CE}}$, and closed $\LLR$-WHNFs $w$:
  $Q(s)\maycon_\LR w$ iff $\exists v:  s \maycon_\LR v$ and $Q(v) \maycon_\LR w$.
\end{lemma}
\begin{proof}
 ``$\Rightarrow$'': First assume $Q$ is of the form $([\cdot]~r)$ for closed $r$.  
Let $(s~r) \maycon_\LR w$. There are two cases, which can be verified by induction on the length $k$ of a reduction sequence $(s~r)\xrightarrow{\LR,k} w$:
 $(s~r) \xrightarrow{\LR,*} ((\lambda x.s')~r) \xrightarrow{\LR,*} w$,
 where $s \xrightarrow{\LR,*}  (\lambda x.s')$, and the claim holds.
The other case is $(s~r) \xrightarrow{\LR,*}  \tletrx{\iEnv}{((\lambda
x.s')~r)} \xrightarrow{\LR,*} w$, where $s \xrightarrow{\LR,*} 
\tletrx{\iEnv}{(\lambda x.s')}$. In this case $(\tletrx{\iEnv}{(\lambda
x.s')}~r) \xrightarrow{\LR,(lapp)} \tletrx{\iEnv}{((\lambda x.s')~r)}
\xrightarrow{\LR,*} w$, and thus the claim is proven. 
The other cases where $Q$ is of the form $(\tcase_T~[\cdot]~\tof~\ldots)$ can be proven similarly.
\\
The ``$\Leftarrow$''-direction can be proven using induction on the length of reduction sequences. 
\end{proof}

\begin{lemma}\label{lemma:LR-open-extension-ok}
In $\LLR$, the equation $(\lec_\LR^c)^o~= ~\lec_\LR$  holds.
\end{lemma}
\begin{proof}
If $s,t$ are (open) $\LETRECEXPR$-expressions with $s \lec_\LR t$, then $(\lambda x_1.\ldots x_n.s) ~s_1 \ldots s_n \lec_\LR^c (\lambda x_1.\ldots x_n.t)~ s_1 \ldots s_n$ 
 for closed expressions $s_i$,
and then by correctness of reduction in $\LLR$, $\sigma(s) \lec_\LR^c \sigma(t)$, and hence $\lec_\LR ~\subseteq {(\lec_\LR^c)^o}$.

If for all closing $\LETRECEXPR$-substitutions $\sigma$: $\sigma(s) \lec_\LR^c \sigma(t)$, then using the fully abstract translations $\transComp \circ W$, we obtain
$\transComp \circ W(\sigma)(\transComp \circ W(s)) \lec_\LCC^c \transComp \circ W(\sigma)(\transComp \circ W(t))$,
hence $\transComp \circ W(s) \lec_\LCC^c (\transComp \circ W(t))$ by 
Theorem \ref{theorem:bisim-alternatives}. Again using fully abstractness of $\transComp \circ W$, we obtain $s \lec_\LR t$.
\end{proof}

\begin{theorem}\label{theo:closed-need-bisim-is-csim}
In $\LLR$, for closed $\LETRECEXPR$-expressions $s$ and $t$ the statements 
$s \leb_{\LR,Q_{\mathit{CE}}} t$, 
$s\lec_{\LR,Q_{\mathit{CE}}} t$ and 
$s \lec_\LR t$ are all equivalent.
\end{theorem}
\begin{proof}
Lemma~\ref{lem:w-condition-need} shows that Theorem \ref{thm:conv-admissibile-implies} is applicable
for the testing contexts from  $Q_{\mathit{CE}}$, i.e. $\leb_{\LR,Q_{\mathit{CE}}} = {\lec_{\LR,Q_{\mathit{CE}}}}$ and
Proposition~\ref{prop:finite-simulation-LLR} shows ${\lec_{\LR,Q_{\mathit{CE}}}} = {\lec_{\LR}^c}$
\end{proof}

For open  $\LETRECEXPR$-expressions, we can lift the properties from $\LLAZYCC$, which also follows from full abstraction of  $\transComp \circ W$ and  
from Lemma \ref{prop:open-closed-case}.

The results above imply the following theorem: 

\begin{maintheorem}\label{thm:maintheorem}
 ${\lec_{\LR}} = {\leb_{\LR,Q_{\mathit{CE}}}^o}$.
\end{maintheorem}
\begin{proof}
  Theorem \ref{theo:closed-need-bisim-is-csim} shows $\leb_{\LR,Q_{\mathit{CE}}}~= ~ \lec_{\LR,Q_{\mathit{CE}}} ~=~ \lec_\LR^c$, 
  hence $\leb_{\LR,Q_{\mathit{CE}}}^o ~= ~ (\lec_\LR^c)^o$. Then  
  Lemma \ref{lemma:LR-open-extension-ok} shows $(\lec_\LR^c)^o~=~ {\lec_{\LR}}~=~  {\leb_{\LR,Q_{\mathit{CE}}}^o}$.
\end{proof}

The Main Theorem \ref{thm:maintheorem} implies that our embedding of $\LLAZYCC$ into  the call-by-need letrec calculus $\LLR$ (modulo   $\sim$)   
  is isomorphic \wrt\ the corresponding term models, \ie\:
\begin{theorem}\label{thm:isomorphism}
The identical embedding $\iota: \LAMBDAEXPR \to \LETRECEXPR$ is an isomorphism according to Definition \ref{def:translation-compo-etal}. \qed
\end{theorem}

\begin{remark}
Consider a polymorphically typed variant of $\LLR$,
say $\LLR^{\mathrm{poly}}$, and a
corresponding type-indexed contextual preorder $\le_{\LR,\mathrm{poly},\tau}$ which relates expressions of polymorphic type $\tau$ 
and where the testing contexts are restricted to well-typed contexts, \ie~for $s,t$ of type $\tau$ the inequality $s \lec_{\LR,\mathrm{poly},\tau} t$ holds
iff for all contexts $C$ such that $C[s]$ and $C[t]$ are well-typed: $C[s]\maycon_{\LR} \implies C[t]\maycon_{\LR}$.
Obviously for all expressions $s,t$ of type $\tau$ the inequality $s \lec_{\LR} t$ implies $s \lec_{\LR,\mathrm{poly},\tau} t$, since
any test (context) performed for $\lec_{\LR,\mathrm{poly},\tau}$ is also included in the tests for $\lec_{\LR}$ (there are more contexts).
Thus the main theorem implies that ${\leb_{\LR,Q_{\mathit{CE}}}^o}$ is sound \wrt~the typed preorder $\lec_{\LR,\mathrm{poly},\tau}$.
Of course completeness does not hold, and requires another definition of similarity which respects the typing. 
\end{remark}

\subsection{\texorpdfstring{Similarity up to $\simc_\LR$}{Similarity up to simcLR}}

A more comfortable tool to prove program equivalences in $\LLR$ is the following similarity definition which allows to 
 simplify intermediate expressions that are known to be equivalent.

\begin{definition}[Similarity up to $\simc_\LR$]
Let $\leb_{\LR,\sim}$ be the greatest fixpoint of the following operator $\UptoSimLR$ on closed $\LETRECEXPR$-expressions:

We define an operator $\UptoSimLR$ on binary relations $\eta$ on closed $\LLAZYCC$-expressions:\\
 $s~\UptoSimLR(\eta)~t$  iff the following holds:
\begin{enumerate}
 \item If $s \simc_\LR \lambda x.s'$ then there are two possibilities: 
     (i)  if  $t \simc_\LR (c~t_1 \ldots t_n)$ then $s' \in \cBot_\LR$, or  
     (ii) if  $t \simc_\LR \lambda x.t'$ then for all closed $r:   ((\lambda x.s')~r)~\eta~((\lambda x.t')~r)$;
 \item If $s \simc_\LR (c~s_1~\ldots~s_n)$ then $t \simc_\LR (c~t_1 \ldots t_n)$ and $s_i~\eta~t_i$ for all $i$.
\end{enumerate}
\end{definition}

\begin{lemma}\label{lem:simc-implies-simbupto} ${\lec_\LR^c} \subseteq{\leb_{\LR,\sim}}$
\end{lemma}
\begin{proof}
We show that $\eta := {\le^c_\LR}$ is $F_{\LR,\sim}$-dense, i.e. $\eta \subseteq F_{\LR,\sim}(\eta)$.

Let $s~\eta~t$ and $s \simc_\LR \lambda x.s'$.
Since $s \le^c_\LR t$ either $t \simc_\LR \lambda x.t'$ or
$t \simc_\LR c~t_1~\ldots~t_n$ and $s' \in \cBot_\LR$. For the latter case we are finished.
For the former case  we have $\lambda x.t' \simc_\LR^c t$.
Since $\lec_\LR^c$ is a precongruence, this implies $((\lambda x.s')~r) \lec_\LR ((\lambda x.t')~r)$ for all closed $\LETRECEXPR$-expressions
$r$. Thus we conclude $s~F_{\LR,\sim}(\eta)~t$.

Now let $s~\eta~t$ and $s \simc_\LR^c c~s_1~\ldots~s_n$.
Then $t \simc_\LR^c (c~t_1~\ldots~t_n)$ by Proposition~\ref{prop-classification-lr}.
The contexts $C_i := (\tcase~[]~\tof \ldots (c~x_1\ldots~x_n \to x_i) \ldots)$ where all other
right hand sides of \tcase-alternatives are $\bot$, show that also $s_i \lec_\LR t_i$ must hold,
since otherwise $s \lec_\LR^c t$ cannot hold.
Thus also in this case $s~F_{\LR,\sim}(\eta)~t$ holds.
\end{proof}

\begin{lemma}\label{lemma:upto-from-lr-into-lc} $N(W(\leb_{\LR,\sim}))  \subseteq {\leb_{\LCC,\sim}}$.
\end{lemma}
\begin{proof}
 We show that $\eta := \{(N(W(s)), N(W(t))) \mid s \leb_{\LR,\sim} t\}$ is $\UptoSim$-dense (see Definition \ref{def:lcc-sim-upto}), i.e. $\eta \subseteq \UptoSim(\eta)$.
Let $s \leb_{\LR,\sim} t$ for closed $s,t$.      %%  $(s,t) \in \eta$.
If $N(W(s)) \simc_\LCC \lambda x.s'$, then also $s \simc_\LR \lambda x.s'$.
Now there are two cases: If $t \simc_\LR (c~t_1~\ldots~t_n)$ then $s'\in\cBot_\LR$ must hold.
Then also $s' \in \cBot$ and $N(W(t))\simc_\LCC (c~t_1~\ldots~t_n)$ and we are finished.
If $t \simc_\LR \lambda x.t'$ then for all closed $\LETRECEXPR$-expressions $r$:
 $(\lambda x.s')~r \leb_{\LR,\sim} (\lambda x.t')~r$ (by unfolding the fixpoint equation for $\UptoSimLR$).
Since $N\circ W$ is surjective, compositional and fully abstract, this also shows
$N(W(\lambda x.s'))~r~\eta~N(W(\lambda x.t'))~r$ for all $\LLAZYCC$-expressions $r$.

If $N(W(s)) \simc_\LCC (c~s_1 \ldots s_n)$, then also $s \simc_\LR (c~s_1 \ldots s_n)$.
Now $s \leb_{\LR,\sim} t$ shows that $t \simc_\LR (c~t_1~\ldots~t_n)$ such that for all $i$: $s_i \leb_{\LR,\sim} t_i$.
Hence $(s_i,t_i) \in \eta$ and also $N(W(t)) \simc_\LCC (c~t_1~\ldots~t_n)$, since $N\circ W$ is fully abstract.
\end{proof}

\begin{theorem}\label{theo:le_c-eq-bisim-uptop}
 ${\lec_\LR} = {\leb_{\LR,\sim}^o}$
\end{theorem}
\begin{proof}
For the closed relations, one direction of the equation 
${\leb_{\LR,\sim}} = {\lec_{\LR}^c}$ is  Lemma~\ref{lem:simc-implies-simbupto},
the other direction follows from Lemma~\ref{lemma:upto-from-lr-into-lc}, 
since $s \leb_{\LR,\sim} t$ implies $N(W(s)) \leb_{\LCC,\sim} N(W(t))$
which in turn implies $N(W(s)) \lec_\LCC^c N(W(t))$ and finally, 
full-abstraction of $N\circ W$ shows $s \lec_\LR^c t$.

For the open extension the claimed equality holds, since $s \lec_\LR t$ iff 
$\sigma(s) \lec_\LR \sigma(t)$ for all closing substitutions $\sigma$: 
This holds, since for $\sigma = \{x_1 \mapsto s_1,\ldots, x_n \mapsto s_n\}$ the equation
$\sigma(s) \simc_\LR \tletrec~x_1 = s_1,\ldots,x_n=s_n~\tin~s$ holds 
by correctness of the general copy rule (gcp) (Proposition~\ref{prop-gcp-correct}) and
of garbage collection (gc) (Theorem~\ref{theo:lr-trans-corr}).
\end{proof}
\noindent We demonstrate the use of similarity up to $\simc_\LR$ in the following example:

\begin{example}
As an example we prove the list law $R[\mathit{map}~(\lambda x.\ttrue)~(\mathit{repeat}~u)] \simc_\LR R'[(\mathit{repeat}~\ttrue)]$ 
where $u$ is a closed expression and $R'$, $R$, resp. contains the definition of $\mathit{repeat}$, or $\mathit{repeat}$ and $\mathit{map}$, 
resp., i.e. the corresponding $\LETRECEXPR$-expressions are:
$$\begin{array}{l@{~}l}
s := &\tletrec~\\
&~~~~\mathit{repeat} = \lambda x.\tCons~x~(\mathit{repeat}~x),
\\
&~~~~\mathit{map}    = \lambda f.\lambda xs.\tcase_{\mathit{List}}~xs~\tof~(\tnil \to \tnil)~(\tCons~y~ys \to \tCons~(f~y)~(\mathit{map}~f~ys))
\\
   &   \tin~\mathit{map}~(\lambda x.\ttrue)~(\mathit{repeat}~u)
\\[1.2ex]
t := &\tletrec~
\\
&~~~~\mathit{repeat} = \lambda x.\tCons~x~(\mathit{repeat}~x),
\\
      &\tin~\mathit{repeat}~\ttrue 
\end{array}
$$
Let $\eta := \{(t,s), (s,t)\} \cup \{(\ttrue,\ttrue)\}$.
We show that $\eta \subseteq F_{\LR,\sim}(\eta)$ which implies 
$s \leb_{\LR,\sim} t$ as well as $t \leb_{\LR,\sim} s$ and thus
by Theorem~\ref{theo:le_c-eq-bisim-uptop} also $s \simc_\LR t$.

Evaluating $s$ and $t$ in normal order first shows: 
$s \simc_\LR v_1, t \simc_\LR v_2$ with
$$\begin{array}{l@{~}l@{}}
v_1 = &\tletrec~
\\
        &~~~~\mathit{repeat} = \lambda x.\tCons~x~(\mathit{repeat}~ x),
\\
        &~~~~\mathit{map}    = \lambda f.\lambda xs.\tcase_{\mathit{List}}~xs~\tof~(\tnil \to \tnil)~(\tCons~y~ys \to \tCons~(f~y)~(\mathit{map}~f~ys))
\\
        &~~~~f_1 = (\lambda x.\ttrue),x_1 = t, xs_1 = \tCons~x_1'~x_2', x_1' = x_1, x_2' = (\mathit{repeat}~t), y_1 = x_1', ys_1 = x_2'
\\
        &\tin~\tCons~(f_1~y_1)~(\mathit{map}~f_1~ys_1)
\\[1.2ex]
v_2 = &\tletrec~
\\
&~~~~\mathit{repeat} = \lambda x.\tCons~x~(\mathit{repeat}~x), 
\\
&~~~~x_1=\ttrue
\\
&\tin~\tCons~x_1~(\mathit{repeat}~x_1)
\end{array}$$
Using correctness of garbage collection, copying of bindings (gcp), 
shifting constructors over \tletrec, and the other correct reduction rules
(see Theorem~\ref{theo:lr-trans-corr} and Proposition~\ref{prop-gcp-correct}),
we can simplify as follows: 
$v_1 \simc_\LR \tCons~\ttrue~s$ and $v_2 \simc_\LR \tCons~\ttrue~t$.
Now the proof is finished, since obviously $\ttrue~\eta~\ttrue$ and
$s~\eta~t$, $t~\eta~s$.
\end{example}

%%%%%%%%%%%%%%%%%%%%%%%%%%%%%%%%%%%%%%%%%%%%%%%%%%%%%%%%%%%%%%%%%%%%%%%%%%%%%%
%% SECTION 8: CONCLUSION
\section{Conclusion}\label{sec:conclusion}
In this paper we have shown that co-inductive applicative bisimilarity, in the
style of Howe, and also the inductive variant, is equivalent to contextual 
equivalence in a deterministic call-by-need calculus with letrec, case, 
data constructors, and {\tt seq} which models the (untyped) core language of 
Haskell. This also shows soundness of untyped applicative bisimilarity for the
polymorphically typed variant of $\LLR$.  As a further work one may try to
establish a coincidence of the typed applicative bisimilarity and contextual
equivalence for a polymorphically typed core language of Haskell.
%% END CONCLUSION
%%%%%%%%%%%%%%%%%%%%%%%%%%%%%%%%%%%%%%%%%%%%%%%%%%%%%%%%%%%%%%%%%%%%%%%%%%%%%%

\section*{Acknowledgements}
The authors thank the anonymous reviewers for their valuable comments.
%%
%%
%%%%%%%%%%%%%%%%%%%%%%%%%%%%%%%%%%%%%%%%%%%%%%%%%%%%%%%%%%%%%%%%%%%%%%%%%%%%%%
%% BIBLIOGRAPHY
% \bibliographystyle{alpha}
% \bibliography{bibfile}
\newcommand{\etalchar}[1]{$^{#1}$}

%%
%%%%%%%%%%%%%%%%%%%%%%%%%%%%%%%%%%%%%%%%%%%%%%%%%%%%%%%%%%%%%%%%%%%%%%%%%%%%%%
\end{document}